\documentclass[12pt]{amsart}
\usepackage{graphicx}
\usepackage{geometry}
\usepackage{amssymb}
\usepackage{setspace}
\usepackage[dvipsnames]{pstricks}
\usepackage{pst-node}
\usepackage{pst-slpe}
\usepackage{amsthm}
\usepackage{tipa}
\usepackage{protosem}
\usepackage{hieroglf}
\usepackage{subcaption}
\usepackage{graphicx}
\usepackage{amsmath}
\usepackage{float}
\usepackage{natbib}

\usepackage{pst-tree}
\usepackage{amsmath}
\DeclareMathOperator*{\argmax}{argmax}

\usepackage{hyperref}
\hypersetup{
    colorlinks=true,
    linkcolor=blue,
    filecolor=magenta,      
    citecolor=blue,
    pdftitle={Overleaf Example},
    pdfpagemode=FullScreen,
}

\usepackage{cleveref}

\usepackage{amsfonts}
\usepackage{amsmath}
\usepackage{amssymb}
\usepackage{graphicx}
\usepackage{appendix}
\usepackage{multirow}
\usepackage{ulem}
\usepackage{setspace}
\usepackage[dvipsnames]{pstricks}
\usepackage{pst-node}
\usepackage{pst-slpe}
\usepackage{amsthm}
\usepackage{lmodern}
\usepackage[T1]{fontenc}
\usepackage{mathtools}

\usepackage{booktabs}
\usepackage{amssymb}% http://ctan.org/pkg/amssymb
\usepackage{pifont}% http://ctan.org/pkg/pifont

\usepackage{natbib}

 \theoremstyle{plain}
  
  \theoremstyle{plain}
  \theoremstyle{definition}
  \newtheorem{defn}{\protect\definitionname}
  \theoremstyle{plain}
  
  \theoremstyle{plain}
  
  \theoremstyle{definition}
  
  \theoremstyle{plain}
  
  \theoremstyle{plain}
  \newtheorem{lemma}{Lemma}
  \theoremstyle{plain}
  
  \theoremstyle{plain}
  \newtheorem{proposition}{\protect\propositionname}
\theoremstyle{remark}
\newtheorem{rem}{Remark}
  \theoremstyle{definition}
  \newtheorem{example}{Example}
  \newtheorem{corollary}{Corollary}
  \theoremstyle{definition}
  \newtheorem{assump}{Axiom}
  
  \newtheoremstyle{emptyplain}
    {}          % default space above
    {}          % default space below
    {\itshape}  % body font
    {}          % no indent
    {\bfseries} % theorem head font
    {.}         % punctuation after theorem head
    { }         % space after theorem head (normal interword space)
    {#3}        % using "\thmnote{#3}" is redundant, isn't it?
\theoremstyle{emptyplain}

  \providecommand{\axiomname}{Axiom}
  \providecommand{\conjecturename}{Conjecture}
  \providecommand{\definitionname}{Definition}
  \providecommand{\corollaryname}{Corollary}
  \providecommand{\theoremname}{Theorem}
\providecommand{\propositionname}{Proposition}
\usepackage{xcolor,cancel}

\newcommand{\Wecon}{W}
\newcommand{\Wpsych}{W_{\mathrm{psych}}}

\usepackage{empheq}
\usepackage[most]{tcolorbox}

\newtcbox{\mymath}[1][]{%
    nobeforeafter, math upper, tcbox raise base,
    enhanced, colframe=blue!30!black,
    colback=blue!30, boxrule=1pt,
    #1}

\begin{document}

\title[]{Measuring Choice Difficulty$^\dag$}

% brainstorming:
% POSSIBLE TITLE: MEASURING CHOICE DIFFICULTY
% POSSIBLE TITLE: MEASURING CHOICE UNDERSTANDING

\author[]{Christopher P. Chambers$^\ast$}
\author[]{Yusufcan Masatlioglu$^\S$}
\author[]{Paulo Natenzon$^\sharp$}
\author[]{Collin Raymond$^\ddagger$}
\thanks{}
\thanks{$^\ast$ Department of Economics, Georgetown University, ICC 580  37th and O Streets NW, Washington DC 20057. E-mail: \texttt{Christopher.Chambers@georgetown.edu}.}
\thanks{$^\S$ University of Maryland, 3147E Tydings Hall, 7343 Preinkert Dr.,  College Park, MD 20742. E-mail: \texttt{yusufcan@umd.edu}}
\thanks{$^{\sharp}$ WashU Olin Business School, 1 Brookings Drive, St.\ Louis, MO 63130, 
 \texttt{pnatenzon@wustl.edu}.}
\thanks{$^{\ddagger}$ Johnson School of Management, Cornell University, Ithaca, NY 14853, \texttt{collinbraymond@gmail.com}.}
\date{March 2026.  We thank helpful comments from audience members at Gerzensee, Stanford Institute for Theoretical Economics, SOUR BEER,  and the Paris School of Economics.}

 \begin{abstract}

We provide a theoretical framework to understand how widely used measures of choice difficulty relate. In a binary-option Bayesian expected‑utility framework, we show that three measures of difficulty, (i) understanding (ex‑ante value), (ii) choice randomness, and (iii) confidence that the chosen option is ex post correct, are, in general, unrelated, and that this result extends to other potential measures like attenuation. We provide intuitive sufficient conditions which align the orders, using both restrictions on Blackwell experiments that capture well known classes (such as logit) and restrictions on payoffs and demonstrate that in psychophysical tasks that pay only for correctness, confidence coincides with understanding.   We show willingness‑to‑accept to switch, when measured in utils, is equivalent to understanding. Our results suggest caution in interpreting measures of choice difficulty  as well as the degree of portability between economics and psychophysics experiments.  

  \end{abstract}

\maketitle

% TO DO COLLIN/PAULO chat feb 11 2026:
% (1) In interpreting Proposition 3, to throw a bone at the confidence measuring experimentalists, we can say confindence is a better measure to use of understanding than randomness of choice because it requires weaker conditions, just shift alignment but doesn't require indicativeness of the experiments.
% (2) Gaussian signals and symmetric Gaussian prior. What does correlation of signals do?

\section{Introduction}

A growing empirical literature in experimental economics looks to document how difficult a choice is for individuals—whether due to features of the decision problem (e.g., complexity) or to characteristics of the decision-maker (e.g., sophistication). This set of studies draws on a rich literature in the social sciences on stochastic choice (e.g., \cite{fechner1860elemente}) as well as more recent innovations in confidence and meta-cognition (e.g., \cite{fleming2014measure}).
Within economics, the goal is typically to relate notions of how well a decision-maker understands a problem to various observed behaviors.  

The most traditional piece of data which is often used is the degree of randomness in choice (related to, but not the same as accuracy). More recently, experiments have looked at meta-cognition such as cognitive uncertainty (\cite{enke2023cognitive}), as well as other kinds of behavior, such as attenuation (\cite{enke2024behavioral}).  These measures, although conceptually related, are not developed from a single, general conceptual framework.  Many of the intuitions have been developed for specific kinds of measurements using functional form assumptions; e.g., assumptions about Gaussian priors and signals in a Bayesian Expected Utility framework (e.g., \cite{enke2023cognitive}), or generalized Fechnerian choice (e.g., \cite{shubatt2024tradeoffs}).  Similarly, in the broader literature studying meta-cognition, most results are derived under assumptions of normality (e.g., \cite{fleming2014measure}).  This leaves limited understanding of the  general theoretical underpinnings of these behaviors as well as unclear relations between concepts.  

This paper develops a unifying framework that connects various notions of behavior used to measure choice difficulty with the extent of a decision maker's (DM’s) understanding of the decision problem. Our analysis is conducted within Bayesian Expected Utility theory (see \cite{strzalecki2025stochastic} for a textbook treatment) and relates each behavioral notion to an ordering over Blackwell experiments.\footnote{See also \cite{rehbeck2019revealed} for work in this environment, as well as the closely related BAUP environment in  \cite{safonov2017random,safonov2022random}.} Prior to choice, DMs learn about the utility of the available options through information structures (``experiments''). Different choice environments ---such as menus of varying complexity--- induce different experiments. Each measure of ``difficulty'' corresponds to a distinct partial order over experiments, derived from the value of information, the degree of choice randomness, or reported confidence.

Our primary contribution is conceptual: we offer general foundations for widely used measures of choice difficulty beyond parametric normal-signal microfoundations, and clarify precisely how they are related to one another.  Focusing on the widely used setting where there are only binary options, our core results study three key orderings—(i) the degree of ``understanding'' of a decision problem, which we measure as the value of the decision problem under an information structure (in other words, how well the decision maker will choose, as measured by their expected utility), (ii) the degree of randomness in choice, and (iii) expressed confidence that the chosen option is actually the option with higher realized utility. Our interpretation is that the experiment itself is unobserved; we observe repeated, de novo choices without cross-trial information spillovers.\footnote{It is particularly important that the data comes from a single dgp.  The fact that we focus on on ``one-shot'' learning, is, given the assumption of exogenous learning, essentially without loss,  as repeated learning from a single Blackwell experiment can be captured, in a reduced form way, as a single draw from a different Blackwell experiment.} 

We have five main results.  The first is negative: in general, these three orderings are unrelated (ranking experiments in any two orderings does not imply the ranking in the third). The next three are positive.  In the second result, we deliver an intuitive sufficient condition on the set of potential Blackwell experiments under which the three rankings co-move, and show how these conditions relate to well-known classes of experiments.  As a corollary, we show that a weakening of our sufficient condition allows us to infer payoffs from randomness. Third, we also provide a distinct condition on the set of payoff functions (which correspond to payoffs often used in psychophysical experiments)  which guarantees that confidence and understanding co-move.  Fourth, we show that a well-known measure, the willingness to accept to switch a choice (widely used, e.g., in endowment effect experiments such as \cite{Knetsch84}) is, under the assumption that it can be measured in utils, equivalent to the understanding of the DM.  Last, we extend our initial negative results, showing that measuring attenuation cannot overcome our negative results.  

Taken as a whole, our negative results suggest that we should interpret results from existing experiments with some caution, and that insights from psychophysics are not necessarily immediately portable to economics frameworks.  The primary negative result, at its core, is because the ability to understand a problem and choose well from it uses a utility weighted measures (e.g., expected payoffs from the decision), which depends on cardinal attributes of states (i.e., the actual utility values of each option in each state).  On the other hand, measures such as randomness and confidence reflect ordinal notions, such as how accurate a choice is (i.e., does the choice match the option that has higher utility).  On the other hand, our positive results suggest reasonable assumptions under which we can clearly interpret measures such as randomness and confidence in an economic context, and provide alternative measures that directly measure the value of the decision.  

Our results also highlight that depending on the kinds of assumptions we make about the decision-problem, different versions of the same concept will need to be used to recover understanding. For example, under our sufficient condition on the Blackwell experiment, we should measure option by option confidence, averaged across states, while randomness should be measured state-by-state.  In contrast, under our sufficient conditions for payoffs, we should measure confidence averaging across options, and randomness does not help us at all.  Thus, researchers need to think carefully about not just what is being measured, but at what level of aggregation.  

Conceptually, the framework connects canonical statistical decision theory (Blackwell experiments and notions of informativeness) with empirical measures used in economics labs to measure choice difficulty. We can rationalize why noisy choice, low confidence, and low value of information may sometimes move together (when induced experiments vary within a restricted family) and sometimes diverge (when complexity or beliefs shift the experiment in ways that affect these metrics differentially). 

Section \ref{sec:main} begins by providing formal definitions of our choice environment as well as our three primary orderings.  The framework is simple: the DM faces a binary choice set.  They do not know the utility of either option, but have a prior belief over possible utilities.  They observe a signal from a Blackwell experiment, update their beliefs and choose the option with the highest expected utility. We are interested in the comparative statics of this problem as the Blackwell experiment varies.  In other words, the only differences in choice across scenarios must come from differences in the experiment.  A natural question is whether it is plausible to vary the Blackwell experiment (which we conceive of as a mental process) while keeping the choice set fixed. We believe it is --- e.g., one might have choices over acts where the payoff in each state is described as either a single monetary payoff, or the sum of a number of positive and negative numbers (e.g., as in \cite{caplin2011search}).  The latter would likely induce a noisier experiment.  

We focus on three key orderings over experiments.  The first has been extensively studied since \cite{fechner1860elemente}: how much randomness there is in choice.  The second, inspired by recent work on cognitive uncertainty (\cite{enke2023cognitive}), is the degree of confidence the DM has that their choice is the ex-post correct one (i.e., it has the highest utility given the actual state).  The third is how well the DM ``understands'' the decision-problem, measured as the ex-ante expected payoff of the problem given the Blackwell experiment (this is inspired by Blackwell's ordering (\cite{blackwell1951comparison,blackwell}).  We highlight that it may be important to think about whether we believe we can observe randomness and confidence conditional on a state, or just as an average across states.  

Section \ref{sec:neg} presents our first main result: knowing how to rank experiments $\sigma$ and $\sigma'$ in any two of the orderings does not tell us how they may rank in the third.  For instance, one experiment may induce a choice that is both less random and where the DM exhibits more confidence, but it still may be the case that the DM understands the problem less, in that she will obtain a lower expected utility.  This implies that researchers need to be cautious in interpreting measures of randomness and confidence as providing information about the DM's understanding of the choice problem.  Existing results which imply a linkage rely on specific functional forms in order to derive those results.

In Section \ref{sec:suffsignal} we look to bridge the gap between the existing literature and our negative result by providing an easily interpretable, and natural, sufficient condition that ensures that all three rankings co-move.  The sufficient condition is a restriction on the set of experiments.  The condition essentially guarantees that one experiment places more weight on signals that cause the DM to choose $i$ in states where $i$ is optimal, compared to the other.  As an extension we show that if we assume that experiments have the property that in states where option $x$ is ordinally better than than $y$, signals which induce a choice of an option $x$ are more likely than those which induce a choice of $y$, it is the case that less randomness implies higher payoffs. We conclude the section by providing examples of classes of Blackwell experiments which either fail or satisfy our sufficient conditions, and demonstrate how this is reflected in the relationship between the rankings within each class.  

Section \ref{sec:suffpayoff} provides a distinct set of sufficient conditions, which operate over payoffs (rather than experiments).  Specifically, we focus on payoffs used in what we term psychophysical experiments, where the payoff depends only on choosing the option that has ``more'' of some quality (and does not depend on how much more).  We show that in this environment, confidence is equivalent to understanding.  Thus, we caution researchers that we should interpret results from these two different classes of experiments differently, and that we should be cautious when moving between different payoff functions.  

Section \ref{sec:ext} extends our analysis to consider additional rankings over experiments that have been discussed in the literature.  We begin by showing that willingness to accept to switch a choice (WTA), which is widely used in experiments to measure preference intensity, is, within our framework (where preferences are fixed), under the assumption that it can be measured in utils, equivalent to the degree of understanding of the choice problem (i.e.\ the ex-ante expected value of the problem) given the experiment. This suggests an alternative elicitation when measuring choice difficulty. We show that we cannot circumvent our negative results by considering attenuation, another widely measured proxy for choice difficulty (\cite{enke2024behavioral}).

Section \ref{sec:con} concludes by summarizing and discussing various extensions both within economics and related fields.  It also summarizes the contents of the appendices.

\section{Literature Review}\label{sec:lit}

Our framework builds on \cite{blackwell1951comparison,blackwell}'s notion of experiments, and supplements his classic ordering of informativeness—the idea that one experiment is more informative than another if it yields weakly higher expected payoff in all decision problems, with additional orderings over experiments.  Moreover, his ordering specifically inspired our notion of understanding (which is a completion of Blackwell). 

A key inspiration for our work has also been recent empirical work documenting the role of complexity in decision making and how it impacts both choice as well as other meta-cognitive measures.  \cite{enke2023cognitive} formalizes cognitive uncertainty and show that when subjects are cognitively uncertain, responses are attenuated toward intermediate defaults across risk, beliefs, and forecasts.  \cite{enke2024behavioral} report ``behavioral attenuation'' across more than 30 experiments, linking lower elasticities to higher measured cognitive uncertainty. Complementary work quantifies problem complexity in risky choice. \cite{enke2023quantifying} quantify complexity using randomness and cognitive uncertainty, and \cite{agranov2025complex} study subjective heterogeneity in perceived complexity. \cite{de2024caution} relate these measures with measures of ambiguity aversion.

Other researchers also measure the degree of confidence subjects have in their choice. E.g.,  \cite{butler2007imprecision, butler2011imprecision, lucia2024confidence, demartino2013confidence, cubitt2015imprecision} elicit confidence in the chosen option, using scales ranging from binary to a 0-100 scale (or percentages). Relatedly, \cite{falk2025limited}, under parametric assumptions show that one can recover the precision of information for subjects from response patterns in surveys.

Paralleling our theoretical work, \cite{bernheim2026improvable} points out empirical problems in eliciting measures of confidence due to misinterpretation of subjects (we rule this out, in essence assuming subjects perfectly understand the elicitations).

Although most work on confidence elictations are unincentivized, there is a small literature looking at how one can elicit notions of confidence in an incentive compatible way.
\cite{pkeski2025nondistortionary} shows that from a theoretical perspective that this is quite difficult (we assume away this issue, and suppose that the DM truthfully reports).

A key intuition that relates the degree of difficulty in choice to choice randomness comes from models inspired by \cite{fechner1860elemente}, including \cite{he2023random,he2024moderate,shubatt2024tradeoffs} where the probability an option is chosen (from a binary menu) is proportional to the utility difference between options, which offer accounts of random choice via imperfect discrimination.  These models can arise from Bayesian Expected Utility when restricted to specific functional forms. We discuss them in more detail in Appendix \ref{sec:appfech}.   

Our theoretical emphasis on information structures connects to models in which agents acquire signals prior to choice.  Like us, in  \cite{lu2016random} the experiments are exogenous, although his environment is richer than ours (the objects of choice are Anscombe-Aumann acts), which provides additional structure for identification.  \cite{safonov2017random, safonov2022random, DovalEilatLiuZhou2025} consider environments more similar to ours (which lack the richness of \cite{lu2016random}) and demonstrate that identification is typically much more difficult (if not impossible). \cite{rehbeck2019revealed} and \cite{caplin2021comparison} provide revealed preference type characterizations of models of learning in similar environments where the objects of choice generate state-contingent payoffs, and the utility function is known (and the experiments are exogenous), e.g., \cite{caplin2021comparison} considers a GARP style characterization of comparing the Blackwell informativeness of experiments where state-contingent data is available. \footnote{\cite{deOliveiraLamba2022} consider similar issues in a dynamic setting.} \cite{natenzon2019random} considers a specific structure on Blackwell experiments and considers the implications for choice.  \cite{vaeth2025imprecision} provides conditions under which ``noise'' attenuate updating.  Although related to our paper, we focus on measures that depend on the induced choice, whereas he focuses on beliefs over states.

 A distinct literature, rational inattention, allows the experiment to vary across choice problems, as in \cite{matvejka2015rational, caplin2015revealed}, and so the experiments are instead endogenous.

There is also a recent literature on cognitive imprecision that assumes individuals get noisy signals about the utilities (absolute or relative) of items, and then update their priors.  This literature, including \cite{KhawLiWoodford2021,Vieider2024, SteinerStewart2016,GossnerSteiner2018}, focuses on showing various behavioral anomalies can be generated due to imprecise signals, while recent experimental evidence, like \cite{Oprea2024} and \cite{frydman2023source,FrydmanJin2022}, provide some support for these hypotheses.

Relative to these existing theoretical frameworks, we focus on the information experiment as the primitive that choice situations induce, provide a clean mapping from lab measurements (randomness, confidence) to orders over experiments; and provide conditions under which these orders coincide or diverge.

One particularly fruitful area where both choice randomness and additional measurements have been leveraged to think about choice difficulty has been dynamic sequential-sampling models.  These link decision time with accuracy under costly information acquisition.  The literature is vast and extends well beyond economics (e.g., \cite{krajbich2010visual}), but has been recently explored by economists (\cite{fudenberg2018speed,gonccalves2024speed,gonccalvesrevealing}).  Our analysis abstracts from response times.  However, in Section \ref{sec:examples} we discuss an environment where the DM repeatedly samples from the same experiment, which is quite similar to the setting considered in this literature. However, our results are much weaker than those found in the literature, in large part because we allow for a much wider range of information generating processes.  

We also consider a classic measure of preference intensity: willingness-to-accept to switch (\cite{Knetsch84}), which has been used extensively within economics. We leverage similar payoff-based comparisons ---``WTA to switch'' between options--- to recover value-of-information in our framework.

Outside of economics, there is also a large literature on meta-cognition.  This work (see \cite{fleming2014measure} and \cite{rahnev2025comprehensive} for recent surveys) develops various measures that capture how well decision-makers understand the quality of their decisions, including choice confidence.  These typically use a very different paradigm  (where payoffs are not based on the utility of the outcome selected, but on whether the correct answer was made, see Section \ref{sec:suffpayoff} for a formal discussion) and also generally leverage specific functional form assumptions. 

Linking this literature outside of economics to formal economic analysis, recent work by \cite{bilotta2025introspection} develops a formal theory of metacognition and tests it using data. He  uses a psychophysics type framework (where only the ordinal ranking of options in a state matters for the decision-maker).  He then shows one can use confidence data, along with different test scoring regimes to understand whether individuals are sophisticated about the information process, in other words, whether they have a correct understanding of the Blackwell experiment they face.  In contrast, we assume they do.

\section{Understanding Choice Difficulty}\label{sec:main}

\subsection{Definitions}\label{sec:def}

Consider a finite set of alternatives, $X$, where typical options are denoted by $x,y,..,z$.  A decision-maker (DM) is not certain how the observed items map to outcomes;  or in other words, they are uncertain about the utility level of any given item.   To this end, let $\Omega$ be a finite set of states, with prior $\pi\in\Delta(\Omega)$. The decision maker is endowed with a state-dependent utility function $u:X\times\Omega\rightarrow\mathbb{R}$.  Refer to the tuple $(X,\Omega,u,\pi)$ as an \emph{environment}.  For $Y \subseteq X$ we 
 refer to $\Omega(Y)$ as states where elements of $Y$ are the maximizers of utility, and no elements of $X \backslash Y$ are maximizer of utility, that is,  
 $$\Omega(Y)=\{\omega \in \Omega:\arg\max_{x\in X}u(x,\omega)=Y\}.$$

An \emph{experiment} for $(\Omega,\pi)$ is a pair $(S,\sigma)$ consisting of a finite set of signals, $S$, and a map $\sigma:S \times \Omega\rightarrow [0,1]$ satisfying $\sum\limits_{s\in S}\sigma(s|\omega)=1$. The quantity $\sigma(s|\omega)$ denotes the probability that signal $s$ occurs when the state is  $\omega$. %Given a state $\omega$, we have  $\sum\limits_{s\in S}\sigma(s|\omega)=1$, i.e., $\sigma(\cdot|\omega)\in \Delta(S)$. %A type of experiment we will use frequently is experiments where $|S|=|\Omega|$ and each $\omega\in\Omega$ is associated with some $s_{\omega}\in S$.

Given that signals can occur with zero probability in all states, for our purposes it is without loss to assume that $S$ is the same    across experiments, and so simply identify experiments with $\sigma$.

 Given environment $(X,\Omega,u,\pi)$ and experiment $\sigma$, we define the posterior over states after receiving signal $s\in S$. That is, $$\pi_{\sigma}(\omega|s)=\frac{\pi(\omega)\sigma(s|\omega)}{\sum\limits_{\omega' \in \Omega }\pi(\omega')\sigma(s|\omega')}.$$

The posterior $\pi_{\sigma}(\cdot|s)$ governs choices via a uniform randomization over maximizers. Given a signal $s$ and experiment $\sigma$, each alternative $x$ yields the expected utility  

\[
 \sum_{\omega \in \Omega} u(x,\omega)\,\pi_{\sigma}(\omega | s).
\]

Then the set of maximizers is 

\[
C_{\sigma}(s)=\argmax_{x \in X} \sum_{\omega \in \Omega} u(x,\omega)\,\pi_{\sigma}(\omega | s)
\]

We may refer to the uniform randomization over $C_{\sigma}(s)$ as \(c_{\sigma}(s)\in \Delta(X)\).
Note that \(c_{\sigma}(x|s)=c_{\sigma}(y|s)\) for all
\(x,y \in C_{\sigma}(s) \). We also define  $S_\sigma(Y)$ as the set of signals where $Y \subseteq X $ yields maximal expected utility: $$S_\sigma(Y)=\{s\in S| \ Y = C_\sigma(s)\}$$

For the majority of this paper we will focus on a  binary choice environment with elements $x$ and $y$.  Given this, a state is simply a pair of utilities $\omega=(u(x), u(y))$. We also impose that individuals cannot distinguish ex-ante between $x$ and $y$: we impose that $\pi (u(x), u(y)) = \pi(u(y), u(x))$ so that $x$ and $y$ are symmetric; ex-ante, items are completely indistinguishable.\footnote{ The symmetry assumption is meant to assume ex-ante indistinguishability, which makes it harder to generate the negative results.}

Our goal here is to study comparisons of experiments according to various criteria. We focus on three orders.  The first, which is based on traditional measures of choice behavior, captures the amount of randomness the DM exhibits in their choice. The second, inspired by recent work on certainty in choice, measures how confident the DM is in their choice.  The last captures the degree of understanding a DM has for a problem by considering the expected payoff of the DM.  We think of the first two as being observable, while the third is typically unobservable.   

We consider two nested settings in terms of conditioning. The first is where we compare our observables conditional on each possible  state of the world.  In other words, it is possible to identify all data coming from a given state, and to group data by state, even if it is not possible to observe the utility realizations associated with each state.
The second is where the observable orderings hold ``on average,'' i.e., averaging over all possible states of the world, using the weights given by the prior.

First, we compare experiments in terms of choices. To do this, we define both conditional and unconditional random choice. The \emph{conditional random choice rule} induced by experiment $\sigma$ in state $\omega$ is given by 
$$\rho_{\sigma}(x|\omega)=\sum_{s\in S} c_{\sigma}(x|s)\sigma(s|\omega).$$ 
The \emph{(unconditional) random choice rule} induced by $\sigma$ then simply $$\rho_{\sigma}(x)=\sum_{\omega\in\Omega} \pi(\omega)\rho_{\sigma}(x|\omega).$$

We now provide two definitions of randomness. The first one requires less randomness in every state. The second one compares randomness by averaging across states.  
Randomness, or closely related measures, have been used in recent work to help assess the difficulty of problems.  The empirical literature frequently measures whether the same DM facing the same state gives different choices across repetitions, a proxy for randomness. \cite{enke2024behavioral} have subjects repeat a decision twice and then ask whether subjects made the same choice in both repetitions.  \cite{enke2023quantifying} think about the compression of choice probabilities towards 50-50 as a measure of choice difficulty; they look at subjects who faced the same problem five times and ask what fraction of them are inconsistent at least once (i.e. didn't make the same decision in all five iterations).  

\begin{defn}\label{wchoice} Let $\sigma$ and $\sigma'$ be two experiments. We say 
\begin{enumerate}
\item \label{choice} $\sigma$ is less random than $\sigma'$ if  $\max\limits_{x} \ \rho_{\sigma}(x|\omega)\geq \max\limits_{x} \rho_{\sigma'}(x|\omega)$ for every $\omega$,
\item \label{expectedchoice}$\sigma$ is expected-less random than $\sigma'$ if $\max\limits_{x} \ \rho_{\sigma}(x) \geq\max\limits_{x } \ \rho_{\sigma'}(x)$.
\end{enumerate}
\end{defn}

Intuitively, $\max_x \rho_{\sigma}(x\mid\omega)$ measures how decisive the DM is in state $\omega$: it is $1$ under deterministic choice and $1/2$ under maximal randomization. Our first notion of ``less random'' requires that in every state the most likely action under $\sigma$ is at least as likely as under $\sigma'$. The second notion instead compares the ex-ante distribution of choices: $\max_x \rho_{\sigma}(x)$ is the probability of the DM's most frequently chosen action. A higher value means that, on average across states, the DM's choices are more concentrated on a single option, and therefore less random. Notice that although expected less random is a complete ordering over experiments, less random is a partial order. Moreover, they are independent.   We will show that less random will be more useful in generating our positive results later in the paper.

\begin{rem}
    Less random neither implies, nor is implied by, expected less random.
\end{rem}

To see this, consider two states. State by state less random may occur because $x$ is chosen very often in one state, $y$ is chosen very often in another, but this means that aggregating across states, the two are chosen almost at the same probability, and so choice is very random.

To understand the definitions, we provide a simple two-state, two-signal, and two-option example. Consider a world where each option can take on one of two potential utility values, $u_H$ and $u_L$, but the two options cannot have the same utility value.  The prior is equal over the two states. Suppose also there are two signals $s_1$ and $s_2$.  An experiment $\sigma$ has the structure shown in Table \ref{table:Blackwell}. In this world, every experiment is parametrized by two variables: $\sigma=(\theta,\gamma)$. Under $\theta+\gamma > 1$, receiving $s_1$ makes $x$ optimal, and receiving $s_2$ makes $y$ optimal. So the induced choice probabilities are $\rho_{\sigma}(x \mid \omega_1)=\theta $, $ \rho_{\sigma}(y\mid\omega_1)=1-\theta$, $\rho_{\sigma}(y\mid\omega_2)=\gamma$, and $ \rho_{\sigma}(x\mid\omega_2)=1-\gamma$.

\begin{table}[H]
\begin{tabular}{ccc}
\toprule
& \multicolumn{2}{c}{Signals}  \\ \cmidrule(lr){2-3}
\multicolumn{1}{c}{\multirow{-2}{*}[0.2ex]{$\Omega$}}
 & \ $s_1$ \  & $s_2$ \   \\ \midrule
\multicolumn{1}{c|}{$(u_H,u_L) $} & $\theta$ & $1-\theta$  \\ \hline
\multicolumn{1}{c|}{$(u_L,u_H) $} & $1-\gamma$ & $\gamma$   \\  
\bottomrule
\end{tabular}
\caption{\tiny{The Blackwell Experiment parametrized by two variables $\sigma=(\theta,\gamma)$. Under $\theta+\gamma> 1$, receiving signal $s_1$ ($s_2$) implies choosing the first (second) option is optimal.}}\label{table:Blackwell}
    \end{table}

Notice that being less random can occur not only because the currently most common signal (in a given state) occurs more often, but also if it changes to occurring so infrequently the other signal now occurs frequently enough to induce choice dominance.    Formally, $(\theta',\gamma')$ is less random than $(\theta,\gamma)$ if $\max\{ \theta, 1-\theta\} \leq \max\{ \theta', 1-\theta'\}$ and $ \max\{ \gamma, 1-\gamma\} \leq \max\{ \gamma', 1-\gamma'\}$.  Note that $(1,1)$ is the least random experiment while $(0.5,0.5)$ is the most random one.  Figure \ref{fig:choicedom} illustrates state-by-state choice randomness. The red region contains experiments that are more random than $(\theta,\gamma)$, the green regions contain experiments that are less random than $(\theta,\gamma)$, and the white area reflects incomparability under the state-wise randomness order. 

The expected-less randomness relation is illustrated in the right panel of Figure~\ref{fig:expchoicedom}. Formally, $(\theta',\gamma')$ is
expected-less random than $(\theta,\gamma)$ if
\[
\max \{\theta + 1-\gamma, 1-\theta + \gamma\}
\leq \max \{\theta' + 1-\gamma', 1-\theta' + \gamma'\}.
\]
Unlike the state-by-state notion of choice randomness, expected randomness induces a complete order, with straight iso-expected-randomness loci. Note that $(1,1)$ and $(0.5,0.5)$ lie on the same locus and achieve the maximal level of expected randomness. Expected randomness decreases symmetrically as one moves outward from this locus in the direction indicated by the thicker level curves.

\begin{figure}[H]
\centering
\begin{subfigure}{.5\textwidth}
  \centering
  \includegraphics[width=.8\linewidth]{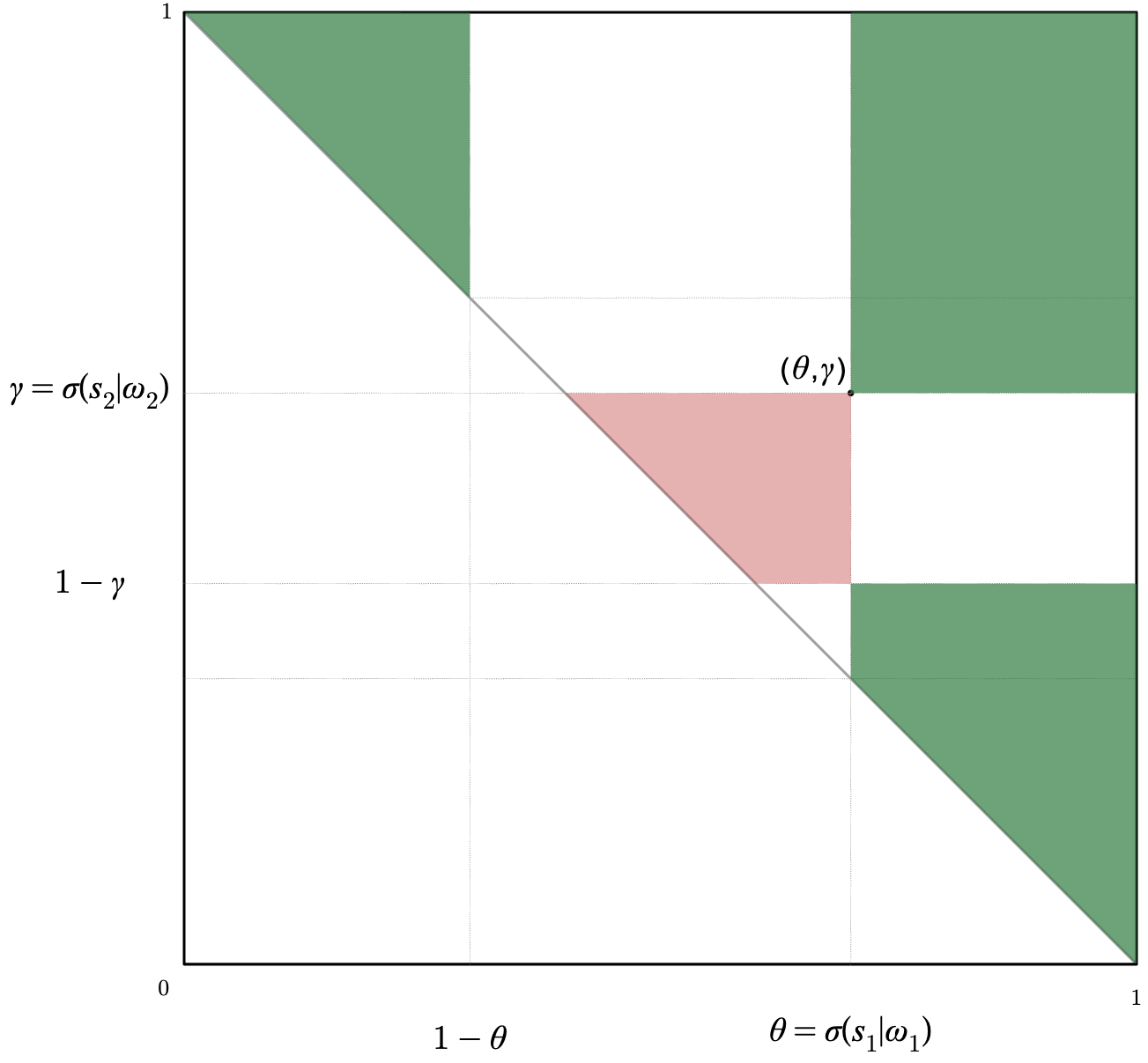}
    \caption{Choice Randomness}
    \label{fig:choicedom}
\end{subfigure}%
\begin{subfigure}{.5\textwidth}
  \centering
   \includegraphics[width=.8\linewidth]{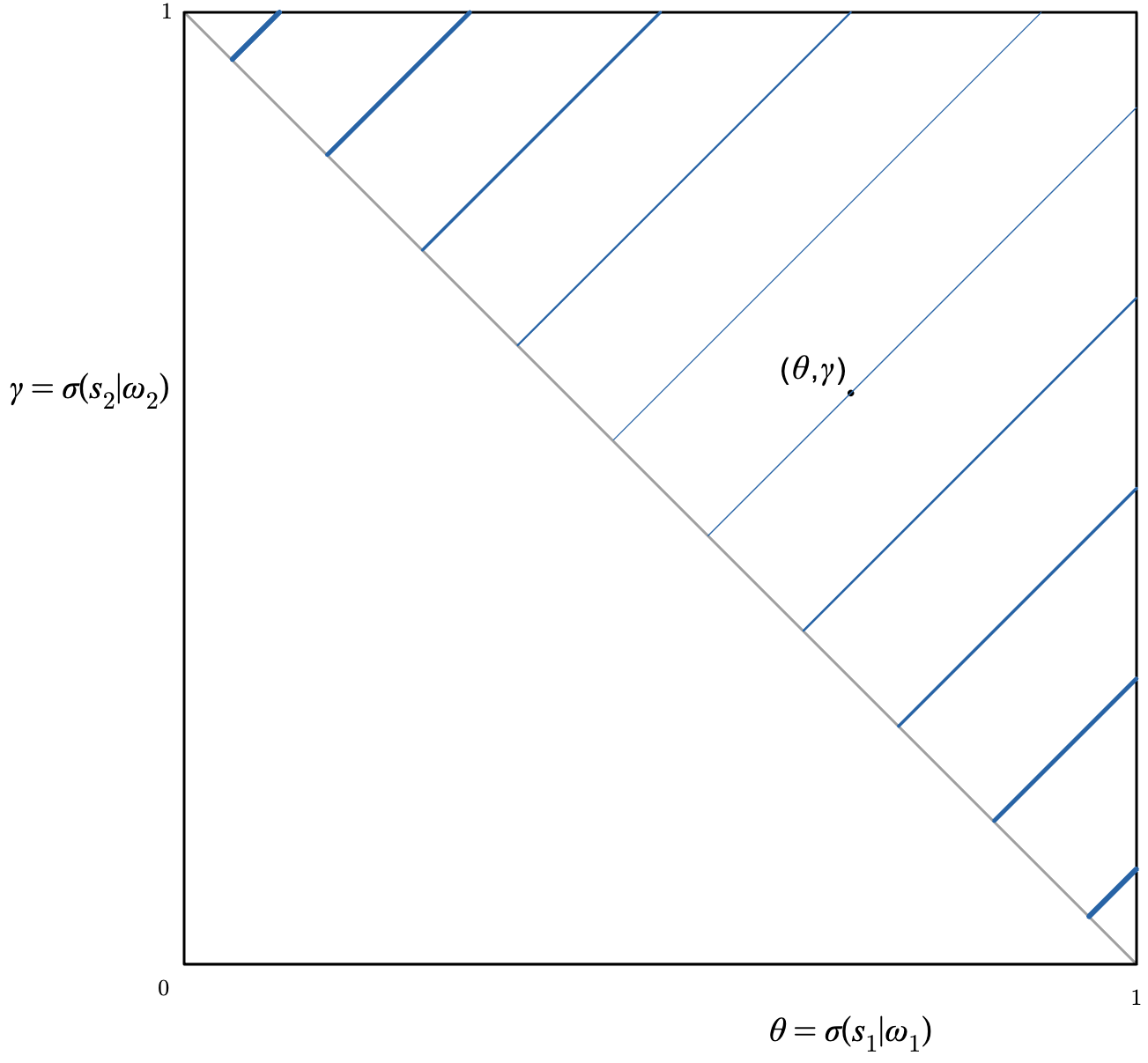}
    \caption{Expected Randomness}
    \label{fig:expchoicedom}
\end{subfigure}
\caption{\tiny{Illustration of randomness in a two-state, two-signal, two-option environment under the experiment $\sigma=(\theta,\gamma)$. Panel (A) shows state-by-state choice randomness: the light red region contains experiments that are more random than $(\theta,\gamma)$, the darker green regions contain experiments that are less random than $(\theta,\gamma)$, and the white area reflects incomparability under the state-wise randomness order. Panel (B) shows expected randomness: the order is complete, and each line depicts an iso-expected-randomness locus, with randomness increasing as one moves outward along the direction indicated by the thicker level curves.} }
\label{fig:random}
\end{figure}

%We can see from Figure \ref{fig:expchoicedom} why expected randomness may be a problematic measure --- there are experiments which align signals which states strictly more often in all states but which fail to be expected less random.  

%The figures also indicate why indicativeness can be important --- they force the less random sets to be convex (in Figure \ref{fig:choicedom} it selects the upper right green area).  

Now, we compare two experiments according to DM's confidence. To do this, we need more notation. Let $\hat{\Omega}(x)$ be the set of states where $x$ is an optimal choice; in other words, it is the union of $\Omega(x)$ and $\Omega(\{ x,y\})$. That is, $\hat{\Omega}(x)=\{\omega: u(x,\omega) \geq  u(y,\omega)\}$. Moreover, let $\hat{S}_{\sigma}(x)$ be the set of signals where $x$ is weakly optimal according to interim expected utility (i.e. the union of $S_{\sigma}(x)$ and $S_{\sigma}(\{x,y\})$.  Then conditional on $s \in \hat{S}_\sigma(x)$, the total probability of $x$ being optimal given $s$ is $\pi_{\sigma}(\hat{\Omega}(x)|s)$.  Then the (conditional) choice confidence given $\sigma$ and $\omega$ is\footnote{Recall that when $x$ and $y$ are tied in terms of interim expected utility the choice rule applies uniform randomization.} \begin{equation}\label{eq:confidence}\psi_{\sigma}(x|\omega)=\frac{\sum_{s\in S}\sigma(s|\omega) c_{\sigma}(x|s) \pi_{\sigma}(\hat{\Omega}(x)|s)}{\rho_{\sigma}(x|\omega)}.\end{equation}

The way to interpret this is that, given that the DM chose $x$ in state $\omega$, what is their average posterior probability that $x$ is actually optimal? The unconditional confidence (averaging across all possible states) is:  \begin{equation}\label{eq:confidence2}\psi_{\sigma}(x)=
\frac{\sum_{\omega\in\Omega}\sum_{s\in S}\pi(\omega)\sigma(s|\omega)c_{\sigma}(x|s)\pi_{\sigma}(\hat{\Omega}(x)|s)}{\sum_{\omega\in\Omega}\pi(\omega)\rho_{\sigma}(x|\omega)}.
 \end{equation}

Measures of confidence have been widely used in the literature.\footnote{Some models elicit confidence for belief elicitation rather than choice; e.g., \cite{de2024caution} or \cite{enke2023cognitive}.}  Recently they have been elicited as through measures of cognitive uncertainty.\footnote{\cite{enke2023cognitive} highlight that cognitive uncertainty may reflect things like imperfect perception, preference uncertainty, computational difficulties, and other aspects of bounded rationality.  To the extent that these can all be captured through a noisy signal about the right choice, they can be captured by our framework.} \cite{agranov2025complex} and \cite{lucia2024confidence} ask individuals how certain they are (on a numerical scale) that the choice they made was the correct one for them. \cite{demartino2013confidence} asks a similar question, albeit with a different scale.  Sometimes, as in \cite{enke2023quantifying, enke2024behavioral, bernheim2026improvable} the elicitation is in percent --- subjects are asked how certain they are (in percent) that they chose the best option.\footnote{The original paper on cognitive uncertainty, \cite{enke2023cognitive}, measures confidence in choice in a lottery as certainty that a lottery's value is within some interval.} \cite{cubitt2015imprecision, butler2007imprecision, butler2011imprecision} use a simple binary ``sure,'' ``not sure'' dichotomy to measure confidence.

\begin{defn} Let $\sigma$ and $\sigma'$ be two experiments and $X(\sigma,\sigma')$ be the set of options chosen unconditionally with strict probability.\footnote{Formally, $X(\sigma,\sigma')$ is equal to $\{x \in X \ | \  \rho_\sigma (x)\rho_{\sigma'} (x)>0\}$.} Then we say 

\begin{enumerate}
%\item  $\sigma$ $\omega$-confidence dominates $\sigma'$ if $\psi_{\sigma}(x|\omega) \geq \psi_{\sigma'}(x|\omega)$ for all $x$ where $\rho_\sigma (x)\rho_\sigma' (x)>0$.
%\item $\sigma$ confidence dominates $\sigma'$ if $\sigma$ $\omega$ confidence dominates $\sigma'$ for every $\omega$
\item $\sigma$ confidence dominates $\sigma'$ if  $\psi_{\sigma}(x|\omega) \geq \psi_{\sigma'}(x|\omega)$ for all $\omega$ and $x \in X(\sigma,\sigma')$,
\item $\sigma$ expected confidence dominates $\sigma'$ if $\psi_{\sigma}(x) \geq \psi_{\sigma'}(x)$
for all $x \in X(\sigma,\sigma')$.
\end{enumerate}
\end{defn}

Analogously to choice randomness, our definitions compare what happens conditional on a specific state, comparing what happens across all states (state by state) and comparing what happens averaging across states (using the prior).  Both of these ordering are partial orderings.  Interestingly enough, unlike randomness, where state-by-state measurements will be useful for our sufficient conditions, the opposite is true for confidence --- averaging across states will be more useful for our sufficient conditions.  We will also discuss, in Section \ref{sec:suffpayoff}, a measure of confidence which aggregates not just across states (like expected confidence) but also across items.  

We now illustrate these definitions in our a two-state, two-signal, two-option environment. The posterior probabilities are
\[
  \pi_{\sigma}(\omega_1| s_1)=\frac{\theta}{\theta+(1-\gamma)}, \qquad
  \pi_{\sigma}(\omega_2| s_2)=\frac{\gamma}{\gamma+(1-\theta)}.
\]
Substituting into \eqref{eq:confidence}, the conditional confidence in $x$ reduces to
\[
  \psi_{\sigma}(x|\omega_1)
  =\frac{\sigma(s_1|\omega_1)\cdot 1\cdot\pi_{\sigma}(\omega_1| s_1)}
        {\sigma(s_1|\omega_1)}
  =\pi_{\sigma}(\omega_1| s_1)
  =\frac{\theta}{\theta+1-\gamma}
\]
for \emph{both} $\omega\in\{\omega_1,\omega_2\}$; symmetrically, $\psi_{\sigma}(y|\omega)=\gamma/(\gamma+1-\theta)$ for both states.  Intuitively, given that the DM chose $x$, the signal $s_1$ must have been observed, so confidence depends only on how diagnostic $s_1$ is---not on which state actually obtained.

Because $\psi_{\sigma}(x|\omega)$ is constant across states, equations \eqref{eq:confidence} and \eqref{eq:confidence2} coincide, and hence confidence domination and expected-confidence domination are equivalent in this environment. Explicitly, $(\theta,\gamma)$ confidence dominates $(\theta',\gamma')$ if and only if 
$$
  \frac{\theta}{1-\gamma}\geq\frac{\theta'}{1-\gamma'}
  \quad\text{and}\quad
  \frac{\gamma}{1-\theta}\geq\frac{\gamma'}{1-\theta'},
$$ i.e.\ both posterior probabilities---of $\omega$ after $s_1$ and of $\omega'$ after
$s_2$---are weakly higher under $(\theta,\gamma)$.

Figure~\ref{fig:confidence} illustrates this in the $(\theta,\gamma)$ plane.  For a
reference experiment $\sigma$, the darker red region contains experiments that confidence dominate $\sigma$ (both likelihood ratios are higher), the lighter red region contains experiments dominated by $\sigma$, and the two remaining wedges consist of experiments that are incomparable with $\sigma$.  The two boundary lines $\theta/(1-\gamma)$ and $\gamma/(1-\theta)$ ( the likelihood ratios of $\sigma$) partition the space into these four regions.  

\begin{figure}[H]
  \centering
  \includegraphics[width=0.42\linewidth]{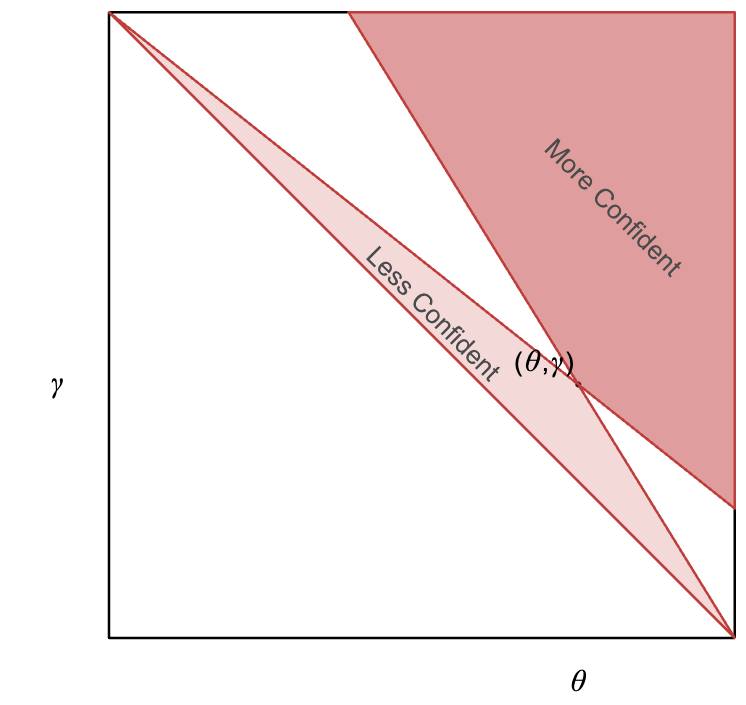}
  \caption{\tiny{Confidence dominance in the binary state/signal environment.  Dark Red: more
           confident than $\sigma$; Light Red: less confident; white wedges: incomparable.}}
  \label{fig:confidence}
\end{figure}

Before continuing, we note that expected confidence is for any option is simply the state-contingent confidence for a given option, averaged across states. 

\begin{rem} Confidence domination implies expected confidence domination.
\end{rem}

Last, we provide definitions that capture how well a decision-maker understands a choice problem, which we take as a proxy for their expected payoff.  The classic notion to capture this is the Blackwell ordering over experiments; but this seems too weak for our purposes (we revisit this ordering in Section~\ref{sec:blackwell}).  Recall that the Blackwell ordering asks whether an experiment provides a higher expected payoff to the DM across all possible decision problems.  Here we have a fixed set of actions (choose one of two options), a fixed mapping from choice to payoffs that depends on the state (the state-dependent utilities) and a fixed prior.  Thus, we focus on an ordering which ranks experiments according to the expected utility induced by the choice problem the DM actually faces.  We call this the choice value ranking.  

The expected utility given $\sigma$, conditional on state $\omega$ is $$W(\sigma|\omega) = \sum_{s \in S} \sigma(s|\omega) \sum_x c_{\sigma}(x|s) u(x|\omega). $$

The ex-ante expected utility of $\sigma$ is then just $$W(\sigma) =\sum_{\omega} \pi(\omega) \sum_{s \in S} \sigma(s|\omega) \sum_x c_{\sigma}(x|S) u(x|\omega). $$

\begin{defn} Let $\sigma$ and $\sigma'$ be two experiments. We say $\sigma$ choice payoff dominates $\sigma'$ if $W(\sigma) \geq W(\sigma')$.

%\item  $\sigma$ $\omega$-choice value dominates $\sigma'$ if $W(\sigma|\omega) \geq W(\sigma'|\omega)$ 
%\item $\sigma$ choice value dominates $\sigma'$ if $\sigma$ $\omega$ choice value dominates $\sigma'$ for every $\omega$

\end{defn}

Choice payoff domination is a complete ordering. To illustrate this concept in our two-state, two-signal, two-option environment, normalize so that selecting the higher-utility option yields $1$ and the lower-utility option yields $0$.  The expected payoff is $\theta$ in state $\omega_1$ and $\gamma$ in state $\omega_2$; with symmetric priors the expected payoff
is proportional to $\theta+\gamma$.  Figure~\ref{fig:payoff} illustrates: iso-payoff curves are lines $\theta+\gamma=\mathrm{const}$, so the higher-payoff set is the region north-east of the reference experiment along these lines.

\begin{figure}[h!]
  \centering
  \includegraphics[width=0.42\linewidth]{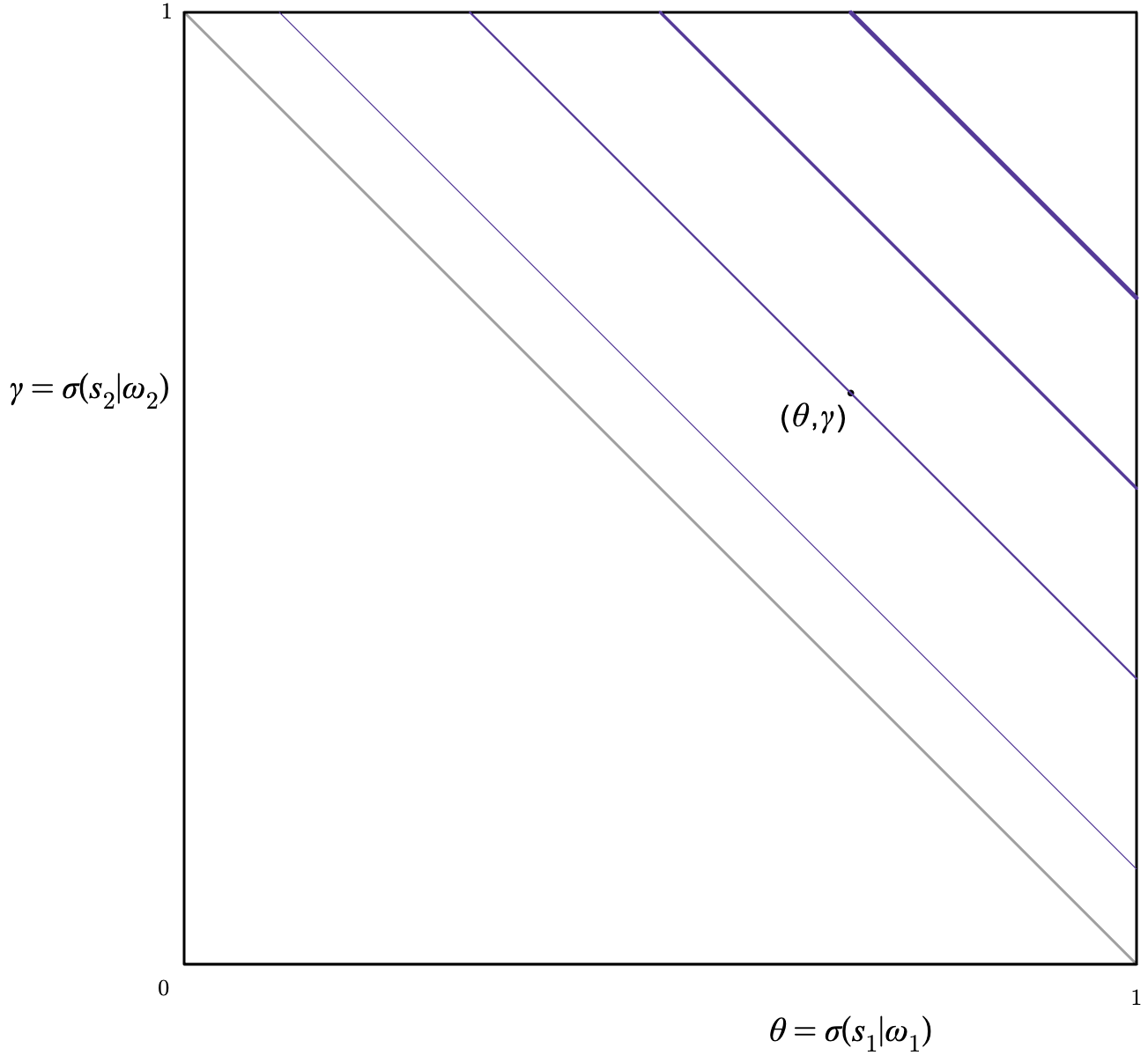}
  \caption{\tiny{Choice Payoff Domination: Iso-payoff curves are lines $\theta+\gamma=\mathrm{const}$, the direction indicated by the thicker level curves.}}
  \label{fig:payoff}
\end{figure}

\subsection{Relation Between Randomness, Confidence and Payoff Dominance}\label{sec:neg}

Although it seems somewhat intuitive that decreased randomness and higher confidence in choices should be indicative of greater understanding of a problem, in that the decision-maker has a higher expected payoff, this is not true.  In fact, even if one experiment generates \emph{both} lower randomness and greater confidence (either state by state, or in expectation) it is not the case that we can infer the DM understands it better; in fact they may have a lower expected payoff.  The lack of relationship between these orderings is even stronger though.  In fact, our next proposition shows that, even if we knew the ranking of expected payoffs induced by the Blackwell experiments, as well as observing one of either observing the randomness ranking or the confidence ranking, we would not know the ranking of the other.

%No two orderings imply the third. That is, for any two of the orderings, a common and known ranking between a pair of experiments does not indicate the same ranking for the third ordering.

\begin{proposition}\label{prop:neg2}~
    \begin{enumerate}

        \item If $\sigma$ is less random (or less expected random) than and confidence dominates $\sigma'$, then it can be the case that $\sigma'$ choice payoff  dominates $\sigma$. 

        \item If $\sigma$ is less random (or expected less random) than and choice payoff dominates $\sigma'$ then $\sigma'$ can confidence dominate $\sigma$.
        
        \item If $\sigma$ confidence dominates and choice payoff dominates $\sigma'$ then $\sigma'$ can be less random (and expected less random) than $\sigma$.
    \end{enumerate}
\end{proposition}

Because confidence domination implies expected confidence domination the first part of  Proposition \ref{prop:neg2} holds if the former is replaced by the latter.  We prove the first part by example.  

\begin{example}\label{ex:3states_2} Suppose we have two items, but three utility levels, $0, 1$ and $1,000$.  We assume that items cannot have the same utility; but that the prior probability that an item has the highest utility ($1,000$) is $2\%$, and the probability it has the other two utility levels ($0$ and $1$) are each $49\%$. We consider two Blackwell experiments described in the following table.

\begin{table}[h!]
\begin{tabular}{ccccccccc}
\toprule
& & \multicolumn{3}{c}{$\sigma$} & \ \  & \multicolumn{3}{c}{$\sigma'$} \\ \cmidrule(lr){3-5} \cmidrule(lr){7-9}
Prior & States  & \multicolumn{3}{c}{Signals} & & \multicolumn{3}{c}{Signals} \\ \cmidrule(lr){1-1}  \cmidrule(lr){2-2} \cmidrule(lr){3-5} \cmidrule(lr){7-9}
\multicolumn{1}{c}{\multirow{-1}{*}[0.2ex]{$\pi$  }} & \multicolumn{1}{|c|}{\multirow{-1}{*}[0.2ex]{$ (u(x), u(y))$}}  & \ $s_1$ \  & $s_2$ \  & $s_3$ \ &  & \ $s_1$ \  & $s_2$ \  & $s_3$ \ \\ \midrule
$0.005$ & \multicolumn{1}{|c|}{(1000,1)} & $1$ & $0$ & $0$ & & $1$ & $0$ & $0$  \\ \midrule
$0.005$ & \multicolumn{1}{|c|}{(1000,0)} & $1$ & $0$ & $0$ & & $\frac{1}{2}$ & $0$ & $\frac{1}{2}$  \\ \midrule
$0.005$ & \multicolumn{1}{|c|}{(1,1000)} & $1$ & $0$ & $0$ & & $0$ & $0$ & $1$ \\ \midrule
$0.490$ & \multicolumn{1}{|c|}{(1,0)} &$1$ & $ 0$ & $0$    & & $0$ & $0$ & $1$ \\ \midrule
$0.005$ & \multicolumn{1}{|c|}{(0,1000)} &$0$ & $ 1$ & $0$ &  & $0$ & $1$ & $0$  \\ \midrule
$0.490$ & \multicolumn{1}{|c|}{(0,1)} & $1$ & $0$ & $0$  &   & $1$ & $0$ & $0$ \\  
\bottomrule
\end{tabular}
\caption{Two Blackwell experiments $\sigma$ and $\sigma'$ for Example \ref{ex:3states_2}}\label{table:three_2}
    \end{table}

% \begin{table}[h!]
% \begin{tabular}{cccc}
% \toprule
% & \multicolumn{3}{c}{Signals}  \\ \cmidrule(lr){2-4}
% \multicolumn{1}{c}{\multirow{-2}{*}[0.2ex]{$\Omega = (u(a_1), u(a_2))$}}
%  & \ $s_1$ \  & $s_2$ \  & $s_3$ \ \\ \midrule
% \multicolumn{1}{c|}{(1000,1)} & $1$ & $0$ & $0$  \\ \midrule
% \multicolumn{1}{c|}{(1000,0)} & $1$ & $0$ & $0$  \\ \midrule
% \multicolumn{1}{c|}{(1,1000)} & $1$ & $0$ & $0$  \\ \midrule
% \multicolumn{1}{c|}{(1,0)} &$1$ & $ 0$ & $0$  \\ \midrule
% \multicolumn{1}{c|}{(0,1000)} &$0$ & $ 1$ & $0$  \\ \midrule
% \multicolumn{1}{c|}{(0,1)} & $1$ & $0$ & $0$  \\  

% \bottomrule\\
% \end{tabular}

% \caption{$\sigma$ for Example \ref{ex:3states_2}}\label{table:three_2.1}
%     \end{table}

% \begin{table}[h!]
% \begin{tabular}{cccc}
% \toprule
% & \multicolumn{3}{c}{Signals}  \\ \cmidrule(lr){2-4}
% \multicolumn{1}{c}{\multirow{-2}{*}[0.2ex]{$\Omega = (u(a_1), u(a_2))$}}
%  & \ $s_1$ \  & $s_2$ \  & $s_3$ \ \\ \midrule
% \multicolumn{1}{c|}{(1000,1)} & $1$ & $0$ & $0$  \\ \midrule
% \multicolumn{1}{c|}{(1000,0)} & $\frac{1}{2}$ & $0$ & $\frac{1}{2}$  \\ \midrule
% \multicolumn{1}{c|}{(1,1000)} & $0$ & $0$ & $1$  \\ \midrule
% \multicolumn{1}{c|}{(1,0)} &$0$ & $ 0$ & $1$  \\ \midrule
% \multicolumn{1}{c|}{(0,1000)} &$0$ & $ 1$ & $0$  \\ \midrule
% \multicolumn{1}{c|}{(0,1)} & $1$ & $0$ & $0$  \\ 

% \bottomrule\\
% \end{tabular}

% \caption{$\sigma'$ Blackwell Experiment for Example \ref{ex:3states_2}}\label{table:three_2.2}
%     \end{table}

Under $\sigma$ in each state, there is no randomness in choice (option 1 is always chosen in 5 of the 6 states, and option two is always chosen in the remaining state).  Moreover, if option 1 is chosen, confidence is greater than 50\%.  If option 2 is chosen, confidence is 100\%. 

Now consider $\sigma'$.  Notice that $\sigma'$ is Blackwell more informative than $\sigma$, and has a higher expected payoff for the DM for all choice problems, thus $\sigma'$ choice payoff dominates $\sigma$. In all but the second state, there is still no randomness (option 1 is always chosen in the first and last states, option 2 in the third, fourth and fifth).  In the second state each option is chosen with equal chance.  When Option 1 is chosen in the first, second, and last state, confidence is below 50\%.  Similarly, confidence when option 2 is chosen in the second, third, and fourth states, confidence is below 50\%.  Confidence in the fifth state is still 100\%.  Thus, randomness has gone up, and confidence has gone down when moving from $\sigma$ to $\sigma'$. 

Intuitively, the example shows that we can move to an experiment that is more informative, but it introduces more randomness in one state (keeping it constant in all other states), and it reduces confidence because the DM chooses the option that has a high potential upside which is infrequently realized.
$\Box$
\end{example}

We can prove the second part Proposition \ref{prop:neg2}  by example as well: here choice payoff goes up and randomness goes down, but confidence goes down.  

    \begin{example}\label{ex:secondcounter} This example leverages similar intuitions to the previous one.  There are two alternatives $X=\{x,y\}$ but four states, represented by the utility vectors $u(\cdot|\omega_1)=(1,0), u(\cdot|\omega_2)=(1,100), u(\cdot|\omega_3)=(0,1), u(\cdot|\omega_4)=(100,1)$.  Thus, one of the alternatives has utility $1$ with certainty, but there is uncertainty as to which one it is.  Assume a prior over states which is symmetric, so that $\pi(\omega_1)=\pi(\omega_3)=.49$ and $\pi(\omega_2)=\pi(\omega_4)=.01$.  
    
With the completely uninformative information structure, the confidence that the right decision has been made after either choice $x$ or $y$ is $.5$, as the two alternatives are indifferent with the given prior.  Thus choice of either $x$ or $y$ is uninformative about the state.  The probability of a state in which $x$ is optimal is $0.5$ (similarly for $y$).

Imagine now the binary signal structure $S=\{s_x,s_y\}$ where $\sigma(\omega_1)=\sigma(\omega_2)=\delta_{s_x}$ and $\sigma(\omega_3)=\sigma(\omega_4)=\delta_{s_y}$.  This signal describes which of the two alternatives has utility 1, but nothing more.  

We claim that $y$ is chosen with probability one exactly when $s_x$ is realized.  This is because when $s_x$ is realized, $\pi_{\sigma}(\omega_1|s_x)=0.98$ and $\pi_{\sigma}(\omega_2|s_x)=0.02$, with all remaining states having zero probability.  Choice of $x$ therefore obtains a utility of $1$, whereas choice of $y$ obtains an expected utility of $2$, so $y$ is the uniquely optimal choice in this case.  Similarly, when $s_y$ realizes, $x$ is the uniquely optimal choice.

On the other hand, when $x$ is chosen (that is, when $s_y$ realizes), the probability that a state has obtained which renders $x$ optimal is $\pi(\omega_4|s_y)=0.02$, with a symmetric statement for when $y$ is chosen.  $0.5 > 0.02$, so a decrease in choice confidence has occurred even though there has been an increase in Blackwell informativeness. 

The intuition here is similar to the other example.  Because the DM learns which option has a utility of 1 when the experiment is informative, they choose the option that is very infrequently the best option, but when it is, it is much much better, leading to a reduction in confidence.
$\Box$ \end{example}

For the last example (proving the third part of Proposition \ref{prop:neg2}), suppose we have binary signals; we can show that confidence goes up, choice payoff goes up  but randomness also goes up.  

\begin{example}\label{ex:binarcounter}
Revisit the two-state, two-signal, two-option environment  in Section \ref{sec:def}. Suppose  $1-\theta < \gamma < \gamma' < 0.5$. Then   $(\theta,\gamma')$ is Blackwell more informative than $(\theta,\gamma)$ and so choice payoff dominates.  Moreover, state-by-state, $(\theta,\gamma')$ induces greater confidence than $(\theta,\gamma)$.  In the first state, randomness is clearly the same across $(\theta,\gamma)$ and $(\theta,\gamma')$.  In the second state, $(\theta,\gamma')$ is more random than $(\theta,\gamma)$. Moreover, $(\theta,\gamma')$ is also more expected random.  $\Box$
\end{example}

Thus, even if one experiment induces more understanding, and we observe less randomness, we cannot conclude that the individual will be more confident.  Similarly, knowing an experiment induces more confidence and greater understanding does not allow us to predict changes in randomness.  

Experiments relating randomness to confidence are not testing a theoretical equivalence.  Observing correlation tells us something about the structure of the experiment. In the same vein, one needs to be careful about inferring the level of understanding of the choice problem from randomness and confidence.

Can we produce ``positive'' results? We will do so in three ways.  First, we can restrict the set of experiments we consider to ensure that all three orders co-move.  Second, we discuss restriction on payoff structure that again induce co-movement within a subset of our measures.  Third, we suggest an alternative empirical measurement that comoves with choice payoffs.

\section{Sufficient Conditions on Signal Structure}\label{sec:suffsignal}

In this section, we provide conditions on Blackwell experiments that guarantee the three notions co-move. Given a Blackwell experiment $\sigma$, we construct another Blackwell experiment $\sigma'$ via a single \emph{aligned shift} if in a state $\omega$ where option $x$ is better than option $y$, we move weight towards a signal that induces the choice of $x$ more often from a signal that induces the choice of $x$ less often.

\begin{defn}\label{def:aligned}
     Experiment $\sigma'$ is obtained from an \emph{aligned shift} from $\sigma$ if for some option $x$ and state $\omega_i \in \Omega(x)$, $s \in S_{\sigma}(-x)$, $\hat{s} \in  S_{\sigma}(x)$, and $\sigma(s|\omega_i) - \sigma'(s|\omega_i) = \sigma'(\hat{s}|\omega_i) - \sigma(\hat{s}|\omega_i) >0$. 
\end{defn}

An aligned shift takes probability weight away from a signal that leads to the ``wrong'' choice and gives it to a signal that leads to the ``right'' choice within a given state.\footnote{Because two information structures may have a different number of signals that occur with positive probability, we say that a signal which has all entries equal to 0 is in $S_{\sigma}(x,y)$.} Importantly, a single aligned shift does not (i) allow changes in weight for states where the two outcomes have the same utility, and (ii) allow for shifts in probability that involves signals which indicate $x$ and $y$ are tied in expected utility.\footnote{In states where $x$ and $y$ have the same utility, shifting around signal probabilities never impacts expected payoffs but can change randomness and confidence.  If we shift probability either from or to a signal that currently indicates $x$ and $y$ are ``tied'' in expected utility, this causes a discontinuous shift in choice probability which can lead confidence and randomness to move in the ``wrong'' direction.}

Next, we define a neutral shift, which redistributes probability mass
only across signals that all induce the same choice.

\begin{defn}
Experiment $\sigma'$ is obtained from a \emph{neutral shift} from $\sigma$
if there exists an option $x$, a state $\omega_i \in \Omega(x)$, and two
signals $s,\hat s \in S_{\sigma}(x) \cap S_{\sigma'}(x)$ such that
$
\sigma(s|\omega_i) - \sigma'(s|\omega_i)
=
\sigma'(\hat s|\omega_i) - \sigma(\hat s|\omega_i)
> 0.
$
\end{defn}

In words, at state $\omega_i$ we transfer some positive probability mass
from one signal $s$ that leads to choosing $x$ to another signal $\hat s$
that also leads to choosing $x$, leaving the total probability of choosing
$x$ unchanged in that state.

Aligned and neutral shifts describe how to move from one experiment to another while tracking how signals change the DM’s choices. In addition, we sometimes want a property of a single experiment: whenever option $x$ is truly optimal in a state, the experiment should, in that state, generate signals that point toward $x$ at least as often as signals that point away from it. This rules out perverse experiments that systematically ``mislead'' the DM in states where a given option is correct. We call such experiments indicative.

\begin{defn}
    An experiment $\sigma$ is indicative if for any $\omega_i \in \Omega(k)$, $\sum_{j \in S(k)} \sigma(s_j|\omega_i) \geq \sum_{j \in S(-k)} \sigma(s_j|\omega_i)$.\footnote{$S(-k)$  indicates the set of signals that induce the DM to choose the option that is not $k$.  Note that the definition is silent about ``tie'' signals --- this is deliberate, as they do not impact the indicativeness of an experiment. }
\end{defn}

Indicative experiments have the property where signals that cause DM to choose $x$ occur more often than those that cause a choice of $y$ when $x$ is correct choice, and similarly for $y$.  

We now illustrate the implication of being indicative in the simple two-state, two-signal environment of Table~\ref{table:Blackwell}. It  demands that, in any state where option $x$ is truly better, signals that lead the DM to choose option $x$ occur at least as often as signals that lead her to choose option $y$, and analogously when option $2$ is better. Applied to Table~\ref{table:Blackwell}, this means that in the state where $(u_H,u_L)$ obtains we must have $\theta \geq 1-\theta$, and in the state where $(u_L,u_H)$ obtains we must have $\gamma \geq 1-\gamma$, so both $\theta$ and $\gamma$ must exceed $1/2$. Thus many  Blackwell experiments (those with $\theta<1/2$ or $\gamma<1/2$) are ruled out, showing that being indicative is a substantive restriction on the signal structure.

Our next result shows that aligned and neutral shifts provide a simple sufficient condition under which our three measures of difficulty move in the same direction. An aligned shift explicitly reallocates  probability mass from signals that lead the DM to make the wrong choice toward signals that lead her to make the correct choice (or at least a tie), state by state. Neutral shifts only reshuffle mass across signals that induce the same choice and are therefore behaviorally irrelevant. A sequence of such shifts from $\sigma$ to $\sigma'$ guarantees that every state is (in terms of signal probabilities) at least as favorable to correct choice under $\sigma'$ as under $\sigma$.
 
\begin{proposition}\label{prop:suff}
    Suppose $\sigma'$ can be constructed from $\sigma$ by a sequence of aligned and neutral shifts.  Then 
    \begin{enumerate}
        \item  $\sigma'$  choice payoff dominates  $\sigma$,
        \item $\sigma'$ expected confidence dominates $\sigma$,
        \item if $\sigma$ is indicative, then $\sigma'$ is less  random than $\sigma$.
    \end{enumerate} 
\end{proposition}

\begin{proof}
    Since choice payoff dominance, expected confidence dominance, and less-randomness are all transitive, it suffices to prove each part for a single aligned shift and a single neutral shift. Neutral shifts rearrange weight among signals that all induce the same choice distribution, so they affect neither choice probabilities, payoffs, nor confidence. We therefore focus entirely on a single aligned shift.

Let $\sigma'$ be obtained from $\sigma$ by a single aligned shift. By Definition~\ref{def:aligned}, there exist an option $x$, a state $\omega_i \in \Omega(x)$ (so $u(x|\omega_i) > u(y|\omega_i)$), signals $s \in S_\sigma(-x)$ and $\hat{s} \in S_\sigma(x)$, and $\varepsilon > 0$ such that $\sigma'$ agrees with $\sigma$ everywhere except:
\[
\sigma'(s|\omega_i) = \sigma(s|\omega_i) - \varepsilon, \qquad \sigma'(\hat{s}|\omega_i) = \sigma(\hat{s}|\omega_i) + \varepsilon.
\]

Next, we show that the choice induced by any given signal is stable under aligned shifts. Define the unnormalized expected utility advantage of $x$ over $y$ after signal $\tilde{s}$:
\[
\Delta_\sigma(\tilde{s}) := \sum_\omega \pi(\omega)\,\sigma(\tilde{s}|\omega)\,[u(x|\omega) - u(y|\omega)].
\]
Option $x$ is chosen after $\tilde{s}$ iff $\Delta_\sigma(\tilde{s}) > 0$; $y$ is chosen iff $\Delta_\sigma(\tilde{s}) < 0$; a tie when $\Delta_\sigma(\tilde{s}) = 0$. After the aligned shift, only signals $s$ and $\hat{s}$ are affected:
\begin{align*}
\Delta_{\sigma'}(s) &= \Delta_\sigma(s) - \varepsilon\,\pi(\omega_i)\,[u(x|\omega_i) - u(y|\omega_i)] < \Delta_\sigma(s) < 0, \\
\Delta_{\sigma'}(\hat{s}) &= \Delta_\sigma(\hat{s}) + \varepsilon\,\pi(\omega_i)\,[u(x|\omega_i) - u(y|\omega_i)] > \Delta_\sigma(\hat{s}) > 0, \\
\Delta_{\sigma'}(\tilde{s}) &= \Delta_\sigma(\tilde{s}) \quad \text{for all other } \tilde{s}.
\end{align*}
The first inequality uses $s \in S_\sigma(-x)$, i.e., $\Delta_\sigma(s) < 0$; the second uses $\hat{s} \in S_\sigma(x)$, i.e., $\Delta_\sigma(\hat{s}) > 0$. Therefore $s$ remains in $S_{\sigma'}(-x)$, $\hat{s}$ remains in $S_{\sigma'}(x)$, and all other signals keep their classification. 

Next, we show that $\sigma'$ choice payoff dominates $\sigma$. 

\noindent \textit{Proof of Part 1)} Recall that because both signals maintain their strict classifications, the optimal actions at $s$ and $\hat{s}$ are the same under $\sigma$ and $\sigma'$.

Expected payoff admits the representation
\[
W(\sigma) = \sum_s \max_a\, V_\sigma(a, s),
\]
where $V_\sigma(a, s) = \sum_\omega \pi(\omega)\,\sigma(s|\omega)\,u(a|\omega)$ is the unnormalized expected utility of action $a$ at signal $s$.

The shift affects only signals $s$ and $\hat{s}$. Define $h(a) = \varepsilon \cdot \pi(\omega_i) \cdot u(a|\omega_i)$. Then:
\[
V_{\sigma'}(a, s) = V_\sigma(a, s) - h(a), \qquad V_{\sigma'}(a, \hat{s}) = V_\sigma(a, \hat{s}) + h(a).
\]
Note $h(x) > h(y)$ since $u(x|\omega_i) > u(y|\omega_i)$.

Since $s \in S_\sigma(-x)$, the optimal action at $s$ under $\sigma$ is $y$: $\max_a V_\sigma(a,s) = V_\sigma(y,s)$. Since $\hat{s} \in S_\sigma(x)$, the optimal action at $\hat{s}$ is $x$: $\max_a V_\sigma(a,\hat{s}) = V_\sigma(x,\hat{s})$. By signal classification stability, these remain optimal under $\sigma'$. Therefore:
\begin{align*}
W(\sigma') - W(\sigma) &= \bigl[V_\sigma(y,s) - h(y)\bigr] + \bigl[V_\sigma(x,\hat{s}) + h(x)\bigr] - V_\sigma(y,s) - V_\sigma(x,\hat{s}) \\
&= h(x) - h(y) \\
&= \varepsilon\,\pi(\omega_i)\,\bigl[u(x|\omega_i) - u(y|\omega_i)\bigr] > 0.
\end{align*}
$\Box$

Next, we prove that $\sigma'$ expected confidence dominates $\sigma$.

\noindent \textit{Proof of Part 2)} Recall confidence, conditional on signal $s \in S(x)$ is $\frac{\sum_{\omega \in \hat{\Omega}(x)} \sigma(s|\omega) \pi(\omega)} {\sum_{\omega \in \Omega} \sigma(s|\omega) \pi(\omega)}$.  Confidence, conditional on signal $s \in S(x,y)$ is, by construction, $\frac{1}{2}$.  Thus, the expected confidence, conditional on choosing $x$ must be the (normalized) weighted average over confidence for each $s \in S(x)$, where the weights are the ex-ante chances that each $s$ occurs, with an extra weight of $\frac{1}{2}$ applied to signals in $S(x,y)$:

$$\frac{\sum_{s \in S(x)}   [\frac{\sum_{\omega \in \hat{\Omega}(x)} \sigma(s|\omega)\pi(\omega)} {\sum_{\omega \in \Omega} \sigma(s|\omega)\pi(\omega)}  (\sum_{\omega \in \Omega} \sigma(s|\omega)\pi(\omega))] + \frac{1}{2} \sum_{s\in S(x,y)} \frac{\sum_{\omega \in \hat{ \Omega}(x)} \sigma(s|\omega)\pi(\omega)} {\sum_{\omega \in \Omega} \sigma(s|\omega)\pi(\omega)}  (\sum_{\omega \in \Omega} \sigma(s|\omega)\pi(\omega)) }{\sum_{s \in S(x)}\sum_{\omega \in \Omega} \sigma(s|\omega)\pi(\omega) + \frac{1}{2} \sum_{s \in S(x,y)}\sum_{\omega \in \Omega} \sigma(s|\omega) \pi(\omega) }$$

Simplifying 

$$\frac{\sum_{s \in S(x)}   [\sum_{\omega \in \hat{\Omega}(x)} \sigma(s|\omega)\pi(\omega)] + \frac{1}{2} \sum_{s\in S(x,y)} \sum_{\omega \in \hat{ \Omega}(x)} \sigma(s|\omega)\pi(\omega)}{\sum_{s \in S(x)}\sum_{\omega \in \Omega} \sigma(s|\omega)\pi(\omega) + \frac{1}{2} \sum_{s \in S(x,y)}\sum_{\omega \in \Omega} \sigma(s|\omega)\pi(\omega)  }$$

Under an aligned shift, the second term in the numerator and denominator do not change.  So we can focus on the first terms.  Since we are considering signals that induce $x$ to be chosen, consider a state $\omega' \in \Omega(x)$.  For that state, and for $s
\in S(x)$ and $s'' \in S(y)$, increase $\sigma(s'|\omega')$ by $\epsilon \pi(\omega')$, decrease the $\sigma(s''|\omega')$ by $\epsilon \pi(\omega')$, ensuring that both still stay weakly positive after the change.

Now consider the expected confidence about the choice of $x$.  The first term in the numerator has increased by $\epsilon \pi(\omega')$, as has the first term in the denominator (by the same amount).  Since the fraction is $<1$, it must be the case that both the numerator and denominator increasing by the same amount causes the fraction to get closer to 1, i.e. expected confidence after choosing $x$ increases.  

Now consider ex pected confidence after choosing $y$ (i.e., the analogous equation to what we considered).  Now, given the shift we just considered, the numerator of the expected confidence after $y$ doesn't change at all, but the denominator decreases, thus the fraction increases, so expected confidence after choosing $y$ increases.  $\Box$

Last we prove less random under indicativeness.

\noindent \textit{Proof of Part 3)}  Assume $\sigma$ and $\sigma'$ are both indicative. We need: for every state $\omega$,
\[
\max_a\,\rho_{\sigma'}(a|\omega) \geq \max_a\,\rho_\sigma(a|\omega).
\]
Recall $\rho_\sigma(a|\omega) = \sum_{\tilde{s}} \sigma(\tilde{s}|\omega)\,c_\sigma(a|\tilde{s})$.

Since signal classifications are preserved ($c_{\sigma'}(\cdot|\tilde{s}) = c_\sigma(\cdot|\tilde{s})$ for all $\tilde{s}$), changes to choice probabilities arise only from the direct weight change at $\omega_i$.

For state $\omega_i$ (the shifted state, $\omega_i \in \Omega(x)$) we know that 
\begin{align*}
\rho_{\sigma'}(x|\omega_i) &= \sum_{\tilde{s}} \sigma'(\tilde{s}|\omega_i)\,c_\sigma(x|\tilde{s}) \\
&= \rho_\sigma(x|\omega_i) - \varepsilon \cdot \underbrace{c_\sigma(x|s)}_{=\,0} + \varepsilon \cdot \underbrace{c_\sigma(x|\hat{s})}_{=\,1} \\
&= \rho_\sigma(x|\omega_i) + \varepsilon.
\end{align*}
By indicativeness of $\sigma$, $\max_a \rho_\sigma(a|\omega_i) = \rho_\sigma(x|\omega_i)$ since $\omega_i \in \Omega(x)$. Therefore:
\[
\max_a\,\rho_{\sigma'}(a|\omega_i) \geq \rho_{\sigma'}(x|\omega_i) = \rho_\sigma(x|\omega_i) + \varepsilon > \max_a\,\rho_\sigma(a|\omega_i). 
\]

For states $\omega \neq \omega_i$, the shift only modifies $\sigma(\cdot|\omega_i)$ and leaves $\sigma(\cdot|\omega)$ unchanged for $\omega \neq \omega_i$. Combined with signal classification stability:
\[
\rho_{\sigma'}(a|\omega) = \sum_{\tilde{s}} \sigma'(\tilde{s}|\omega)\,c_{\sigma'}(a|\tilde{s}) = \sum_{\tilde{s}} \sigma(\tilde{s}|\omega)\,c_\sigma(a|\tilde{s}) = \rho_\sigma(a|\omega).
\]
Hence $\max_a\,\rho_{\sigma'}(a|\omega) = \max_a\,\rho_\sigma(a|\omega).$ Thus, across all states, $\sigma'$ is less random than $\sigma$. Notice that by construction $\sigma'$ must also be indicative.  Thus, indicativeness is preserved across any aligned or neutral shift.  $\Box$

All three parts have been proven.
\end{proof}

Our conditions ensure that $\sigma'$ has higher expected confidence than $\sigma$, but they \emph{do not} imply the stronger, state-by-state ordering in which $\sigma'$ is more confident than $\sigma$ in every state. Aligned and neutral shifts raise \emph{average} confidence, yet confidence need not increase pointwise. Nor do these conditions ensure that moving from $\sigma$ to $\sigma'$ through aligned shifts renders $\sigma'$ expected-less random than $\sigma$. Expected randomness may move in the opposite direction.

To see this, return to the binary-state, binary-signal environment. Let $\theta_1 = 0.6$ and $\gamma_1 = 0.7$. The second option is then chosen 55\% of the time, so the ex-ante choice distribution is relatively concentrated. Now perform an aligned shift setting $\theta_2 = 0.7$ while keeping $\gamma_2 = \gamma_1$. Although $\sigma_2$ is obtained from $\sigma_1$ via an aligned shift and $\sigma_1$ is indicative, $\sigma_1$ is expected-less random than $\sigma_2$, even though $\sigma_2$ is less random than $\sigma_1$ in every state. Thus, even in this simple environment, aligned shifts and indicativeness are insufficient to ensure that expected randomness co-moves with our other measures. Even more simply, if $\theta = \gamma$, every experiment in the resulting class has identical expected randomness, equal to $0.5$, yet confidence ranges from $0.5$ to $1$, and experiments are strictly ordered by choice payoffs (and by Blackwell informativeness).
  
Instead, our result implies that if we move from $\sigma$ to $\sigma'$ exclusively through aligned and neutral shifts (or exclusively through anti-aligned and neutral shifts), then confidence and expected payoffs co-move, and, under indicativeness, so does state-by-state randomness. 

An intuitive perspective on aligned shifts emerges in environments without ties—neither in utilities across options nor in the choices induced by signals. In such settings, aligned and neutral shifts increase, state by state, the total probability of receiving a signal that induces the correct choice. Finally, two experiments share the same support if they assign positive probability to the same set of signals; aligned and neutral shifts then adjust only the probabilities on this common support.

\begin{proposition}\label{prop:suffequiv}
    Suppose $\Omega(x,y)$ is empty, $\sigma$ and $\sigma'$ have the same support, and both $S_{\sigma}(x,y)$ and $S_{\sigma'}(x,y)$ are empty.  Then $\sigma'$ can be obtained from $\sigma$ by a sequence of aligned and neutral shifts if and only if for all $\omega \in \Omega(k)$, $\sum_{s \in S_{\sigma'}(k)} \sigma'(s|\omega) \geq \sum_{s \in S_{\sigma}(k)} \sigma(s|\omega)$.
\end{proposition}

\begin{proof}
    One direction is simple.  Given the restrictive environment, aligned shifts must cause, for $\omega \in  \Omega(k)$,  $\sum_{s \in S_{\sigma}(k)} \sigma(s|\omega)$ to increase so that the post shift value $\sum_{s \in S_{\sigma'}(k)} \sigma'(s|\omega)$ is greater than the pre-shift value $ \sum_{s \in S_{\sigma}(k)} \sigma(s|\omega)$.  Neutral shifts cannot cause the mass to decrease. 
    
    The other direction is also fairly simple.  Consider $\sigma$ and $\sigma'$ such that for all $\omega \in \Omega(k)$, $\sum_{s \in S_{\sigma'}(k)} \sigma'(s|\omega) \geq \sum_{s \in S_{\sigma}(k)} \sigma(s|\omega)$.  How do we construct the aligned and neutral shifts to go from $\sigma$ to $\sigma'$?  We do so in two steps.  First, we choose pairs of abitrary signals $s \in S(k)$ and $s' \in S(-k)$.  Fix a state $\omega \in \Omega(k)$.  We move probability mass from $s'$ to $s$ until the total probability assigned to $S(k)$ equals $\sum_{s \in S_{\sigma'}(k)} \sigma'(s|\omega)$.  We then do neutral shifts so that the probability attached to each signal equals $\sigma'(s|\omega)$.  We do this for each state.    
\end{proof}

Proposition \ref{prop:suff} shows that expected confidence and choice payoffs co-move under aligned shifts.  Moreover, if the experiment is indicative then randomness also co-moves.  This framing raises the question of what indicativeness itself does. The following result shows that for indicative experiments, less random implies choice payoff domination.  The converse is not true.\footnote{In fact, if we rule out tie states, we can obtain a somewhat stronger result.  In this case, higher state by state expected payoffs (see Section \ref{sec:addorder} for a discussion of this ordering which we call state-conditional choice payoff domination), which is strictly stronger than choice payoff domination, implies less random.  Moreover, less random also implies expected confidence domination.}

\begin{proposition}\label{prop:indict}

Suppose $\sigma$ and $\sigma'$ are indicative.  If $\sigma$ is less random than $\sigma'$ then $\sigma$ choice payoff dominates $\sigma'$.

\end{proposition}

\begin{proof}
Fix any state $\omega \in \Omega$. We consider two cases.

\smallskip
\noindent\emph{Case 1: $\omega$ is not a tie state.} Then either $\omega \in \Omega(x)$ or $\omega \in \Omega(y)$. Suppose without loss of generality that $\omega \in \Omega(x)$, so $\Delta(\omega) := u(x|\omega) - u(y|\omega) > 0$. The state-conditional expected payoff decomposes as
\[
W(\sigma|\omega) = \sum_s \sigma(s|\omega)\bigl[c_\sigma(x|s)\, u(x|\omega) + c_\sigma(y|s)\, u(y|\omega)\bigr] = u(y|\omega) + \Delta(\omega)\, \rho_\sigma(x|\omega),
\]
since $c_\sigma(x|s) + c_\sigma(y|s) = 1$ for every signal (the DM either picks $x$, picks $y$, or randomizes with probabilities summing to one at a tie signal).

Because $\omega \in \Omega(x)$, we have $\rho_\sigma(x|\omega) = \rho_\sigma(\text{correct}|\omega)$. Indicativeness of $\sigma$ implies
\[
\rho_\sigma(\text{correct}|\omega) = \max_a \rho_\sigma(a|\omega),
\]
and the analogous identity holds for $\sigma'$. Hence, because $\sigma$ is less random than $\sigma'$,
\[
\rho_\sigma(x|\omega) = \max_a \rho_\sigma(a|\omega) \geq \max_a \rho_{\sigma'}(a|\omega) = \rho_{\sigma'}(x|\omega).
\]
Since $\Delta(\omega) > 0$,
\[
W(\sigma|\omega) = u(y|\omega) + \Delta(\omega)\, \rho_\sigma(x|\omega) \geq u(y|\omega) + \Delta(\omega)\, \rho_{\sigma'}(x|\omega) = W(\sigma'|\omega).
\]
A symmetric argument applies when $\omega \in \Omega(y)$.

\smallskip
\noindent\emph{Case 2: $\omega$ is a tie state.} Then $u(x|\omega) = u(y|\omega)$, and
\[
W(\sigma|\omega) = \sum_s \sigma(s|\omega)\bigl[c_\sigma(x|s)\, u(x|\omega) + c_\sigma(y|s)\, u(y|\omega)\bigr] = u(x|\omega),
\]
which does not depend on the experiment. Hence $W(\sigma|\omega) = W(\sigma'|\omega)$.

\smallskip
In both cases $W(\sigma|\omega) \geq W(\sigma'|\omega)$. Taking the expectation over $\omega$ with respect to $\pi$,
\[
W(\sigma) = \sum_\omega \pi(\omega)\, W(\sigma|\omega) \geq \sum_\omega \pi(\omega)\, W(\sigma'|\omega) = W(\sigma'),
\]
so $\sigma$ choice payoff dominates $\sigma'$.
\end{proof}

Thus, under indicativeness, less random is sufficient for choice payoff domination.  Recall that indicativeness requires that for any given state, the probability of obtaining a signal that ``agrees'' with the true ordinal ranking is higher than the probability of obtaining a signal that ``disagrees'' with the true ordinal ranking (where agreement and disagreement is given by the choice induced by the posteriors).

\subsection{Examples}\label{sec:examples}

We now present examples both satisfying and violating our sufficient conditions.

\subsubsection{Rational Inattention Logit} Consider experiment $\sigma$ where (i) there are two signals and ii) $\sigma(s_1|\omega_i)  = \frac{e^{\frac{u_i(x)}{\lambda}}}{e^{\frac{u_i(x)}{\lambda}}+e^{\frac{u_i(y)}{\lambda}}}$ for $\lambda>0$.  Call it a $\lambda$-Luce experiment, denoted by $\sigma^l_{\lambda}$. It is clear to see that this framework generates \emph{state-conditional Luce} choice probabilities (i.e. logit).  Moreover, the beliefs and choice probabilities can be derived as solution to rational inattention model where cost function is linear in mutual information with parameter $\lambda$ (see \cite{matvejka2015rational}).\footnote{Moreover, this model falls under the class of generalized Fechnerian approach discussed earlier in the paper.}

\begin{proposition}
    Consider the set of $\lambda$-Luce experiments.  The following are equivalent.\begin{itemize}
    \item \textit{$\lambda \leq \lambda'$}
    \item \textit{$\sigma^l_{\lambda}$ choice payoff dominates  $\sigma^l_{\lambda'}$}
    \item \textit{$\sigma^l_{\lambda}$ less random than (but equivalently expected random to) $\sigma^l_{\lambda'}$} 
    \item \textit{$\sigma^l_{\lambda}$ is more confident than $\sigma^l_{\lambda'}$}
\end{itemize}
\end{proposition} 

\begin{proof}
    It is easy to verify that decreasing $\lambda$ is an aligned shift. Therefore Proposition~\ref{prop:suff} yields almost all the stated equivalence, except item (2) in Proposition~\ref{prop:suff} only guarantees expected confidence domination, and we now show the full state by state confidence domination holds. This follows immediately from the proof of Proposition~\ref{prop:suff} when there are only two signals.  Similarly, it is easy to verify that the experiment is indicative.  
\end{proof}

We next present natural examples that violate our sufficient conditions and in which comovement breaks down.
 
 \subsubsection{Repeated Experiment} Consider a (non-trivial) Blackwell experiment $\sigma$, and let  $\sigma_t$ be the experiment generated by getting $t$ independent signals drawn from $\sigma$, that is, collecting data from $t$ independent replications of the original Blackwell experiment.

 % Let $\hat{\sigma}_\sigma$ denote the set of experiments $\sigma_1, \sigma_2...$ 

 \begin{proposition}
     If $t > t'$ then $\sigma_t$ choice payoff dominates $\sigma_{t'}$. But $\sigma_t$ may not be less random (and may be more expected random) and may not be more confident (nor more expected confident) than $\sigma_{t'}$.
 \end{proposition}

 \begin{proof}
     Take a binary signal structure that is non-indicative such as Example 3 and repeat it twice.  By way of a concrete example, consider $\theta_k=.9, \gamma_k=.2$.  With a single draw the vector of max choice probabilities (over states) is (0.9,0.8); the confidence after choosing $x$ is .529, after $y$ it is .66 (regardless of state).  If we do two draws then clearly we have a Blackwell more informative experiment.  Notice that the DM chooses $x$ only after observing $s_1$ twice, otherwise they choose $y$.  The vector of max choice probabilities is (.81, .64).  The confidence after $s_1, s_1$, $s_1, s_2$ or $s_2, s_1$, and $s_2, s_2$ respectively are .558, .64, .8.  Thus, the average confidence after choosing $y$ in the first state is .648 and in the second state .657.  Thus confidence after choosing $y$ has gone down in both states.\footnote{We conjecture that it is not possible to reduce both $x$ and $y$'s confidence in all states by repeating the experiment.}   
 \end{proof}

 Often, the literature interprets repeated draws from the same Blackwell experiment as a natural mental process for learning the value of the options (e.g. drift diffusion type models). This results shows that independent repetition of experiments is not sufficient to generate co-movement in the empirical measurements we are interested in.

 \subsubsection{CMC Rational Inattention}

Our discussion of $\lambda$-logit choice might lead to the intuition that perhaps rational inattention models, or models where experiments are chosen optimally, generally lead to comovement of our measures.  This is not true.

We consider a well-known cost structure used in the rational inattention literature: the CMC (constant marginal cost) approach of \cite{pomatto2023cost} where, given states of nature $\omega_i, \omega_j$ the cost of experiment $\sigma$ is $C(\sigma) = \sum_{\omega_i, \omega_j \in \Omega} \beta_{ij} ( \sum_{s} \sigma(s|\omega_i) \log \frac{\sigma(s|\omega_i}{\sigma(s|\omega_j} )$.  

Now consider the structure of Example \ref{ex:secondcounter}.  We will assume that $\beta_{1,2}=\beta_{2,1} =\beta_{3,4}=\beta_{4,3}= \infty$; and that all other $\beta$s=$\bar{\beta}$.  Consider two situations, one where $\bar{\beta}=\infty$, the other where $\bar{\beta}=0$.  These correspond to the two experiments discussed in Example \ref{ex:secondcounter}. 

  \subsubsection{Binary Experiments}
Even if our sufficient conditions fail, we may not get co-movement, but some relationship between orderings.  We demonstrate this in the environment of Example \ref{ex:binarcounter} so that there are binary states and binary signals.  Without loss we can assume that $\theta+\gamma \geq 1$.

\begin{proposition}\label{prop:binary}
    With binary states and binary signals, if $\sigma$ confidence (or expected confidence) dominates $\sigma'$ then $\sigma$ choice payoff dominates $\sigma'$ (but the opposite is not true).  If experiments are indicative, then if $\sigma$ is less random than $\sigma'$ then $\sigma$ confidence dominates $\sigma'$.
\end{proposition}

The proof immediately follows from visual inspection of Figures \ref{fig:random} to \ref{fig:payoff}. Thus, although our sufficient conditions are not satisfied in this domain, there is a relationship between our different orderings --- they are nested. 

A single aligned shift involves either increasing either $\theta$ or $\gamma$ --- this traces out a region that dominates the original experiment in both the horizontal and vertical dimensions in our figure.  One can see how, under this restriction, confidence and choice payoffs are equivalent; and if we add that experiments are indicative, that is also equivalent.  

%  \subsubsection{XXXX}
%{\color{red} Can we do another rational inattention model where we don't get comovement?}

\section{Sufficient Conditions on Payoff Structure}\label{sec:suffpayoff}

% \section{Psychophysical Choice}\label{sec:psycho}

The previous section provided a positive result by restricting the set of Blackwell experiments.  An alternative approach would be to restrict the set of allowable payoffs; we do so here.  This will allow us to link a measure of confidence to choice payoff domination, although it doesn't allow us to relate randomness to choice payoff domination.

The experimental paradigm described previously in the paper is the canonical binary-choice framework in economics, in which agents select the option that maximizes expected utility. A distinct paradigm, widely used in the study of human cognition—and central to Fechner’s \cite{fechner1860elemente} original work on random choice—asks individuals to compare two options and select the one that has ``more'' of a given attribute.

 These experiments span a broad range of domains, including line orientation \cite{blakemore1969orientation}, brightness \cite{weber1834weberlaw}, line length \cite{fechner1860elemente}, spatial frequency and contrast \cite{campbell1968fourier}, motion \cite{newsome1989neuronal}, and size \cite{krider2001pizzas}. In a related vein, recent economics experiments ask subjects to identify which of two options has the higher expected monetary payoff \cite{enke2023quantifying}. 

The crucial distinction in these environments is that incentives depend only on correctness: the decision maker does not value how much one option exceeds the other, but only whether it does. As we show below, this difference has substantive implications for the relationship between confidence and our other measures.  

We model this paradigm within our existing environment with one modification. The payoff from a choice is no longer the utility of the selected option. Instead, the decision maker receives a payment normalized to $1$ if she chooses the option with higher utility in the realized state, and $0$ otherwise.\footnote{Because the decision maker is an expected-utility-maximizing agent, this transformation is without loss.} We refer to these as {\it{psychophysical payoffs}}.

\begin{defn}
    The psychophysical payoff is:
\[
W^{psych}(\sigma) = \sum_{\omega \in \Omega} \pi(\omega) \sum_{s \in S} \sigma(s|\omega) \sum_x c_\sigma(x|s)\, \mathbf{1}_{\hat{\Omega}(x)}(\omega).
\]
\end{defn}

This alteration in the incentives of the problem changes the relationship between confidence and the other measures. In order to formally show this, we now define the natural completion of expected confidence dominance --- we average across items and states.

Define overall confidence as the average confidence, averaging both across states and chosen items:
\begin{align*}
    \psi_{\sigma} &= \sum_{s \in S(x)} \sum_{\omega \in \Omega} \psi_{\sigma}(x) \, \sigma(s|\omega) \, \pi(\omega) + \frac{1}{2} \sum_{s \in S(x,y)} \sum_{\omega \in \Omega} \psi_{\sigma}(x) \, \sigma(s|\omega) \, \pi(\omega) \\
    &\quad + \frac{1}{2} \sum_{s \in S(x,y)} \sum_{\omega \in \Omega} \psi_{\sigma}(y) \, \sigma(s|\omega) \, \pi(\omega) + \sum_{s \in S(y)} \sum_{\omega \in \Omega} \psi_{\sigma}(y) \, \sigma(s|\omega) \, \pi(\omega)
\end{align*}

We want to relate this back to one of the measures we consider previously: state, item dependent confidence.

\begin{lemma}
    $\psi_\sigma = \sum_{\omega \in \Omega} \pi(\omega) \sum_{x} \rho_\sigma(x|\omega) \, \psi_\sigma(x|\omega).$
\end{lemma}

\begin{proof}
Recall that conditional confidence satisfies
\[
\psi_\sigma(x|\omega) = \frac{\sum_{s \in S} \sigma(s|\omega)\, c_\sigma(x|s)\, \pi_\sigma(\hat{\Omega}(x)|s)}{\rho_\sigma(x|\omega)},
\]
so that
\[
\rho_\sigma(x|\omega)\, \psi_\sigma(x|\omega) = \sum_{s \in S} \sigma(s|\omega)\, c_\sigma(x|s)\, \pi_\sigma(\hat{\Omega}(x)|s).
\]
Summing over items and states:
\begin{align*}
\sum_{\omega} \pi(\omega) \sum_x \rho_\sigma(x|\omega)\, \psi_\sigma(x|\omega) 
&= \sum_{\omega} \pi(\omega) \sum_x \sum_{s} \sigma(s|\omega)\, c_\sigma(x|s)\, \pi_\sigma(\hat{\Omega}(x)|s) \\
&= \sum_s \sum_x c_\sigma(x|s)\, \pi_\sigma(\hat{\Omega}(x)|s) \sum_{\omega} \pi(\omega)\, \sigma(s|\omega) \\
&= \sum_s p(s) \sum_x c_\sigma(x|s)\, \pi_\sigma(\hat{\Omega}(x)|s),
\end{align*}
where $p(s) = \sum_\omega \pi(\omega)\sigma(s|\omega)$ is the marginal probability of signal $s$. Now expand the posterior:
\[
\pi_\sigma(\hat{\Omega}(x)|s) = \frac{\sum_{\omega \in \hat{\Omega}(x)} \pi(\omega)\,\sigma(s|\omega)}{p(s)}.
\]
Substituting:
\begin{align*}
\sum_s p(s) \sum_x c_\sigma(x|s)\, \pi_\sigma(\hat{\Omega}(x)|s) 
&= \sum_s \sum_x c_\sigma(x|s) \sum_{\omega \in \hat{\Omega}(x)} \pi(\omega)\,\sigma(s|\omega) \\
&= \sum_\omega \pi(\omega) \sum_s \sigma(s|\omega) \sum_x c_\sigma(x|s)\, \mathbf{1}_{\hat{\Omega}(x)}(\omega).
\end{align*}
This last expression equals $\psi_\sigma$ as defined above, since for each signal $s$: if $s \in S_\sigma(x)$, then $c_\sigma(x|s) = 1$ and the indicator selects only states where $x$ is optimal; if $s \in S_\sigma(x,y)$, the uniform randomization assigns weight $\frac{1}{2}$ to each item; and analogously for $s \in S_\sigma(y)$.
\end{proof}

\begin{defn}
    $\sigma$ overall confidence dominates $\sigma'$ if $\psi_{\sigma} \geq \psi_{\sigma'}$
\end{defn}

Although randomness is still distinct from the other two variables, overall confidence and the expected payoff from a problem now coincide.

\begin{proposition}\label{prop:psycho} $\sigma$ overall confidence dominates $\sigma'$ if and only if $\sigma$ has higher expected (ex-ante) psychophysical payoffs than $\sigma'$.  If $\sigma$ has higher expected (ex-ante) psychophysical payoffs than $\sigma'$, $\sigma$ may also have higher randomness.
\end{proposition}
\begin{proof}
We begin with the first sentence.  Recall that 
\[
\psi_\sigma = \sum_\omega \pi(\omega) \sum_s \sigma(s|\omega) \sum_x c_\sigma(x|s)\, \mathbf{1}_{\hat{\Omega}(x)}(\omega).
\]
The psychophysical payoff is
\[
W^{psych}(\sigma) = \sum_{\omega \in \Omega} \pi(\omega) \sum_{s \in S} \sigma(s|\omega) \sum_x c_\sigma(x|s)\, \mathbf{1}_{\hat{\Omega}(x)}(\omega).
\]
These two expressions are identical, so $\psi_\sigma = W^{psych}(\sigma)$. Therefore $\psi_\sigma \geq \psi_{\sigma'}$ if and only if $W^{psych}(\sigma) \geq W^{psych}(\sigma')$.

We show the second part by example.  Consider the binary-state, binary-signal setup from Example~\ref{ex:binarcounter}. Set $\theta_1 = \theta_2 > 0.5$ and $\gamma_1 < \gamma_2 < 0.5$, but that $\theta_1 > 1-\gamma_1, 1-\gamma_2$ Then $\sigma_2$ is Blackwell more informative than $\sigma_1$, and hence has higher psychophysical payoffs.

In the first state, randomness is identical across the two experiments (both have the same $\theta$). In the second state, $\sigma_2$ assigns probability $\gamma_2 > \gamma_1$ to signal $s_1$. Since $\gamma_k < 0.5$, the most-likely choice in state~2 under both experiments is $x$ (chosen via signal $s_1$), with probability $1 - \gamma_k$. Because $\gamma_2 > \gamma_1$, we have $1 - \gamma_2 < 1 - \gamma_1$, so
\[
\max_x \rho_{\sigma_2}(x|\omega_2) < \max_x \rho_{\sigma_1}(x|\omega_2).
\]
Thus $\sigma_2$ is more random in state~2 despite having strictly higher psychophysical payoffs.
\end{proof}
%%% PAULO ASKS: what does it mean for a real life lab experiment to be 'indicative'? What in the design of a psychophisical lab experiment can be manipulated to change 'indicativeness'?

This result highlights that the structure of incentives matter when thinking of how to interpret measurements of choice difficulty.

\section{Other Orderings}\label{sec:ext}

\subsection{Willingness to Accept to Switch}
Rather than restricting attention to a subset of payoff functions or Blackwell experiments, a distinct way of attempting to circumvent our negative result is to look for alternative observable measures that relate to choice payoff domination. Here, we pursue this third option and discuss how a widely used measure of preference intensity allows us, in our environment to recover understanding

$WTA(\sigma)$ is (average) utility that the DM is willing to accept to switch to choosing the other option, from their initially chosen option, given experiment $\sigma$

\begin{defn}
    $\sigma$ has a higher WTA than $\sigma'$ if $WTA(\sigma) \geq WTA(\sigma')$.
\end{defn}

This is a complete ordering, and has been widely used in economics in endowment effect experiments (\cite{Knetsch84}) as well as to measure preference intensity.  In our setting, where preferences are fixed, it will recover the degree of understanding.  

In order to show this we need to make two assumptions.  First, that the DM understands WTA measure perfectly. Second, the researchers can measure WTA in utils (e.g., suppose the stakes are small enough and the utility function is smooth so that we have an approximately linear mapping from utils to money).  

Conditional on signal $j$ that causes DM to choose $x$, the DM's WTA to switch to $y$ is:

$$\sum_{i \in \Omega} [u_i(x) -u_i(y)] \frac{\sigma(s_j|\omega_i) \pi(\omega_i)}{\sum_{k} \sigma(s_j|\omega_k) \pi(\omega_k)}$$

$WTA(\sigma)$ is then 
$$\sum_{j \in S(x)} \sum_{i \in \Omega} [u_i(x) -u_i(y)] \sigma(s_j|\omega_i) \pi(\omega_i) + \sum_{j \in S(y)} \sum_{i \in \Omega} [u_i(y) -u_i(x)] \sigma(s_j|\omega_i) \pi(\omega_i) $$

\begin{proposition}\label{prop:wta}
    $\sigma$ has a higher WTA than $\sigma'$ if and only if $\sigma$ choice payoff dominates $\sigma'$.  
\end{proposition} 

\begin{proof}
We show that $WTA(\sigma)$ is an affine transformation of $W(\sigma)$ with experiment-independent intercept, so the two induce the same ranking.

First, consider the definition of $W(\sigma)$.  Since the DM chooses $x$ with certainty when $s \in S_\sigma(x)$, chooses $y$ when $s \in S_\sigma(y)$, and randomizes uniformly when $s \in S_\sigma(x,y)$:
\begin{align*}
W(\sigma) = \sum_{\omega} \pi(\omega) \bigg[ &\sum_{s \in S_\sigma(x)} \sigma(s|\omega)\,u(x,\omega) + \sum_{s \in S_\sigma(y)} \sigma(s|\omega)\,u(y,\omega) \\
&+ \sum_{s \in S_\sigma(x,y)} \sigma(s|\omega)\,\frac{u(x,\omega)+u(y,\omega)}{2} \bigg].
\end{align*}

Under the completely uninformative experiment, the symmetry assumption $\pi(u(x),u(y)) = \pi(u(y),u(x))$ implies the DM is indifferent and randomizes uniformly, yielding a baseline payoff that is independent of the experiment:
\[
W_0 = \sum_{\omega} \pi(\omega)\,\frac{u(x,\omega) + u(y,\omega)}{2}.
\]

Next we compute $W(\sigma) - W_0$. Using the fact that $\sum_{s \in S_\sigma(x)} \sigma(s|\omega) + \sum_{s \in S_\sigma(y)} \sigma(s|\omega) + \sum_{s \in S_\sigma(x,y)} \sigma(s|\omega) = 1$, the terms involving $S_\sigma(x,y)$ cancel, giving:
\begin{align*}
W(\sigma) - W_0 = \sum_{\omega} \pi(\omega) \bigg[ &\sum_{s \in S_\sigma(x)} \sigma(s|\omega)\,\frac{u(x,\omega) - u(y,\omega)}{2} \\
&+ \sum_{s \in S_\sigma(y)} \sigma(s|\omega)\,\frac{u(y,\omega) - u(x,\omega)}{2} \bigg].
\end{align*}

Comparing with the expression for $WTA(\sigma)$:
\[
WTA(\sigma) = 2\bigl[W(\sigma) - W_0\bigr].
\]
Since $W_0$ is a constant independent of the experiment, $WTA(\sigma) \geq WTA(\sigma')$ if and only if $W(\sigma) \geq W(\sigma')$.
\end{proof}

The intuition is that WTA measures the expected utility gap between the chosen and unchosen option, which equals twice the improvement in expected payoff over the baseline of random choice. Since the baseline is common to all experiments, ranking by WTA is equivalent to ranking by expected payoff.

Thus, a higher WTA to switch (measured in utils) corresponds to a higher value of information for the choice problem. This elicitation is often thought of to measure ``strength of preference.'' In our environment, strength of preference corresponds to how well the problem is understood.  

This result also points out that having a state-by-state higher WTA is ``too strong'' in that we may have higher expected payoffs from an experiment but not have higher state-by-state WTA.

If we drop the assumption that we can only pay in money (not utils), then how ``robust'' this result will be depends on the curvature of utility, as well as how small the WTA measure is (smaller measures are more likely to be reversed due to utility curvature).\footnote{E.g., one can construct bounds on how much curvature utility needs to have for a given WTA so that choice payoffs would be the reverse of WTA.}

\subsection{Attenuation}\label{sec:atten}
One important behavioral measure that is often considered is attenuation - see \cite{enke2024behavioral} for recent work on this.  In their settings, they define attenuation using the elasticity of an observable choice variable (i.e. a behavior) with respect to a change in a parameter.  

In our context, the ``choice'' variable is the probability that an option is chosen, and a change in parameter can be represented by changing the state.  Thus, given our Bayesian expected utility framework there are two key considerations.  First, we have to have access to state-contingent data, as attenuation is typically measured as the change across states.  Second, recall that researchers do not know the identity of any given state (even if they have state-contingent data).  Thus, when we have state-contingent data, we can only ``identify'' states by the choice probabilities assigned to each state.  

Given these desiderata, it seems natural then to measure attenuation as the change in choice probabilities across any pair of states; for a given pair of states $\omega_i,\omega_j$, define $\Delta_{ij}(\sigma) = \sum_{s \in S_{\sigma}(x)} \sigma(s|\omega_i) + \frac{1}{2} \sum_{s \in S_{\sigma}(x,y)} \sigma(s|\omega_i) - \sum_{s \in S_{\sigma}(x)} \sigma(s|\omega_j) - \frac{1}{2} \sum_{s \in S_{\sigma}(x,y)} \sigma(s|\omega_j)$.  

\begin{defn}
    $\sigma$ is less attenuated than $\sigma'$ if, for every $\omega_i,\omega_j$, (i) when $\Delta_{ij}(\sigma') > 0 $ then $\Delta_{ij}(\sigma) > \Delta_{ij}(\sigma')$, and when $\Delta_{ij}(\sigma') < 0 $ then $\Delta_{ij}(\sigma) < \Delta_{ij}(\sigma')$
\end{defn} 

This definition raises a distinct issue.  Define the set of $x$-maximal states under $\sigma$ as those with the highest value of $\sum_{s \in S_{\sigma}(x)} \sigma(s|\omega_i) + \frac{1}{2} \sum_{s \in S_{\sigma}(x,y)} \sigma(s|\omega_i)$, and $x$-minimal states under $\sigma$ as those with the lowest values.  

If $\sigma$ is more attenuated than $\sigma'$ then it must be the case that either the all states in the set of $x$-maximal states under $\sigma'$ have a higher total probability of choosing $x$ than under $\sigma$, or all the states in the set of $x$-minimal states under $\sigma'$ have a lower chance of choosing $x$ than under $\sigma$.  

In other words, more attenuation implies that either the minimum probabilities need to be higher, or the maximal probabilities need to be lower.  To see that this is true, suppose not.  Then if we change any interior probability state, that state must become more attenuated relative to either $x$ maximal states or $x$ minimal states, and less compared to the other.  Thus, in our environment, what seems to be a reasonable definition of attenuation is an extremely strong requirement that says the most extreme states (in terms of choice probabilities) must shift.  

This immediately rules out us being able to compare structures where both of them have choice probabilities of 0 and 1 in some states.  

Despite the strength of this definition, it isn't immune to the issues we raised previously.  The following proposition summarizes, showing that adding attenuation to either of the first two antecedents of Proposition \ref{prop:neg2} does not change the results, and also does not change the results when when we replace randomness with attenuation in the third part (notice that in the third part we can't add attenuation as an antecedent since it is a direct function of randomness).  

\begin{proposition}\label{prop:negatt}~
    \begin{enumerate}

        \item If $\sigma$ is less random (or less expected random) than, confidence dominates, and is less attenuated than $\sigma'$, then it can be the case that $\sigma'$ choice payoff  dominates $\sigma$. 

        \item If $\sigma$ is less random (or expected less random) than,  choice payoff dominates, and is less attenuated than $\sigma'$ then $\sigma'$ can confidence dominate $\sigma$.
        
        \item If $\sigma$ confidence dominates and choice payoff dominates $\sigma'$ then $\sigma'$ can be less attenuated than $\sigma$.
    \end{enumerate}
\end{proposition}

The proof, again, is by example.  We first note that Example \ref{ex:secondcounter} features a shift in attenuation in the same direction as randomness (because an uninformative experiment is more attenuated than any other experiment), and so it directly works here for the second part of the proposition. 

We now provide an example where confidence increases, randomness decreases, attenuation decreases, yet choice payoffs go down.  
\begin{example}\label{ex:att_1}

    \begin{table}[ht!]
\begin{tabular}{cccccc}
\toprule
& \multicolumn{2}{c}{$\sigma$} & \ \  & \multicolumn{2}{c}{$\sigma'$} \\ \cmidrule(lr){2-3} \cmidrule(lr){5-6}
& \multicolumn{2}{c}{Signals} & & \multicolumn{2}{c}{Signals} \\ \cmidrule(lr){2-3} \cmidrule(lr){5-6}
\multicolumn{1}{c}{\multirow{-2}{*}[0.2ex]{$\Omega = (u(a), u(b))$}}
 & \ $s_1$ \  & $s_2$ \  &  & \ $s_1$ \  & $s_2$ \  \\ \midrule
\multicolumn{1}{c|}{(10,0)} & $0.49$ & $0.51$ & & $0.48$ & $0.52$   \\ \midrule
\multicolumn{1}{c|}{(0,10)} & $0.5$ & $0.5$ & & $0.52$ & $0.48$  \\ \midrule
\multicolumn{1}{c|}{(1,0)} & $0.9$ & $0.1$ & & $0.95$ & $0.05$ \\ \midrule
\multicolumn{1}{c|}{(0,1)} &$0.1$ & $0.9$   & & $0.05$ & $0.95$  \\ 

\bottomrule\\
\end{tabular}
\caption{Two Blackwell experiments $\sigma$ and $\sigma'$ for Example \ref{ex:att_1}}\label{table:att_1}
    \end{table}

We consider an experiment with four states and two signals as in Table \ref{table:att_1}, with equal priors across states. Under both $\sigma$ and $\sigma'$, $a$ is optimal under $s_1$ and $b$ must be optimal under $s_2$. 

Under $\sigma$
\begin{enumerate}
\item the vector of maximum choice probabilities by state is $ (0.51, 0.50, 0.9, 0.9)$
\item the expected payoff to choosing $a$ after $s_1$ is 2.915 and $b$ is 2.563; while under $s_2$ those payoffs are 2.587 and 2.935, and so the total expected payoff is 2.925.
\item the confidence after choosing $a$ is .698 and after $b$ is .697 (in all states).
\end{enumerate}

Under $\sigma'$ 

\begin{enumerate}
\item the vector of maximum choice probabilities by state is $(0.52, 0.52, 0.95, 0.95)$
\item the expected payoff to choosing $a$ after $s_1$ is 2.875 and $b$ is 2.625 (and the reverse under $s_2$), and so the total expected payoff is 2.875 
\item the confidence after choosing either option is 0.715 (in all states). 
\end{enumerate}

Clearly $\sigma'$ is less random and has higher confidence.  Moreover, it is also clear that it has a lower expected payoff.  It is simple to see that $\sigma'$ is less attenuated; we have amplified choice probabilities directionally under $\sigma'$.\footnote{As should be clear, taking a slightly different definition of attenuation, only looking at the absolute changes in the choice probabilities, will not help us here --- suppose $\sigma$ is less attenuated than $\sigma'$ if $|\Delta_{ij}(\sigma)| \geq |\Delta_{ij}(\sigma)|$ for all $i,j$.  Under this definition we still have the exact same issue in the example.}$\Box$  
\end{example}

We last show that we can have higher payoffs, as well as higher confidence, but higher attenuation.  

\begin{example}\label{ex:att_2}

The next example is similar to the previous one, with experiments as in Table \ref{table:att_2}, however the prior on states (10,0) and (0,10) is 0.1, while the prior on states (1,0) and (0,1) is 0.4.

    \begin{table}[h!]
\begin{tabular}{cccccc}
\toprule
& \multicolumn{2}{c}{$\sigma$} & \ \  & \multicolumn{2}{c}{$\sigma'$} \\ \cmidrule(lr){2-3} \cmidrule(lr){5-6}
& \multicolumn{2}{c}{Signals} & & \multicolumn{2}{c}{Signals} \\ \cmidrule(lr){2-3} \cmidrule(lr){5-6}
\multicolumn{1}{c}{\multirow{-2}{*}[0.2ex]{$\Omega = (u(a), u(b))$}}
 & \ $s_1$ \  & $s_2$ \  &  & \ $s_1$ \  & $s_2$ \  \\ \midrule
\multicolumn{1}{c|}{(10,0)} & $0.7$ & $0.3$ & & $0.7$ & $0.3$   \\ \midrule
\multicolumn{1}{c|}{(0,10)} & $0.3$ & $0.7$ & & $0.3$ & $0.7$  \\ \midrule
\multicolumn{1}{c|}{(1,0)} & $0.2$ & $0.8$ & & $0.18$ & $0.82$ \\ \midrule
\multicolumn{1}{c|}{(0,1)} &$0.8$ & $0.2$   & & $0.82$ & $0.18$  \\ 

\bottomrule\\
\end{tabular}
\caption{Two Blackwell experiments $\sigma$ and $\sigma'$ for Example \ref{ex:att_2}}\label{table:att_2}
    \end{table}

Under $\sigma$

\begin{enumerate}
\item the vector of maximum choice probabilities by state is $ (0.7, 0.7, 0.8, 0.8)$
\item the expected payoff to choosing $a$ after $s_1$ is 1.56 (and also for $b$ under $s_2$) 
\item the confidence after choosing $a$ or $b$ is .3 (in all states).
\end{enumerate}

Under $\sigma'$ 

\begin{enumerate}
\item the vector of maximum choice probabilities by state is $(0.7, 0.7, 0.82, 0.82)$
\item the expected payoff to choosing $a$ after $s_1$ is 1.54 (and also for $b$ under $s_2$) 
\item the confidence after choosing $a$ or $b$ is .284 (in all states).
\end{enumerate}

Thus, one can see that $\sigma$ is more attenuated than $\sigma'$, but it also is more confident and has a higher payoff.$\Box$  

    \end{example}

    \section{Concluding Remarks}\label{sec:con}

\subsection{Summary} 

This paper provides a unifying framework for understanding empirical measures of choice difficulty through the lens of Blackwell experiments. Our central finding is cautionary: in general, common measures—choice randomness, expressed confidence, and decision-maker understanding—need not move together. However, results are not all negative.  We identify intuitive sufficient conditions under which these measures align, both by (i) restricting the set of allowable Blackwell experiments, and (ii) restricting the set of payoff functions.  The former shows that if we assume experiments are related by aligned and neutral shifts the three measures must co-move.  Moreover, within the class of experiments that are indicative, the ranking in terms of randomness is sufficient to know how payoffs are ranked.  The latter demonstrates psychophysical experiments --- where payoffs depend only on correct identification --- exhibit different properties than standard economic experiments, with confidence equivalent to understanding in the former but not the latter. 

Our results highlight that depending on what kind of sufficient condition one wants to assume, the relevant kind of ordering (confidence or randomness) and the level of aggregation (state by state or averaged across states), matters.  Section \ref{sec:sumorder} provides a summary of the different orderings and levels of aggregation and how each one relates to our results.  

We further show that willingness-to-accept measures directly capture understanding when elicited in utils, offering an alternative empirical approach. We also show that measuring attenuation fails to overcome our negative results.  

Broadly, our findings highlight that some caution is necessary in interpreting empirical measurements often thought to correspond to choice difficulty and when generalizing insights across experimental paradigms.\footnote{For example, one interpretation of the System 1/System 2 literature (\cite{kahneman2011thinking}) is that that System 2 would correspond to a more informative experiment.  However, our results would imply that System 2 would not necessarily lead to more confidence or less randomness.}$^,$\footnote{There is generally a suggestion that combining sources of information is always ``good.'' In our environment these repeated experiments clearly improve payoffs, but may not generate improvements in randomness and confidence.} 

\subsection{Assumptions and Robustness} There are various ways one could try to aggregate or disaggregate our measurements.  For example, rather than measuring average confidence conditional on a choice, we could instead measure the full distribution of confidence (since each signal is associated with its own confidence level).  Such alternatives, although interesting, fail to avoid the negative results, as can be seen by checking our examples.  One might also consider using different ways of comparing payoffs across structures (e.g., using the Blackwell order, rather than our completion).\footnote{
\cite{caplin2021comparison}, using state-contingent random choice data, provide an elegant characterization result of when one experiment is Blackwell more informative than another.  Their tests can be useful in settings such as our, although because Blackwell is incomplete, the tests can  fail to rank many experiments.}  We discuss these in Appendix \ref{sec:addorder} and show they also do not overcome our primary negative results.  

Of course, if we actually know cardinal, or even just the ordinal, ranking of items in each state, along with the state conditional data, then it is easy to construct measurements that map from expected payoffs to behavior (and back).  E.g., suppose we know the ordinal ranking of each state.  Then we can simply look in the state-contingent data and ask how often the DM makes the ``right'' choice.  A ranking in this implies choice payoff domination.  

The analysis in this paper focuses on comparing what happens to a fixed choice problem as the experiment changes; in other words the only differences between choices could come from different Blackwell experiments.  In many empirical settings decision-makers face different choice sets, and researchers want to compare behavior (and understanding) across choice sets, where not just the experiment, but also the state space (and so utilities and priors) may change.  This may be important in settings where the choice set itself determines the experiment and the experiment cannot be independently manipulated by researchers.  This would, of course, only strengthen our negative results.  In Appendix \ref{sec:crosschoice} we discuss how to extend our sufficient conditions to this environment, showing it requires a strengthening of aligned shifts.  

One may also want to extend the analysis to allow for non-binary choice sets.  Most of the definitions naturally extend, with the exception of randomness.  While distance to 50-50 choice is natural in binary sets, there are many possible extensions to larger choice sets, although some seem quite natural, e.g., entropy.    

Our approach assumes sophistication of DMs --- they know their Blackwell experiment, and can correctly infer confidence (i.e., posterior beliefs) from signals.  Such an assumption may not be realistic; e.g., papers such as \cite{lichtenstein1977those} documenting the ``hard-easy'' effect, or  \cite{frederick2005cognitive}  cognitive reflection task questions provide evidence of systematic over or underconfidence relative to accuracy in questions, and \cite{bilotta2025introspection} provides evidence that individuals may not correctly estimate their chances of getting a question on an exam correctly.  That said, \cite{SandersHangyaKepecs2016} shows that human reports of ``confidence'' in a lab setting often track a formal statistical notion of confidence.

We believe our assumption of sophistication is a natural starting point.   One could extend our model to allow for specific forms of naivete and see how it impacts the relationship between orderings.

\subsection{Specific Functional Forms} In our framework, we assume that the signal is about the state.  In many other models (e.g., \cite{enke2023cognitive, woodford2020modeling}) each option gets a signal (often independent).  One can map between our setting and a setting where one signal is observed for each options and translate our results.  

More specifically, many of these approaches leverage a Gaussian framework,  where the DM has a normal prior and normal signals.  This does not precisely fit our finite state, finite signal framework.  We extend our analysis to this approach in Appendix \ref{sec:normal} and show that the Gaussian approach satisfies the conditions of Proposition \ref{prop:suffequiv}; the proof of Proposition \ref{prop:suff} then naturally extends.  This allows us to explain why many existing papers which leverage the Gaussian framework prove that randomness, confidence and understanding should co-move.  

As discussed, many of the intuitions for comovement of our orderings emerge from a generalized Fechnerian framework, which can often be interpreted through the lens of a Gaussian framework (see \cite{he2023random,he2024moderate}).   Appendix \ref{sec:appfech} considers generalized Fechnerian models and shows that they satisfy the condition that less randomness comoves with higher expected payoffs (both measured state-by-state), but that attenuation cannot comove with expected payoffs within this framework.

\subsection{Related Work on Information and Belief Elicitation}
The results in this paper connect to the broader literature on elicitation of beliefs.  In a two state environment one can think of both our expected confidence measure as well as choice payoff measure as being generated from piecewise linear proper scoring rules which feature a single kink.  In general, these kinks are at different points, leading to these two measures being distinct (in contrast, our randomness measures cannot be represented as emerging from a proper scoring rule). 

Our work also relates to a large existing literature on signal detection theory used in psychology, neuroscience and more recently computer science.  This literature analyses the relationship between sensitivity, i.e. the accuracy of choice, or $d'$ (\cite{green1966signal}), and meta-sensitivity, i.e., confidence, or meta $d'$ (\citet{fleming2014measure}).  A key question in the literature is whether sensitivity and meta-sensitivity should co-move.  Our results highlight that the relationship between these two measures is sensitive to specific assumptions about the information generating process and payoff structure. Another key construction in the literature is the ROC curve, which captures the tradeoff between incorrectly choosing $x$ when $y$  is better and the opposite across different decision thresholds.  In Appendix \ref{sec:crosschoice} we discuss how our measures relate to the ROC curves, as well as the the Lehmann ordering \cite{lehmann1988comparing}.

\subsection{State Conditional vs Unconditional Data} Our work also highlights and issue that has been underemphasized in the relevant empirical literature: the importance of being clear about, conditional on measuring a specific kind of behavior, whether the measurement is taking place on a state-by-state basis, or an aggregate basis.  For example, even under aligned and neutral shifts we still do not obtain that state-by-state confidence, or aggregate randomness, is equivalent to the choice payoff measure.  Recall that because states are utility realizations, this means a stand needs to be taken on whether the utilities for options are fixed or whether they vary.  The answers to such questions likely depend on whether the data is within person, where it seems plausible a single utility realization happens, or between individuals, each of whom might have a different utility realization. But even in the latter case, it may not be true that there is positive weight ex-post on all utility realizations that the DMs consider plausible ex-ante (i.e. receive weight in their prior); or that the realized distribution over utilities matches the prior.

Moreover, our results imply that if the researcher and the DM have different conceptions of the state space, the researcher may mis-apply our sufficient conditions on the set of Blackwell experiments, because they may believe they are measuring state-by-state, when really they are measuring aggregates of the DM's subjective states.  Thus, aggregate measures may be more robust to analyze.

\subsection{Extensions}
Our approach focuses on complexity of a choice decision.  Much of the literature on the complexity also speaks to the complexity of a specific option.  One could imagine working within our farmework to define what it means for one option to be more/less difficult to understand than other (rather than having the comparison be at the choice set level).

Our framework also specifically focuses on environments where the DM is making a choice.  Many experiments involve subjects choosing a valuation for an option.  We suspect we can extend our approach to allow for valuation, although it is not necessarily more straightforward.

We have also assumed that the DM has access to an exogenous signal.  In contrast, there is a large literature thinking about how DMs endogenously choose the information structure, in response to incentives, as well as the complexity of the problems.  One can imagine extending our negative and positive results to environments where the information acquisition is endogenous.

Given our results it is natural to think about characterization of our model using both the traditional choice data as well as data on confidence (for work doing this with just choice data, see \cite{safonov2017random,safonov2022random} and the discussion in \cite{strzalecki2011axiomatic}).  Given the fact that we can fail to identify which experiment generates a higher payoff for the DM, clearly even with a characterization we will fail to point identify the experiments (see \cite{safonov2017random,safonov2022random} for related discussions).  That said, we could leverage data for set-identification.  Given that we know that randomness and confidence do not need to co-move, observing them co-vary empirically (which often happens) provides information about the structure of subjective Blackwell experiments.  Leveraging this within our framework to identify this structure could provide new information about the way in which individuals process choice problems.

\newpage
{\small
\bibliographystyle{aer}
\bibliography{reference.bib}

}

\section{Appendix: Summary of Orderings}\label{sec:sumorder}

\begin{table}[h!]
\centering
\renewcommand{\arraystretch}{1.8}
\begin{tabular}{lcccc}
\toprule
\textbf{Ordering} & \multicolumn{4}{c}{\textbf{Level of Aggregation}} \\
\cmidrule(lr){2-5}
& \rotatebox{90}{\parbox{3cm}{\centering \textit{Item by Item}\\[-2pt] \textit{and State by State}}}
& \rotatebox{90}{\parbox{3cm}{\centering \textit{State by State,}\\[-2pt] \textit{Averaged Across Items}}}
& \rotatebox{90}{\parbox{3cm}{\centering \textit{Item by Item,}\\[-2pt] \textit{Averaged Across States}}}
& \rotatebox{90}{\parbox{3cm}{\centering \textit{Averaged Across}\\[-2pt] \textit{Items and States}}} \\
\midrule
\textit{Randomness}
  & N/A
  & \begin{tabular}[c]{@{}c@{}}Randomness\\(Section 3.1)\end{tabular}
  & N/A
  & \begin{tabular}[c]{@{}c@{}}Expected Randomness\\(Section 3.1)\end{tabular} \\[6pt]
\hline\\[-8pt]
\textit{Confidence}
  & \begin{tabular}[c]{@{}c@{}}Confidence\\(Section 3.1)\end{tabular}
  &
  & \begin{tabular}[c]{@{}c@{}}Expected Confidence\\(Section 3.1)\end{tabular}
  & \begin{tabular}[c]{@{}c@{}}Overall Confidence\\(Section 5)\end{tabular} \\[6pt]
\hline\\[-8pt]
\textit{Payoffs}
  &
  & \begin{tabular}[c]{@{}c@{}}State Conditional\\Choice Payoffs\\(Section 11.2)\end{tabular}
  &
  & \begin{tabular}[c]{@{}c@{}}Choice Payoffs\\(Section 3.1)\end{tabular} \\
\bottomrule
\end{tabular}
\caption{Ordering and Level of Aggregation}
\label{table:orders_aggregation}
\end{table}
 
 Table \ref{table:orders_aggregation} summarizes the different orderings we use in the paper, broken down by level of aggregation. All the orderings listed are subject to negative results provided in the paper.  Randomness, because it is determined at the choice set level, cannot be done item by item.  Both definitions are provided in Section \ref{sec:def}, and randomness is used as part of a positive result in Proposition \ref{prop:suff}.  For confidence orderings, the basic definitions for confidence and expected confidence are provided in Section \ref{sec:def}, and the latter is used in Proposition \ref{prop:suff}.  Overall confidence is defined in Section \ref{sec:suffpayoff} and used in Proposition \ref{prop:psycho}.  Choice payoffs are defined in Section \ref{sec:def} and used in both Proposition \ref{prop:suff} (for general economic payoffs) and Proposition \ref{prop:psycho} (for psychophysical payoffs).  State conditional choice payoffs are analyzed in Section \ref{sec:statecondpayoffs}.

 The table suggests a notion of confidence not 
explicitly considered in the paper: one could define a 
a notion of confidence  which averages confidence across items within each state but does not aggregate across states.  This ordering sits 
strictly between confidence domination 
($\psi_{\sigma}(x|\omega) \geq \psi_{\sigma'}(x|\omega)$ 
for all $x$ and $\omega$) and expected confidence domination 
($\psi_{\sigma}(x) \geq \psi_{\sigma'}(x)$ for all $x$) 
in terms of strength.  The negative results of 
Proposition~\ref{prop:neg2} are robust to this 
intermediate order: since confidence domination implies 
dominance under $\bar{\psi}_{\sigma}(\omega)$, any 
counterexample that works against the weaker expected 
confidence ordering  still applies.  Similarly, overall confidence 
$\psi_{\sigma}$, which averages across both items and 
states, is the weakest of all these notions and is 
likewise insufficient to overturn the negative results.

The table also suggests item-level payoff 
orderings as alternatives to our orderings involving payoffs. One suggested alternative, which is payoffs item-by-item, state-by-state, collapses to  $u(x|\omega)$ and is independent of $\sigma$ altogether, so it generates no (non-trivial) ordering over experiments.

The other is the expected payoff conditional on choosing item $x$, 
averaging across states, which measures whether the DM tends to choose $x$ in states where $x$ is actually valuable.    This does not circumvent  the negative results: the same examples used for 
Proposition~\ref{prop:neg2} apply, since the decoupling 
between high-stakes and low-stakes states that drives 
the payoff reversals operates through the cardinal 
weighting of states, which neither randomness nor 
confidence controls.  Interestingly, in the psychophysical 
case where $u(x|\omega) \in \{0,1\}$, this item-specific payoff notion coincides with expected confidence.

\section{Appendix: Generalized Fechnerian Models}\label{sec:appfech}

We now show that the generalized Fechnerian approach (see \cite{he2024moderate,he2023random,shubatt2024tradeoffs}) generates data that is consistent with the ability to recover understanding from randomness.  In our restricted setting with only two alternatives, the generalized Fechnerian approach assumes that the choice probabilities depend on the (realized) utility gap between $x$ and $y$, as well as some parameter $\lambda$ which captures the difficulty of the choice situation.  Formally the probability of chooosing $x$ is $P(x) = f(\frac{u(x)- u(y)}{\lambda})$.  We assume that $f$ is 

\begin{enumerate}
    \item strictly increasing
%    \item continuous
%    \item differentiable 
    \item always in the interval $[0,1]$
    \item symmetric: $f(-s) = 1-f(s)$,
\end{enumerate}

Notice that this immediately implies that when $u(x)=u(y)$, $P(x)=\frac{1}{2}$.  These generalized Fechnerian models have no direct role for ``beliefs'', although \cite{he2023random} demonstrate that they can be rewritten in a way consistent with a Bayesian expected utility story under specific assumptions.  However, they immediately imply that the value of a choice problem is directly related to the degree of randomness.  

\begin{proposition}
    In any generalized Fechnerian random choice model $f$ which is i) strictly increasing, ii) always in [0,1] and iii) symmetric: $f(-s) = 1-f(s)$ it is the case that  $$|f(\frac{u(x)- u(y)}{\lambda}) -.5| \geq |f(\frac{u(x)- u(y)}{\lambda'}) -.5| $$ if and only if \begin{eqnarray*}
        && u(x) f(\frac{u(x)- u(y)}{\lambda}) + u(y) (1-f(\frac{u(x)- u(y)}{\lambda}))\\
        &\geq & u(x) f(\frac{u(x)- u(y)}{\lambda'}) + u(y) (1-f(\frac{u(x)- u(y)}{\lambda'})) 
    \end{eqnarray*} 
\end{proposition}

To see this, condition on any $u(x), u(y)$ pair.  Notice that which option is chosen more often is independent of $\lambda$, it only depends on $u(x)- u(y)$.\footnote{If we interpret the choice probabilities through the lens of our Bayesian expected utility environment, this implies that the Blackwell experiment is indicative.}  Moreover, notice that as $\lambda$ increases, the sign of $\frac{u(x)- u(y)}{\lambda}$ stays the same, but values are pushed closer to 0.  Thus, because $f$ is strictly increasing, choice probabilities must get ``more random'' as $\lambda$ increases.  

We can also consider the ``expected payoff'' from the decision problem given the choice probabilities.  This is clearly $u(x) P(x) + u(y) P(y)$.  For any given $u(x) > u(y)$, consider what happens when $\lambda$ increases.  We have just shown that $P(x)$ falls, and $P(y)$ increases; thus, the expected payoff falls; the analogous result is true if $u(y) > u(x)$.  Thus, as $\lambda$ increases, the expected payoff falls.

Thus, we obtain that randomness (state by state, where here a state is just a joint realization of the utility of $x$ and $y$) and the expected payoff from the choice problem (state by state) must co-move (in opposite directions).

In contrast, we now show that attentuation is very badly behaved in generalized Fechnerian models. The most standard definition of attenuation, as related to choice difficulty, is that it is a lower elasticity of behavior with respect to a change in the parameter.  In our setting, we can think of the behavior s the probability of choosing $x$, while the parameter is the utility gap between $x$ and $y$.  For the latter, we will fix the utility of $y$, and so focus on a change in the utility of $x$.  

 We begin by considering the $\lambda$-Luce experiments from Section \ref{sec:examples}.  We first note that fixing $\lambda$, the elasticity of choice with respect to $u(x)$ is non-monotone in $u(x)$.  Let $P(x)$ denote the probability of choosing $x$.  Then the elasticity of $P(x)$ with respect to $u(x)$ is $\epsilon_{P(x),u(x)} = \frac{1}{\lambda} u(x) (1-P(x))$.  The derivative with respect to $\lambda$ is $$\frac{-u(x)}{\lambda^2 (1+ e^{\frac{u(x)-u(y)}{\lambda}})^2}[(1+ e^{\frac{u(x)-u(y)}{\lambda}})-\frac{u(x)-u(y)}{\lambda}e^{\frac{u(x)-u(y)}{\lambda}}] = \frac{-u(x)(1-P(x))}{\lambda^2}[1-\frac{u(x)-u(y)}{\lambda} P(x)].$$  

Notably this does not always have the same sign --- it depends on whether $\frac{u(x)-u(y)}{\lambda}$ is larger or smaller than $\frac{1}{P(x)}$.  Thus, focusing on elasticity in this environment is likely not the best way forward.  

Instead, rather than focusing on elasticity, which also depends on the initial levels of choice probability, we focus only on the change in choice probability as $u(x)$ changes --- i.e. just the derivative.  Thus, in order to ask if choice difficulty (a shift in $\lambda$ alters attenuation, i.e. the derivative of the choice probability with respect to $u(x)$, we focus on the cross-partial of the choice probability with respect to $u(x)$ and $\lambda.$  This is $$\frac{-P(x)(1-P(x))}{\lambda^2}[1+\frac{u(x) - u(y)}{\lambda}(1-2P(x))].$$  
This again can have either sign, and in fact changes sign at exactly  $\kappa$ where $\kappa \tanh{\frac{\kappa}{2}} =1$. This means that increases (or decreases) in $\lambda$ do not uniformly make choice more or less attenuated (in contrast, recall that even in the standard $\lambda$-Luce model, we know that in each state, randomness must decrease).

The same intuition applies more broadly to the generalized Fechnerian framework where choosing $x$ $P(x) = f(\frac{u(x)- u(y)}{\lambda})$.  We assume $f$ is strictly increasing, continuous and differentiable, and always in the interval $[0,1]$.  Now, the cross partial is equal to $\frac{-1}{\lambda^2}\frac{d}{ds}[sf'(s)]$ where $s=\frac{u(x)-u(y)}{\lambda}$.  In fact, the sign of the cross partial will be the reverse of $f'(s)+sf''(s)$. Casual inspection reveals that many other sigmoid functions that we might use for $f$ also have the property that the cross partial cannot always always have the same sign.  Because of this, it seems that attenuation is not a particularly robust behavioral pattern even in the generalized Fechnerian environment.   

In fact, there is a deeper problem.  Consider additionally the requirement of symmetry: $f(-s) = 1-f(s)$.  Under this additional assumption we can never have any generalized Fechnerian random choice model which exhibits a negative cross partial with respect to $u(x)$ and $\lambda$.  

\begin{proposition}
    Consider a generalized Fechnerian random choice model and assume that $f$ is: i)strictly increasing in $\frac{u(x)-u(y)}{\lambda}$, ii) continuous and differentiable, iii) always in the interval [0,1], iv) symmetric: $f(-s)=1-f(s)$.  Then it cannot exhibit a negative cross partial with respect to $u(x)$ and $\lambda$ everywhere.  
\end{proposition}

\begin{proof}
    Recall that we want a function such that $g(s) = sf'(s)$ is strictly increasing everywhere.  By symmetry (and differentiation) $f'(-s)= f'(s)$ so $f'$  is even, implying $g$ is odd (and $g(0)=0$).  Strictly increasing functions that are odd implies that $\lim_{s\rightarrow \infty} g(s) = L \in (0,\infty]$.  This implies there exist a $S$ such that for all $s > S,$ $g(s) \geq \frac{L}{2} $ or $f'(s) \geq \frac{L}{2s}$.  But this implies $\int_{S}^{\infty} f'(s) ds = \infty $.  But recall that $f$ is bounded above by 1.  This is a contradiction.  
\end{proof}

%sigmoid function where the cross partial is always negative (e.g., $f(s) = \frac{1}{2} + C \int_0^s \frac{\hbox{tanh}(t)}{t} dt$, where $C = (2 \int_0^s \frac{\hbox{tanh}(t)}{t} dt)^{-1}$). However, we know of no papers that leverage these kinds of functional forms.

\section{Appendix: Normal Priors and Normally Distributed Signals}\label{sec:normal}

Assumptions about Gaussian priors and noise are often made as parametric assumptions when considering models of learning about the utility of options (e.g., \cite{enke2023cognitive}).  In the body of the paper we restricted ourself to considering finite states along with finite signals for technical ease, which rules out these examples,  Here we discuss this is more detail, and show that the standard approach in the literature satisfies our sufficient conditions.  This demonstrates why these papers find that the different measures should co-move.   

First we describe the standard setting modified in a way to keep it expositionally similar our approach. There are two outcomes $x$ and $y$.  Each outcome has an unknown utility, which is drawn i.i.d. from a normal distribution $N(\mu, z_0)$.  Denote the actual draws as $u(x)$ and $u(y)$.  The decision-maker does not observe the actual utility draws.  They observe two signals, both real numbers, which are independently drawn from two distributions.  One signal ($s_x$)is about the utility of $x$ and is drawn from a normal distribution $N(u(x), \alpha z_1)$; the other ($s_y$) is about the utility of $y$ and is drawn from a normal distribution $N(u(y), \alpha z_1)$.  Here $\alpha$ is a variable that controls the informativeness of the signal, as $\alpha$ goes to 0, the signals become fully informative.  We assume, as does much of the literature, that the signals have equal variance.  

Notice that the state space is the same as in our paper, it is $u(x), u(y)$.  The prior probability over states satisfies our symmetry condition, and is jointly normal.  Conditional on a given state, the distribution of the two signals $s_x, s_y$ is jointly normal.  

Given the structure of the problem, it is clear that the optimal action given two signals is to simply choose the signal that has a higher value; in other words the individual believes $x$ is the optimal choice if and only if $s_x \geq s_y$ (and $y$ is optimal if and only if $s_y \geq s_x$).  Moreover, our analysis is simplified because the set of states where $x$ and $y$ have the same utility is measure 0, as is the set of signals where $s_x = s_y$.  

\begin{proposition}\label{prop:norm}
Let $\alpha' \leq \alpha''$ and $u(x) \geq u(y)$. Then $Prob(s_x > s_y|\alpha') \geq Prob(s_x > s_y|\alpha'')$. 
\end{proposition}

\begin{proof}
    Define $ \tilde{s} = s_x - s_y$. Given that the signals are drawn independently, we have
\[
\tilde{s} \sim N(u(x) - u(y), 2\alpha z_1).
\]
Then
\[
\Pr(s_x > s_y) = \Pr(\tilde{s} > 0) = \Phi \left(\frac{u(x) - u(y)}{\sqrt{2\alpha z_1}}\right),
\]
where $\Phi$ denotes the standard normal CDF. Since $u(x) > u(y)$, the argument inside $\Phi$ is strictly positive. As $\alpha$ decreases, $\sqrt{2\alpha z_1}$ decreases, so the fraction inside $\Phi$ increases. Because $\Phi$ is strictly increasing, $\Pr(s_x > s_y)$ increases. \end{proof}
Given Proposition \ref{prop:suffequiv} this implies that we can go from the signal structure with the higher $\alpha$ to the lower $\alpha$ through a series of aligned and neutral shifts.  The proof of Proposition \ref{prop:suff} naturally extends to the continuous environment.  

\section{Additional Orderings}\label{sec:addorder}

\subsection{Blackwell Informativeness}\label{sec:blackwell}
We can also consider a weaker order than our notion of choice payoff domination to capture understanding.  A standard notion of understanding induced by an experiment is the Blackwell order.  There are numerous discussions of how to characterize Blackwell more informative experiments (\cite{blackwell1951comparison,grant1998}), see those references for details.  We remind readers of the intuition: fixing a state space, consider a menu of actions $A$ such that each action maps a state to a payoff.  Agents, after receiving a signal from an experiment, take an action from the menu.  If for any menu of actions one experiment $\sigma$ generates a higher expected payoff than another $\sigma'$ then $\sigma$ Blackwell dominates $\sigma'$.  Recall that this is a partial order.  

\begin{proposition}
    The results of Proposition \ref{prop:neg2} hold if ``choice payoff dominates'' is replaced by ``Blackwell dominates''.  However, if $\sigma$ can be be obtained from $\sigma'$ by a series of single aligned shifts, $\sigma$ may not Blackwell dominate $\sigma'$.
\end{proposition}

The first part of the proposition should not necessarily be that surprising.  In Examples \ref{ex:3states_2}, \ref{ex:secondcounter} and \ref{ex:binarcounter} which prove Proposition  \ref{prop:neg2} choice payoff domination is the same as Blackwell dominance, and so the results go through.  

Moreover, just because we can move from one experiment ($\sigma'$) to another ($\sigma$) through aligned and neutral shifts does not mean that $\sigma$ Blackwell dominates $\sigma'$.  Notice that by construction $\sigma'$ cannot Blackwell dominate $\sigma$.  However, choice payoff dominance only considers whether, in a given state, the choice aligns with which of the binary options has a higher utility.  Thus, a shift which makes it harder to distinguish between two states with the same ordinal ranking does not reduce payoffs.  Thus, aligned and neutral shifts can make it easier to distinguish between some states, but harder to distinguish between others.  The latter, for some other decision problem, may be the payoff relevant information.  This leads to an inability to rank via Blackwell dominance, because there could be a decision problem for which distinguishing these two states is important (and distinguishing between all others is unimportant).   In other words, for the decision problem currently faced by the DM, aligned and neutral shifts increase payoffs, but there could be other decision problems for which they reduce payoffs.

\subsection{State Conditional Choice Payoffs}\label{sec:statecondpayoffs}
An alternative weakening of choice payoff dominance would be state-conditional payoff dominance.   Recall the expected utility given $\sigma$, conditional on state $\omega$ is $$W(\sigma|\omega) = \sum_{s \in S} \sigma(s|\omega) \sum_x c_{\sigma}(x|s) u(x|\omega). $$

\begin{defn} Let $\sigma$ and $\sigma'$ be two experiments. We say $\sigma$ state-conditional choice payoff dominates $\sigma'$ if $W(\sigma|\omega) \geq W(\sigma'|\omega)$ for every $\omega.$
\end{defn}

Like Blackwell, this is a partial order.  However, it is a distinct partial order from Blackwell and neither implies nor is implied by Blackwell.

Notice however, state-conditional choice payoff dominance implies payoff dominance, but not the other way around.  This is because payoff dominance may occur when payoffs in one state go up, another go down, but the gains from the former outweigh the losses from the latter.  

Because of the fact that state conditional choice payoff dominance requires a comparison for each state, we obtain a slightly weaker negative result compared to Proposition \ref{prop:neg2}.

\begin{proposition} ~
\begin{enumerate}
\item If $\sigma$ is less random (or less expected random) than, and confidence dominates $\sigma'$, then $\sigma'$ can state-contingent choice payoff dominates~$\sigma$.  However, if $\sigma$ is less random than $\sigma'$ then expected payoffs must be higher under at least one state in $\sigma$ compared to $\sigma'$  
\item   If $\sigma$ is less random (or expected less random) than, and state conditional choice payoff dominates $\sigma'$ then $\sigma$ need not confidence dominate $\sigma'$.  However, if in addition, for all $\omega \in \Omega(x,y)$ and for all $s$, $\sigma(s|\omega) = \sigma'(s|\omega)$ then $\sigma$ can be obtained from $\sigma'$ by a series of single aligned shifts, and so $\sigma$ is more expected confident than $\sigma'$.
\item The results of Proposition \ref{prop:neg2} Part (3) hold if ``choice payoff dominates'' is replaced by ``state conditional choice payoff dominates''.
\end{enumerate}
\end{proposition}

To see why the first part of (1) is true, simply look at Example \ref{ex:3states_2}, where we observed the expected payoff (across all states) goes down.  Thus, payoff in at least one state must have gone down.  

The second part of (1) says that the converse failure also cannot occur: $\sigma'$ cannot state-contingent payoff dominate~$\sigma$ either. In fact, it is a stronger result that requires only less random (not confidence domination).  We prove this by beginning with a simple lemma of independent interest.

\begin{lemma}\label{lem:indicative}
Every Bayesian-optimal experiment $\sigma$ is indicative in at least one state.\footnote{As the proof will demonstrate this requires key features of our environment: symmetric priors and binary options.}
\end{lemma}

\begin{proof}
For each signal $s$ where option $x$ is (weakly) chosen, Bayesian optimality requires
\[
\sum_{\omega \in \Omega} \pi(\omega)\, \sigma(s|\omega)\,[u(x|\omega) - u(y|\omega)] \geq 0.
\]
Summing over all $x$-inducing signals $s \in \hat{S}_\sigma(x)$:
\begin{equation}\label{eq:bayesian-sum}
\sum_{\omega \in \Omega} \pi(\omega)\, \rho_\sigma(x|\omega)\,[u(x|\omega) - u(y|\omega)] \geq 0.
\end{equation}

Now pair symmetric states. For each $\omega = (u_1, u_2) \in \Omega(x)$ with utility gap $\Delta(\omega) = u_1 - u_2 > 0$, there is a symmetric partner $\omega' = (u_2, u_1) \in \Omega(y)$ with $\pi(\omega') = \pi(\omega)$ and $u(x|\omega') - u(y|\omega') = -\Delta(\omega)$. Substituting into \eqref{eq:bayesian-sum}:
\begin{equation}\label{eq:paired}
\sum_{\text{pairs} (\omega, \omega')} \pi(\omega)\, \Delta(\omega) \left[\rho_\sigma(x|\omega) - \rho_\sigma(x|\omega')\right] \geq 0.
\end{equation}
Since $\pi(\omega) > 0$ and $\Delta(\omega) > 0$ for every pair, there must exist at least one pair $(\omega^*, \omega^{*\prime})$ with
\begin{equation}\label{eq:pair-ineq}
\rho_\sigma(x|\omega^*) \geq \rho_\sigma(x|\omega^{*\prime}).
\end{equation}
We now show $\sigma$ is indicative in at least one state of this pair.

\smallskip
\noindent\textbf{Case 1:} $\rho_\sigma(x|\omega^*) \geq \tfrac{1}{2}$. Then $x$ is the (weakly) most-chosen option in the $x$-optimal state~$\omega^*$, so $\sigma$ is indicative in~$\omega^*$.

\smallskip
\noindent\textbf{Case 2:} $\rho_\sigma(x|\omega^*) < \tfrac{1}{2}$. By \eqref{eq:pair-ineq}, $\rho_\sigma(x|\omega^{*\prime}) \leq \rho_\sigma(x|\omega^*) < \tfrac{1}{2}$. Hence $\rho_\sigma(y|\omega^{*\prime}) > \tfrac{1}{2}$, and since $\omega^{*\prime} \in \Omega(y)$, the correct option~$y$ is chosen more than half the time in~$\omega^{*\prime}$---so $\sigma$ is indicative in~$\omega^{*\prime}$.
\end{proof}

This leads to the following corollary\footnote{Again, this requires the symmetry and binary choice aspects of our environment.}

\begin{corollary}\label{cor:positive}
If $\sigma$ is less random than~$\sigma'$ state-by-state, then there exists a state $\omega^*$ with $W(\sigma|\omega^*) \geq W(\sigma'|\omega^*)$.
\end{corollary}

\begin{proof}
By Lemma~\ref{lem:indicative}, $\sigma$ is indicative in some state~$\omega^*$. Without loss, suppose $\omega^* \in \Omega(x)$, so $u(x|\omega^*) > u(y|\omega^*)$ and $\rho_\sigma(x|\omega^*) \geq \tfrac{1}{2}$.

Since $\sigma$ is less random:
\[
\rho_\sigma(x|\omega^*) = \max\{\rho_\sigma(x|\omega^*),\, 1 - \rho_\sigma(x|\omega^*)\} \geq \max\{\rho_{\sigma'}(x|\omega^*),\, 1 - \rho_{\sigma'}(x|\omega^*)\} \geq \rho_{\sigma'}(x|\omega^*).
\]
Since $x$ is optimal in $\omega^*$ and $u(x|\omega^*) - u(y|\omega^*) > 0$:
\begin{align*}
W(\sigma|\omega^*) &= u(y|\omega^*) + \rho_\sigma(x|\omega^*)\,[u(x|\omega^*) - u(y|\omega^*)] \\
&\geq u(y|\omega^*) + \rho_{\sigma'}(x|\omega^*)\,[u(x|\omega^*) - u(y|\omega^*)] = W(\sigma'|\omega^*). \qedhere
\end{align*}
\end{proof}

Less random alone, combined with Bayesian optimality and symmetric priors, guarantees that payoffs improve in at least one state. The key mechanism is that Bayesian optimality prevents an experiment from being anti-indicative (choosing the wrong option by majority) in \emph{every} state. Intuitively, signals that cause the DM to choose~$x$ must, on a prior-weighted average, carry more evidence that $x$ is optimal---so they cannot systematically point in the wrong direction in every state simultaneously. In whichever state the experiment is indicative, less random directly translates into higher payoffs.

Thus, from Part (1) we learn that if $\sigma$ is less random and confidence dominates~$\sigma'$, the state-contingent payoff comparison is indeterminate in both directions.

Part (2) has two parts.  To prove the first part we construct an example in which $\sigma$ is less random than (and expected less random than) $\sigma'$, $\sigma$ state-conditional choice payoff dominates $\sigma'$, yet $\sigma$ does not confidence dominate~$\sigma'$.

\begin{example}\label{ex:statecon}
There are two options $\{x,y\}$ and three states:
$\omega_1 = (u_H, u_L)$, $\omega_2 = (u_L, u_H)$, $\omega_3 = (c,c)$ with $u_H > u_L$.  The prior is symmetric with $\pi(\omega_1) = \pi(\omega_2) = p$ and $\pi(\omega_3) = 1-2p$ for some $p \in (0, 1/2)$.  Note that $\omega_3 \in \Omega(x,y)$: the DM is indifferent between $x$ and $y$ in this state.

There are two signals $S = \{s_1, s_2\}$.  Define experiments $\sigma$ and $\sigma'$ that agree on the non-tied states:
\[
\sigma(s_1|\omega_1) = \sigma'(s_1|\omega_1) = \theta, \qquad
\sigma(s_1|\omega_2) = \sigma'(s_1|\omega_2) = 1-\theta,
\]
with $\theta > \frac{1}{2}$.  In the tied state, they differ:
\[
\sigma'(s_1|\omega_3) = \tfrac{1}{2}, \qquad
\sigma(s_1|\omega_3) = \tfrac{1}{2} + \varepsilon,
\]
for some $\varepsilon \in (0,\frac{1}{2})$.

Under both experiments, the DM chooses $x$ after $s_1$ and $y$ after $s_2$.  To verify: the expected utility difference $\mathbb{E}[u(x) - u(y) \mid s_1]$ is proportional to $\pi(\omega_1)\sigma(s_1|\omega_1)(u_H - u_L) - \pi(\omega_2)\sigma(s_1|\omega_2)(u_H - u_L) = p(u_H - u_L)[\theta - (1-\theta)] > 0$ (the tied state contributes nothing since $u(x|\omega_3) = u(y|\omega_3)$), and symmetrically for $s_2$.

In $\omega_1$ and $\omega_2$ the experiments are identical.  In $\omega_3$, $\rho_\sigma(x|\omega_3) = \frac{1}{2} + \varepsilon > \frac{1}{2} = \rho_{\sigma'}(x|\omega_3)$, so $\max_x \rho_\sigma(x|\omega_3) > \max_x \rho_{\sigma'}(x|\omega_3)$.  Thus $\sigma$ is less random than $\sigma'$ state-by-state.  Moreover,
\[
\max_x \rho_\sigma(x) = p\theta + p(1-\theta) + (1-2p)\bigl(\tfrac{1}{2}+\varepsilon\bigr) = p + (1-2p)\bigl(\tfrac{1}{2}+\varepsilon\bigr),
\]
which strictly exceeds $p + (1-2p)\frac{1}{2} = \max_x \rho_{\sigma'}(x)$, so $\sigma$ is also expected less random.

In $\omega_1$ and $\omega_2$ the experiments are identical, so $W(\sigma|\omega_i) = W(\sigma'|\omega_i)$ for $i = 1,2$.  In $\omega_3$, since $u(x|\omega_3) = u(y|\omega_3) = c$, the expected payoff equals $c$ regardless of the signal structure:
\[
W(\sigma|\omega_3) = \sigma(s_1|\omega_3)\, c + \sigma(s_2|\omega_3)\, c = c = W(\sigma'|\omega_3).
\]
Thus $\sigma$ state-conditional choice payoff dominates $\sigma'$ (with equality in every state).

Since each signal uniquely determines the chosen option, the confidence of item $x$ conditional on any state $\omega$ reduces to $\psi_\sigma(x|\omega) = \pi_\sigma(\hat{\Omega}(x) \mid s_1)$, where $\hat{\Omega}(x) = \{\omega_1, \omega_3\}$.  Similarly, $\psi_\sigma(y|\omega) = \pi_\sigma(\hat{\Omega}(y) \mid s_2)$, where $\hat{\Omega}(y) = \{\omega_2, \omega_3\}$.

For option $x$:
\[
\pi_\sigma(\hat{\Omega}(x)|s_1) = \frac{p\theta + (1-2p)(\frac{1}{2}+\varepsilon)}{p\theta + p(1-\theta) + (1-2p)(\frac{1}{2}+\varepsilon)}, \qquad
\pi_{\sigma'}(\hat{\Omega}(x)|s_1) = \frac{p\theta + (1-2p)\frac{1}{2}}{p\theta + p(1-\theta) + (1-2p)\frac{1}{2}}.
\]
The shift from $\sigma'$ to $\sigma$ adds $d = (1-2p)\varepsilon > 0$ to both the numerator and denominator.  Writing the $\sigma'$ ratio as $\frac{A+C}{A+B+C}$ where $A = p\theta$, $B = p(1-\theta)$, $C = (1-2p)/2$, we have
\[
\frac{A+C+d}{A+B+C+d} > \frac{A+C}{A+B+C}
\]
since cross-multiplying yields $(A+C+d)(A+B+C) - (A+C)(A+B+C+d) = dB > 0$.  So confidence in $x$ \emph{rises} under $\sigma$.

For option $y$:
\[
\pi_\sigma(\hat{\Omega}(y)|s_2) = \frac{p\theta + (1-2p)(\frac{1}{2}-\varepsilon)}{p(1-\theta) + p\theta + (1-2p)(\frac{1}{2}-\varepsilon)}, \qquad
\pi_{\sigma'}(\hat{\Omega}(y)|s_2) = \frac{p\theta + (1-2p)\frac{1}{2}}{p(1-\theta) + p\theta + (1-2p)\frac{1}{2}}.
\]
Now the shift from $\sigma'$ to $\sigma$ \emph{subtracts} $d = (1-2p)\varepsilon > 0$ from both the numerator and denominator.  By the symmetric argument, this \emph{decreases} the ratio: confidence in $y$ \emph{falls} under~$\sigma$.

Since $\psi_\sigma(x|\omega) > \psi_{\sigma'}(x|\omega)$ but $\psi_\sigma(y|\omega) < \psi_{\sigma'}(y|\omega)$ for all $\omega$, neither experiment confidence dominates the other. $\Box$
\end{example}

Intuitively, because the shift occurs only in the tied state $\omega_3$, it cannot affect payoffs in any state.  But the asymmetric reallocation of signal probability in $\omega_3$ increases the posterior that $x$ is optimal given $s_1$ while decreasing the posterior that $y$ is optimal given $s_2$.  The tied state acts as ``neutral mass'' that can be shifted between signals to affect confidence without affecting payoffs or worsening randomness.

We now prove the second sentence of Part~(2).  Since $\sigma(s|\omega) = \sigma'(s|\omega)$ for all $\omega \in \Omega(x,y)$ and all $s$, the two experiments differ only in non-tied states.

Consider any non-tied state $\omega \in \Omega(k)$ for $k \in \{x,y\}$, so that option $k$ is strictly optimal: $u(k|\omega) > u(-k|\omega)$.  State-conditional choice payoff dominance gives $W(\sigma|\omega) \geq W(\sigma'|\omega)$.  Since
\[
W(\sigma|\omega) = u(-k|\omega) + \rho_\sigma(k|\omega)\bigl[u(k|\omega) - u(-k|\omega)\bigr],
\]
and $u(k|\omega) - u(-k|\omega) > 0$, this is equivalent to 
\[
\rho_\sigma(k|\omega) \geq \rho_{\sigma'}(k|\omega).
\]
That is, the correct option is chosen with weakly higher probability under $\sigma$ in every non-tied state.

Now observe that $\rho_\sigma(k|\omega)$ can be decomposed as
\[
\rho_\sigma(k|\omega) = \sum_{s \in S_\sigma(k)} \sigma(s|\omega) + \frac{1}{2}\sum_{s \in S_\sigma(\{x,y\})} \sigma(s|\omega),
\]
and similarly for $\sigma'$.  The condition $\rho_\sigma(k|\omega) \geq \rho_{\sigma'}(k|\omega)$ means that, in state $\omega$, the effective weight on correct-choice signals is weakly higher under $\sigma$.  By the same constructive argument as in Proposition \ref{prop:suffequiv}, the change from $\sigma'$ to $\sigma$ in each non-tied state $\omega$ can be decomposed into: (i)~aligned shifts, which move probability weight from signals inducing the wrong choice toward signals inducing the correct choice and (ii)~neutral shifts, which rearrange weight among signals inducing the same choice. Since no changes occur in tied states, $\sigma$ can be obtained from $\sigma'$ by a sequence of aligned and neutral shifts. Proposition \ref{prop:suff} then implies that $\sigma$ expected confidence dominates $\sigma'$.

Part (3) of the Proposition follows immediately from Example \ref{ex:binarcounter}.  

%Thus, in some sense, observing confidence and randomness allows for ``more'' recovery of state conditional choice payoff dominance than simply choice payoff dominance.

\section{Cross Choice Set Comparisons}\label{sec:crosschoice}

Our analysis so far has fixed the binary choice problem---the options $x$ and $y$, the state space $\Omega$, and the utility function $u$---and varied only the experiment $\sigma$.  A natural question is whether (and how) our results extend to settings where the \emph{choice problem itself} varies alongside the experiment.  This is relevant in many applications; there are many environments where one cannot change the difficulty of the choice problem without changing the options themselves.  It is clear that our negative results still hold when we extend to comparisons across choice sets (since they held when both choice sets were the same).  The primary question is whether there is a natural way to extend the sufficient conditions.

In this section we develop a framework for such cross-problem comparisons.  A \textbf{binary choice problem} is a tuple $\mathcal{P} = (\Omega, u, \pi, \sigma)$ where $\Omega$ is the state space, $u: \{x,y\} \times \Omega \to \mathbb{R}$ is the utility function, $\pi$ is a symmetric prior over $\Omega$, and $\sigma$ is an experiment with signal space $S$.  In order to simplify our analysis we will assume that tie states have measure zero:
\[
\pi\bigl(\{\omega \in \Omega : u(x|\omega) = u(y|\omega)\}\bigr) = 0,
\]
and that there are no tie signals: $S_\sigma(x,y) = \emptyset$.

Consider two binary choice problems $\mathcal{P}^1 = (\Omega^1, u^1, \pi^1, \sigma^1)$ and $\mathcal{P}^2 = (\Omega^2, u^2, \pi^2, \sigma^2)$.  Each has its own state space, utility function, prior, and experiment.  The DM in problem $i$ observes a signal from $\sigma^i$, updates via Bayes' rule using prior $\pi^i$, and chooses the option with higher expected utility.  Our measures have been defined within a problem; we want to now understand whether we can compare them across problems.

We focus on extending our sufficient conditions on the set of experiments.  Under psychophysical payoffs, because states always only generate payoffs of 1 or 0, the extension is much more straightforward.  

Two issues arise.  First, there are the definitions of the orderings themselves.  It is not obvious how definitions stated in terms of objects of the model (e.g., state contingent measurements) extend.  Some of our measurements, like expected randomness, expected confidence, and psychophysical payoffs are dimensionless and can be compared across problems.  However, ``economic understanding'' $\Wecon(\sigma^i)$ depends on the utility function $u^i$, which differs across problems.  Thus, two problems measure understanding in different ``units.''\footnote{One could normalize by, say, the full-information payoff minus the no-information payoff, yielding a ``fraction of value of full information'' in $[0,1]$.  But even this normalized ordering depends on the specific utility function, not just ``how much'' information is provided.}

Second, our definition of aligned shifts relies on moving probability mass around for a fixed problem, specifically for a fixed state.  With different state spaces, aligned shifts have no bite.  We need a different way of comparing experiments across choice sets.

\subsection{Aligned Dominance}

The approach we take is to replace aligned shifts with a state-level coupling across the problems that captures the same intuition: at each matched pair of states across the two problems, the probability of making the correct choice is (weakly) higher in one problem than in the other.

\begin{defn}\label{def:aligned-dominance}
Problem $\mathcal{P}^2$ \textbf{aligned-dominates} problem $\mathcal{P}^1$ if there exists a coupling $\mu$ on $\Omega^1 \times \Omega^2$---a joint probability measure with marginals $\pi^1$ and $\pi^2$---satisfying:
\begin{enumerate}
\item \textbf{Ordinal agreement.} The coupling is concentrated on states that agree on which option is correct:
\[
\mu\Bigl(\bigl(\widehat{\Omega}^1(x) \times \widehat{\Omega}^2(x)\bigr) \cup \bigl(\widehat{\Omega}^1(y) \times \widehat{\Omega}^2(y)\bigr)\Bigr) = 1.
\]
\item \textbf{Correct-choice improvement.} For $\mu$-a.e.\ $(\omega^1, \omega^2) \in \widehat{\Omega}^1(x) \times \widehat{\Omega}^2(x)$:
\[
\sum_{s^2 \in S_{\sigma^2}(x)} \sigma^2(s^2|\omega^2) \geq \sum_{s^1 \in S_{\sigma^1}(x)} \sigma^1(s^1|\omega^1),
\]
and the symmetric condition for $\mu$-a.e.\ $(\omega^1, \omega^2) \in \widehat{\Omega}^1(y) \times \widehat{\Omega}^2(y)$.
\end{enumerate}
\end{defn}

The coupling $\mu$ matches states across the two problems.  Ordinal agreement ensures that matched states agree on which option is better; the comparison is between an $x$-correct state in one problem and an $x$-correct state in the other.  Correct-choice improvement requires that, at each matched pair, the DM in problem 2 is at least as likely to choose correctly as the DM in problem 1.

The coupling is not unique.  The definition asks for the \emph{existence} of some coupling satisfying both conditions.  The ordinal agreement condition has real content: since both priors are symmetric, $\pi^i(\widehat{\Omega}^i(x)) = \pi^i(\widehat{\Omega}^i(y)) = 1/2$ for $i = 1,2$.  The coupling must respect these marginals while pairing $x$-correct states in one problem with $x$-correct states in the other (and vice versa).  This is always feasible (any product of the conditional measures works), but then imposing correct-choice improvement in addition imposes substantive restrictions.

\subsection{Coupled Less Random}

We also introduce a cross-problem analog of the state-by-state less-random notion from Definition~\ref{wchoice}.  Within a single problem, less random requires that at every state the DM is at least as decisive in one experiment as in the other.  Across problems, with different state spaces, the natural extension again leverages coupling. We match states across the two problems and require that the DM is at least as less random at each matched pair.

\begin{defn}\label{def:coupled-less-random}
Problem $\mathcal{P}^2$ is \textbf{coupled less random} than problem $\mathcal{P}^1$ if there exists a coupling $\mu$ on $\Omega^1 \times \Omega^2$ with marginals $\pi^1$ and $\pi^2$ satisfying ordinal agreement and, for $\mu$-a.e.\ $(\omega^1, \omega^2)$:
\[
\max_a \rho^2(a|\omega^2) \geq \max_a \rho^1(a|\omega^1).
\]
\end{defn}

As with aligned dominance, coupled less random is defined via the existence of a coupling; the coupling may or may not be the same as the one witnessing aligned dominance.

\subsection{Cross-Problem Comovement}

Before stating our cross-problem result, recall that for any single problem with a symmetric prior and no tie signals, $\psi = \Wpsych(\sigma)$: overall confidence \emph{equals} the psychophysical payoff.  This identity depends only on the internal structure of each problem and requires no relationship between problems.  Thus, for any two binary choice problems $\mathcal{P}^1$ and $\mathcal{P}^2$, each with a symmetric prior and no tie signals,
\[
\Wpsych(\sigma^2) \geq \Wpsych(\sigma^1) \quad \Longleftrightarrow \quad \psi^2 \geq \psi^1.
\]
The two problems may have different state spaces, utility functions, priors, and experiments, and no relationship between them is needed.

We now show that aligned dominance guarantees psychophysical payoffs (i.e., confidence) are higher in one problem than in another, and, under indicativeness, that randomness (in the coupled sense) comoves.

\begin{proposition}[Cross-Problem Comovement]\label{prop:cross-comovement}
Suppose $\mathcal{P}^2$ aligned-dominates $\mathcal{P}^1$.  Then:
\begin{enumerate}
\item $\mathcal{P}^2$ has (weakly) higher psychophysical payoffs (equivalently, higher overall confidence): $\Wpsych(\sigma^2) \geq \Wpsych(\sigma^1)$.
\item If $\sigma^1$ and $\sigma^2$ are both indicative, $\mathcal{P}^2$ is coupled less random than $\mathcal{P}^1$.
\end{enumerate}
\end{proposition}

\begin{proof}
Write $P^i(\mathrm{correct}|\omega)$ for the probability of choosing correctly in problem $i$ at state $\omega$:
\[
P^i(\mathrm{correct}|\omega) = \begin{cases} \sum_{s \in S_{\sigma^i}(x)} \sigma^i(s|\omega) & \text{if } \omega \in \widehat{\Omega}^i(x),\\ \sum_{s \in S_{\sigma^i}(y)} \sigma^i(s|\omega) & \text{if } \omega \in \widehat{\Omega}^i(y). \end{cases}
\]
Let $\mu$ be the coupling under which  aligned dominance holds.

\medskip\noindent\emph{Part (1).}  With no ties, the psychophysical payoff in problem $i$ is
\[
\Wpsych(\sigma^i) = \sum_{\omega \in \Omega^i} \pi^i(\omega)\, P^i(\mathrm{correct}|\omega) = \int_{\Omega^1 \times \Omega^2} P^i(\mathrm{correct}|\omega^i)\, d\mu(\omega^1, \omega^2),
\]
where the second equality uses that $\mu$ has marginals $\pi^1$ and $\pi^2$.  By ordinal agreement, $\mu$ is concentrated on pairs $(\omega^1,\omega^2)$ in
$(\widehat{\Omega}^1(x)\times\widehat{\Omega}^2(x))\cup(\widehat{\Omega}^1(y)\times\widehat{\Omega}^2(y))$,
so the correct option is the same in both problems at every coupled pair. By correct-choice improvement, $P^2(\mathrm{correct}|\omega^2) \geq P^1(\mathrm{correct}|\omega^1)$ pointwise $\mu$-a.e., so $\Wpsych(\sigma^2) \geq \Wpsych(\sigma^1)$.

\medskip\noindent\emph{Part (2).}  Under indicativeness, the most-chosen option at each state is the correct one, so $\max_a \rho^i(a|\omega) = P^i(\mathrm{correct}|\omega)$.  Correct-choice improvement under $\mu$ therefore implies $\max_a \rho^2(a|\omega^2) \geq \max_a \rho^1(a|\omega^1)$ for $\mu$-a.e.\ $(\omega^1, \omega^2)$.  The same coupling $\mu$ (which already satisfies ordinal agreement) thus also induces coupled less random.
\end{proof}

Proposition~\ref{prop:cross-comovement} is silent on economic payoffs.  This is deliberate.  As discussed, $\Wecon(\sigma^1)$ and $\Wecon(\sigma^2)$ are measured in different utility units and are not directly comparable.  One route would be to impose additional structure on the utility functions through the coupling---for instance, requiring that utility gaps $|u^2(x|\omega^2) - u^2(y|\omega^2)| \geq |u^1(x|\omega^1) - u^1(y|\omega^1)|$ for $\mu$-a.e.\ $(\omega^1, \omega^2)$.  We do not pursue this route.  Instead, in the next subsection we look for a different way to extend the notion of economic understanding---one that sidesteps the units problem by working at the level of the signal structure rather than realized payoffs.

\subsection{Informational Aligned Dominance and ROC Dominance}

An alternative route to comparing ``payoffs'' across problems is to ask how informative each problem's signals are about the ordinal question of which option is correct (in other words we are moving away from cardinal measurement of payoffs and towards an ordinal one).  If problem 2's signals are uniformly more informative than problem 1's on this binary question, then for any likelihood ratio threshold decision rule problem 2 delivers weakly better outcomes.  This sidesteps the units problem: we do not compare realized utilities, we compare the quality of the information itself.

To formalize this, observe that each (problem, experiment) pair $\mathcal{P}^i$ induces a binary hypothesis testing problem: $H_x^i$ (``$x$ is correct'') vs.\ $H_y^i$ (``$y$ is correct'').  The aggregate signal densities under each hypothesis are
\[
f_x^i(s) = \sum_{\omega \in \widehat{\Omega}^i(x)} 2\pi^i(\omega)\, \sigma^i(s|\omega), \qquad f_y^i(s) = \sum_{\omega \in \widehat{\Omega}^i(y)} 2\pi^i(\omega)\, \sigma^i(s|\omega),
\]
and the likelihood ratio $\Lambda^i(s) = f_x^i(s)/f_y^i(s)$ summarizes the evidence that signal $s$ provides for $x$ over $y$.  The ROC curve of the induced test is obtained by varying a threshold $\tau$ on $\Lambda^i$ and plotting the resulting true positive rate (probability of accepting $H_x$ under $H_x$) against the false positive rate (probability of accepting $H_x$ under $H_y$).  The curve characterizes the trade-off between true and false positive rates across all possible likelihood-ratio tests.

To state the strengthened condition, define for each signal $s$ in problem $i$ the posterior probability that $x$ is correct:
\[
e^i(s) = \pi^i_\sigma\bigl(\widehat{\Omega}^i(x) \mid s\bigr) = \frac{f_x^i(s)}{f_x^i(s) + f_y^i(s)}.
\]
This is a monotone transformation of the likelihood ratio $\Lambda^i(s)$, so first-order stochastic dominance conditions on $e^i$ and on $\Lambda^i$ are equivalent.  While $e^i(s)$ does not in general govern the DM's choice directly, that is given by the expected payoffs under the posterior, it is the natural summary statistic about which option is ordinally correct.

\begin{defn}\label{def:info-aligned-dominance}
Problem $\mathcal{P}^2$ \textbf{informationally aligned-dominates} problem $\mathcal{P}^1$ if there exists a coupling $\mu$ on $\Omega^1 \times \Omega^2$ with marginals $\pi^1$ and $\pi^2$ satisfying ordinal agreement and:
\begin{enumerate}
\setcounter{enumi}{1}
\item \textbf{Stochastic evidence improvement.} For $\mu$-a.e.\ $(\omega^1, \omega^2) \in \widehat{\Omega}^1(x) \times \widehat{\Omega}^2(x)$ and all $t \in [0,1]$:
\[
\Pr\bigl(e^2(s) \geq t \mid \omega^2\bigr) \geq \Pr\bigl(e^1(s) \geq t \mid \omega^1\bigr),
\]
and for $\mu$-a.e.\ $(\omega^1, \omega^2) \in \widehat{\Omega}^1(y) \times \widehat{\Omega}^2(y)$ and all $t \in [0,1]$:
\[
\Pr\bigl(e^2(s) \leq t \mid \omega^2\bigr) \geq \Pr\bigl(e^1(s) \leq t \mid \omega^1\bigr).
\]
\end{enumerate}
\end{defn}

Informational aligned dominance says that at each coupled pair of $x$-correct states, the distribution of evidence for $x$ (in the ordinal sense of how likely is $x$ to be the better option) is stochastically stronger in Problem~2.  This controls not just the total probability of the signal landing on one side or the other, but the probability of reaching every level of evidential strength.

\begin{proposition}\label{prop:info-aligned-ROC}
Suppose $\mathcal{P}^2$ informationally aligned-dominates $\mathcal{P}^1$.  Then:
\begin{enumerate}
\item The induced binary hypothesis test of $\mathcal{P}^2$ ROC-dominates that of $\mathcal{P}^1$: at every threshold $t$, the true positive rate is weakly higher and the false positive rate is weakly lower in problem~2.
\item Equivalently, $\mathcal{P}^2$ Blackwell-dominates $\mathcal{P}^1$ on the binary hypothesis problem $H_x$ vs.\ $H_y$.
\item $\mathcal{P}^2$ has weakly higher psychophysical payoffs than $\mathcal{P}^1$, $\Wpsych(\sigma^2) \geq \Wpsych(\sigma^1)$, and equivalently higher overall confidence, $\psi^2 \geq \psi^1$.
\end{enumerate}
\end{proposition}

\begin{proof}
Write $P^i(\mathrm{correct}|\omega)$ for the probability of choosing correctly in problem $i$ at state $\omega$:
\[
P^i(\mathrm{correct}|\omega) = \begin{cases} \sum_{s \in S_{\sigma^i}(x)} \sigma^i(s|\omega) & \text{if } \omega \in \widehat{\Omega}^i(x),\\ \sum_{s \in S_{\sigma^i}(y)} \sigma^i(s|\omega) & \text{if } \omega \in \widehat{\Omega}^i(y). \end{cases}
\]
Let $\mu$ be the coupling under which informational aligned dominance holds.

\medskip\noindent\emph{Part (1).}
The true positive rate at threshold $t$ is
\[
\mathrm{TPR}^i(t)
  = \sum_{\omega \in \widehat{\Omega}^i(x)} 2\pi^i(\omega)\,\Pr\bigl(e^i(s)\geq t\mid\omega\bigr)
  = \int_{\widehat{\Omega}^1(x)\times\widehat{\Omega}^2(x)}
      2\,\Pr\bigl(e^i(s)\geq t\mid\omega^i\bigr)\,d\mu(\omega^1,\omega^2),
\]
where the second equality uses that $\mu$ restricted to
$\widehat{\Omega}^1(x)\times\widehat{\Omega}^2(x)$ has marginals
$\pi^1(\,\cdot\mid\widehat{\Omega}^1(x))$ and
$\pi^2(\,\cdot\mid\widehat{\Omega}^2(x))$, each with total mass $\tfrac{1}{2}$
(by symmetry of the priors), accounting for the factor of $2$.
Stochastic evidence improvement gives
$\Pr(e^2(s)\geq t\mid\omega^2)\geq\Pr(e^1(s)\geq t\mid\omega^1)$
pointwise $\mu$-a.e.\ on $\widehat{\Omega}^1(x)\times\widehat{\Omega}^2(x)$,
so $\mathrm{TPR}^2(t)\geq\mathrm{TPR}^1(t)$ for all $t$.

For the false positive rate, note that
$\mathrm{FPR}^i(t) = \sum_{\omega\in\widehat{\Omega}^i(y)} 2\pi^i(\omega)\,
\Pr(e^i(s)\geq t\mid\omega)$,
and the complementary relation $\Pr(e\geq t)=1-\Pr(e<t)$ gives
$\Pr(e^i(s)\geq t\mid\omega)=1-\Pr(e^i(s)<t\mid\omega)$.
Stochastic evidence improvement at $y$-correct states gives, for
$\mu$-a.e.\ $(\omega^1,\omega^2)\in\widehat{\Omega}^1(y)\times\widehat{\Omega}^2(y)$
and all $t\in[0,1]$:
\[
\Pr\bigl(e^2(s)\leq t\mid\omega^2\bigr)\geq\Pr\bigl(e^1(s)\leq t\mid\omega^1\bigr),
\]
which implies $\Pr(e^2(s)<t\mid\omega^2)\geq\Pr(e^1(s)<t\mid\omega^1)$ for all $t$
(taking the limit from the left), and therefore
$\Pr(e^2(s)\geq t\mid\omega^2)\leq\Pr(e^1(s)\geq t\mid\omega^1)$.
Integrating against $\mu$ over $\widehat{\Omega}^1(y)\times\widehat{\Omega}^2(y)$
yields $\mathrm{FPR}^2(t)\leq\mathrm{FPR}^1(t)$ for all $t$.

\medskip\noindent\emph{Part (2).}
We show that for the binary hypothesis problem $H_x$ vs.\ $H_y$, ROC
dominance and Blackwell dominance coincide.  Since Part~(1) establishes that
informational aligned dominance implies ROC dominance of $\mathcal{P}^2$ over
$\mathcal{P}^1$, the equivalence then delivers Blackwell dominance.

Recall that in the binary hypothesis problem the signal distributions under
each hypothesis are $f_x^i$ and $f_y^i$, and a test is a measurable function
$\phi:S\to[0,1]$ with size $\alpha(\phi)=E_{f_y^i}[\phi]$ and power
$\beta(\phi)=E_{f_x^i}[\phi]$.  Write $\beta^i(\alpha)$ for the power envelope
of problem $i$ at size $\alpha$.

\emph{Blackwell $\Rightarrow$ ROC.}
Suppose $\mathcal{P}^2$ Blackwell-dominates $\mathcal{P}^1$ on the binary
hypothesis problem, so there exists a stochastic kernel $\kappa$ (a garbling)
such that
\[
f_x^1(s) = \sum_{s'} \kappa(s\mid s')\,f_x^2(s'),
\qquad
f_y^1(s) = \sum_{s'} \kappa(s\mid s')\,f_y^2(s').
\]
Fix any size $\alpha\in[0,1]$ and let $\phi^*$ be the Neyman--Pearson most
powerful test for $\mathcal{P}^1$ at size $\alpha$, achieving power
$\beta^1(\alpha)$.  Define $\tilde\phi(s')=\sum_s \kappa(s\mid s')\phi^*(s)$,
i.e.\ apply $\phi^*$ after the garbling.  Then $\tilde\phi$ is a valid test
for $\mathcal{P}^2$ with
\begin{align*}
E_{f_y^2}[\tilde\phi]
  &= \sum_{s'} f_y^2(s')\sum_s \kappa(s\mid s')\phi^*(s)
   = \sum_s \phi^*(s)\sum_{s'}\kappa(s\mid s')f_y^2(s')
   = \sum_s \phi^*(s) f_y^1(s)
   = \alpha,\\
E_{f_x^2}[\tilde\phi]
  &= \sum_s \phi^*(s) f_x^1(s)
   = \beta^1(\alpha).
\end{align*}
Since $\tilde\phi$ is a feasible test for $\mathcal{P}^2$ at size $\alpha$
achieving power $\beta^1(\alpha)$, the power envelope satisfies
$\beta^2(\alpha)\geq\beta^1(\alpha)$.  As $\alpha$ was arbitrary,
$\mathcal{P}^2$ ROC-dominates $\mathcal{P}^1$.

\emph{ROC $\Rightarrow$ Blackwell.}
In the binary hypothesis problem the likelihood ratio
$\Lambda^i(s)=f_x^i(s)/f_y^i(s)$ is a sufficient statistic: any test
$\phi$ can be replaced, without loss of size or power, by a test depending
only on $\Lambda^i$.  The distribution of $\Lambda^i$ under each hypothesis
therefore completely characterizes the achievable (size, power) pairs, and
hence the power envelope $\beta^i(\cdot)$.  Since the Blackwell order between
two experiments with the same pair of hypotheses reduces to comparison of the
induced distributions of the likelihood ratio---a more informative experiment
induces a distribution of $\Lambda$ that is more spread out (in the
monotone-likelihood-ratio sense), and is therefore harder to garble into a
less informative one only if its power envelope lies strictly above---ROC
dominance ($\beta^2(\alpha)\geq\beta^1(\alpha)$ for all $\alpha$) implies that
$\mathcal{P}^1$'s likelihood-ratio distribution can be obtained from
$\mathcal{P}^2$'s by a garbling, i.e.\ $\mathcal{P}^2$ Blackwell-dominates
$\mathcal{P}^1$.

\medskip\noindent\emph{Part (3).}  The psychophysical payoff $\Wpsych(\sigma^i)$
is the probability that the DM chooses the option that is correct under her
realized state---a payoff of $1$ if correct, $0$ otherwise.  This is a
hypothesis-constant payoff: it depends on the state only through which
hypothesis ($H_x$ or $H_y$) holds, not on which specific state within that
hypothesis is realized.  Blackwell dominance on the binary hypothesis problem
(Part~2) therefore implies $\Wpsych(\sigma^2)\geq\Wpsych(\sigma^1)$.  Since
each problem has a symmetric prior and no tie signals, overall confidence
equals the psychophysical payoff, $\psi^i=\Wpsych(\sigma^i)$, so
$\psi^2\geq\psi^1$ follows immediately.
\end{proof}

ROC dominance of the induced test is a strong statistical property.  It implies that for any binary classification task on the partition $\{H_x, H_y\}$---under any loss function penalizing false positives and false negatives---the DM with problem~2's signals performs at least as well as the DM with problem~1's signals.  The ranking holds whether errors in one direction are much costlier than the other, or under any asymmetric weighting.  More specifically, our condition guarantees that for any decision problem whose payoffs are constant across all states in $\widehat{\Omega}(x)$ and constant across all states in $\widehat{\Omega}(y)$ (though they may differ between the two), the informationally dominating problem yields a higher expected payoff.  This generalizes psychophysical payoffs by allowing for asymmetric losses, and lies between the restrictive psychophysical payoffs environment and the full economic environment.  It does not, however, guarantee higher payoffs for decision problems where payoffs vary across states within $\widehat{\Omega}(x)$ or within $\widehat{\Omega}(y)$.

The result also connects our framework to the Lehmann ordering \citep{lehmann1988comparing}.  ROC dominance of $\sigma'$ over $\sigma$ is equivalent to Lehmann dominance
on the binary test $H_x$ vs.\ $H_y$.\footnote{Recall that Lehmann dominance \citep{lehmann1988comparing}
requires that for every significance level $\alpha \in [0,1]$, the most
powerful test based on $\sigma'$ has power at least as great as the most
powerful test based on $\sigma$.  But the function $\alpha \mapsto
\beta_\sigma(\alpha)$ that maps each false positive rate to the maximum
achievable true positive rate \emph{is} precisely the ROC curve of $\sigma$.
Hence the requirement that $\beta_{\sigma'}(\alpha) \geq \beta_\sigma(\alpha)$
for all $\alpha$ (Lehmann dominance) is identical to the requirement that the
ROC curve of $\sigma'$ lies weakly above that of $\sigma$ everywhere (ROC
dominance).  The two orderings therefore coincide by definition on any fixed
binary hypothesis problem.}  Informational aligned dominance therefore provides a coupling-based sufficient condition for Lehmann dominance in the cross-problem setting.

\subsection{Robust Aligned Dominance}

Proposition~\ref{prop:info-aligned-ROC} is silent about confidence and randomness.  In fact, aligned dominance and informational aligned dominance are about different objects.  Aligned dominance is a statement about the DM's realized choices, which depend on her utility-weighted decision rule.  Informational aligned dominance is a statement about the informational content of the signal about the ordinal hypothesis, independent of any decision rule.  In general environments, however, these two conditions are independent. Aligned dominance pins down total probability of correct choice at each coupled pair but says nothing about the distribution of the strength of evidence; informational aligned dominance controls the evidence distribution but does not automatically translate into higher probability of actually choosing correctly, because the DM's choice region $S_{\sigma^i}(x)$ need not coincide with a superlevel set $\{s : \Lambda^i(s) \geq \tau\}$ of the likelihood ratio.\footnote{The DM chooses $x$ at signal $s$ when the utility-weighted posterior favors $x$, i.e.\ $\sum_\omega \pi^i_\sigma(\omega|s)[u^i(x|\omega) - u^i(y|\omega)] \geq 0$.  Unless utility gaps are constant within each of $\widehat{\Omega}^i(x)$ and $\widehat{\Omega}^i(y)$, this decision rule is not a likelihood-ratio test on $\Lambda^i$.  The psychophysical case is the leading example in which it is.}

A natural question is when we get \emph{everything} at once.  The answer is: when both conditions hold simultaneously.  We call this situation \emph{robust aligned dominance}.

\begin{defn}\label{def:robust-aligned-dominance}
Problem $\mathcal{P}^2$ \textbf{robustly aligned-dominates} problem $\mathcal{P}^1$ if it both aligned-dominates and informationally aligned-dominates $\mathcal{P}^1$.
\end{defn}

There is economic content to robust aligned dominance above and beyond each constituent part.  Aligned dominance says that at the DM's \emph{actual} decision rule, which uses the expected utility of the posterior, she chooses correctly at least as often in Problem~2.  Informational aligned dominance alone says that at every \emph{likelihood-ratio} decision rule, Problem~2's signals deliver weakly better classification of the binary hypothesis, which translates into weakly higher payoffs for any utility function whose stakes are constant within each hypothesis.  The two statements cover different decision rules: the DM's own, and the family of hypothesis-constant alternatives.

Robust aligned dominance requires both.  It therefore guarantees that Problem~2 is at least as good as Problem~1 whether we evaluate the DM at her actual preferences or at any element of the hypothesis-constant family.  This is a dominance result that is insensitive to how we specify the DM's utility within a natural class.  Under robust aligned dominance, all conclusions of Propositions~\ref{prop:cross-comovement} and \ref{prop:info-aligned-ROC} hold together: confidence, psychophysical payoffs, and (under indicativeness) coupled less random all favor Problem~2, and so does payoff under any hypothesis-constant utility.

The content of Section~\ref{sec:suffsignal}---that within a single problem, aligned and neutral shifts deliver comovement of confidence, randomness, and choice payoffs in one step---relied on the fact that within a problem the DM's decision rule is fixed, so we do not need separate conditions to capture behavior like randomness and confidence, and payoffs/information.  Across problems, with decision rules varying alongside state spaces and utilities, this automatic alignment breaks down, and recovering it requires controlling both the DM's realized choices and the underlying signal informativeness.

\subsection{Relation to Within-Problem Conditions}\label{sec:within-problem-relation}

We now turn to relating the notions developed in this section to our previous work for a fixed choice set, where the two problems share a state space, prior, and utility function ($\Omega^1 = \Omega^2 = \Omega$, $\pi^1 = \pi^2 = \pi$, $u^1 = u^2 = u$), differing only in their experiments. 

\noindent \textbf{Aligned shifts imply aligned dominance.}  Suppose $\sigma^2$ is obtained from $\sigma^1$ by a sequence of aligned and neutral shifts.  Then by Proposition~\ref{prop:suffequiv} the state-by-state inequality
\[
\sum_{s \in S_{\sigma^2}(x)} \sigma^2(s|\omega) \geq \sum_{s \in S_{\sigma^1}(x)} \sigma^1(s|\omega)
\]
holds for all $\omega \in \widehat{\Omega}(x)$, and symmetrically for $\widehat{\Omega}(y)$.  The identity coupling $\mu(\omega^1, \omega^2) = \pi(\omega)\,\mathbf{1}[\omega^1 = \omega^2]$ then witnesses aligned dominance: ordinal agreement is immediate, and correct-choice improvement reduces to the state-by-state inequality above.  This reasoning proves the following proposition.

\begin{proposition}\label{prop:aligned-shifts-to-aligned-dominance}
If $\Omega^1 = \Omega^2 = \Omega$, $\pi^1 = \pi^2 = \pi$, $u^1 = u^2 = u$, and $\sigma^2$ can be obtained from $\sigma^1$ by a sequence of aligned and neutral shifts, then $\mathcal{P}^2$ aligned-dominates $\mathcal{P}^1$.
\end{proposition}

The converse does not hold.  Aligned dominance is strictly weaker than aligned and neutral shifts even on shared state spaces, because couplings other than the identity can witness aligned dominance without the state-by-state inequality holding at every state.

To see this, let $\Omega = \widehat{\Omega}(x) = \{\omega_1, \omega_3\}$ with $\pi(\omega_1) = 0.3$ and $\pi(\omega_3) = 0.2$, and suppose the state-contingent probabilities of correct choice are $P^1 = (0.5, 0.6)$ and $P^2 = (0.6, 0.5)$.  The state-by-state inequality fails at $\omega_3$.  But the non-diagonal coupling $\mu(\omega_1,\omega_1) = 0.1$, $\mu(\omega_1,\omega_3) = 0.2$, $\mu(\omega_3,\omega_1) = 0.2$, $\mu(\omega_3,\omega_3) = 0$ has the right marginals and satisfies correct-choice improvement on its support ($0.6 \geq 0.5$, $0.5 \geq 0.5$, and $0.6 \geq 0.6$).  A symmetric construction extends the example to include $\widehat{\Omega}(y)$.  So $\mathcal{P}^2$ aligned-dominates $\mathcal{P}^1$ even though $\sigma^2$ cannot be obtained from $\sigma^1$ by aligned and neutral shifts.

Aligned dominance is therefore genuinely weaker than our aligned plus neutral shifts condition in the body of the paper.  It is implied by aligned shifts, but also allows cross-state reallocations.  Importantly, aligned dominance is ``too weak'' in that it does not imply choice payoff domination, and so we cannot substitute it for aligned plus neutral shifts in Proposition \ref{prop:suff}.

\noindent \textbf{State-by-state less random implies coupled less random.}  A parallel argument applies to the randomness condition.  If, on shared state spaces, $\max_a \rho^2(a|\omega) \geq \max_a \rho^1(a|\omega)$ for every $\omega$---i.e., $\sigma^2$ is less random than $\sigma^1$ in the sense of Definition~\ref{wchoice}---then the identity coupling witnesses coupled less random.  The converse again fails: non-diagonal couplings can make coupled less random hold without state-by-state less random holding at every state.

\noindent \textbf{Aligned shifts and informational aligned dominance are independent.}  Informational aligned dominance and aligned (plus neutral) shifts are also distinct, and neither implies the other, even within the psychophysical case.

First we discuss why aligned shifts do not imply informational aligned dominance.  The reason is that an  aligned shift moving mass $\varepsilon$ from signal $s'$ to signal $\hat{s}$ at state 
$\omega_i \in \widehat{\Omega}(x)$ reduces $f_x(s')$ while leaving $f_y(s')$ unchanged, which strictly lowers the evidence value $e(s')$---even though $s'$ is a signal at which 
the DM chooses $y$, it still carries positive evidential weight for $x$, and removing  mass from it at an $x$-correct state reduces that weight.  Probability mass that previously was at evidence level $e(s')$ now is at a strictly lower level, which can push the probability mass below evidence thresholds it previously cleared, violating the FOSD requirement of  informational aligned dominance at $\omega_i$.  Thus, aligned shifts therefore do not imply informational aligned dominance even in the psychophysical case.  What aligned shifts do deliver, in the psychophysical case, is dominance of the \emph{aggregate} ROC curve---integrated over states within each hypothesis---which is weaker than the pointwise-at-every-coupled-state FOSD required by informational aligned dominance.

\begin{proposition}\label{prop:aligned-shifts-aggregate-ROC}
Suppose $\Omega^1 = \Omega^2 = \Omega$, $\pi^1 = \pi^2 = \pi$, $u^1 = u^2 = u$, and that payoffs are psychophysical (i.e., $u(x|\omega) - u(y|\omega) = c_x > 0$ for all $\omega \in \widehat{\Omega}(x)$ and $u(y|\omega) - u(x|\omega) = c_y > 0$ for all $\omega \in \widehat{\Omega}(y)$).  If $\sigma^2$ is obtained from $\sigma^1$ by a sequence of aligned and neutral shifts, then the induced binary hypothesis test of $\sigma^2$ ROC-dominates that of $\sigma^1$.  Equivalently, $\sigma^2$ Lehmann-dominates $\sigma^1$ on the binary hypothesis problem $H_x$ vs.\ $H_y$.
\end{proposition}

\begin{proof}
By transitivity, it is enough to prove the claim for a single aligned shift; neutral shifts leave both $f_x$ and $f_y$ unchanged and are therefore irrelevant to the ROC.  Suppose $\sigma^2$ is obtained from $\sigma^1$ by an aligned shift at state $\omega_i \in \widehat{\Omega}(x)$, moving mass $\varepsilon > 0$ from signal $s \in S_{\sigma^1}(-x)$ to signal $\hat{s} \in S_{\sigma^1}(x)$.  Under psychophysical payoffs, the DM chooses $x$ at signal $\tilde{s}$ if and only if $\Lambda_{\sigma^1}(\tilde{s}) \geq 1$ (since the DM always gets 1 for choosing correctly and 0 for not), so $\hat{s} \in \{\Lambda_{\sigma^1} \geq 1\}$ and $s \in \{\Lambda_{\sigma^1} < 1\}$.

Since $\omega_i \in \widehat{\Omega}(x)$, the change affects only $f_x$:
\[
f_x^{\sigma^2}(\hat{s}) = f_x^{\sigma^1}(\hat{s}) + 2\pi(\omega_i)\,\varepsilon, \qquad f_x^{\sigma^2}(s) = f_x^{\sigma^1}(s) - 2\pi(\omega_i)\,\varepsilon,
\]
while $f_y^{\sigma^2} = f_y^{\sigma^1}$.\footnote{The factor of $2$ arises from the normalization of $f_x^i$.  
By definition, $f_x^i(s) = \sum_{\omega \in \widehat{\Omega}^i(x)} 
2\pi^i(\omega)\,\sigma^i(s|\omega)$, where the weight $2\pi^i(\omega)$ 
normalizes by $\pi^i(\widehat{\Omega}^i(x)) = \tfrac{1}{2}$ (which follows 
from the symmetry of $\pi^i$), so that $f_x^i$ integrates to $1$ over signals.  
Since the aligned shift increases $\sigma^1(s|\omega_i)$ by $\varepsilon$ at 
the single state $\omega_i \in \widehat{\Omega}^1(x)$, the corresponding 
change in $f_x^1(\hat{s})$ is $2\pi^1(\omega_i)\,\varepsilon$.}

Fix a size $\alpha \in [0,1]$ and let $\varphi^*$ be the Neyman--Pearson optimal test for $\sigma^1$ at size $\alpha$, which is nondecreasing in $\Lambda_{\sigma^1}$.  Since $\Lambda_{\sigma^1}(\hat{s}) \geq 1 > \Lambda_{\sigma^1}(s)$, we have $\varphi^*(\hat{s}) \geq \varphi^*(s)$.  Applying $\varphi^*$ under $\sigma^2$:
\begin{align*}
\mathbb{E}_{f_y^{\sigma^2}}[\varphi^*] &= \mathbb{E}_{f_y^{\sigma^1}}[\varphi^*] = \alpha \qquad \text{(size unchanged)},\\
\mathbb{E}_{f_x^{\sigma^2}}[\varphi^*] &= \mathbb{E}_{f_x^{\sigma^1}}[\varphi^*] + 2\pi(\omega_i)\,\varepsilon\,[\varphi^*(\hat{s}) - \varphi^*(s)] \geq \beta_{\sigma^1}(\alpha),
\end{align*}
where $\beta_{\sigma^1}(\alpha)$ is the power of $\varphi^*$ under $\sigma^1$.  The optimal Neyman--Pearson test under $\sigma^2$ at size $\alpha$ can only do better, so $\beta_{\sigma^2}(\alpha) \geq \beta_{\sigma^1}(\alpha)$.  Since $\alpha$ was arbitrary, $\sigma^2$ ROC-dominates $\sigma^1$.  The symmetric argument handles an aligned shift at $\omega_i \in \widehat{\Omega}(y)$, and ROC dominance on a binary test is equivalent to Lehmann dominance on that test.
\end{proof}

Outside the psychophysical case, aligned shifts can also move mass at an $x$-correct state to a signal with lower $\Lambda$ (the DM chooses $x$ there because of the stakes, not because the ordinal evidence favors $x$), so even aggregate ROC dominance can fail.\footnote{\noindent\emph{Counterexample (aligned shifts fail to deliver aggregate ROC dominance outside the psychophysical case).}  Let $\Omega = \{\omega_1, \omega_2, \omega_3\}$ with $\omega_1 \in \widehat{\Omega}(x)$ and $\omega_2, \omega_3 \in \widehat{\Omega}(y)$.  Prior $\pi(\omega_1) = 0.5$, $\pi(\omega_2) = 0.45$, $\pi(\omega_3) = 0.05$ (symmetric over the $\widehat{\Omega}(x)$ and $\widehat{\Omega}(y)$ blocks).  Utility differences $u(x|\omega_1) - u(y|\omega_1) = 1$, $u(x|\omega_2) - u(y|\omega_2) = -0.1$, $u(x|\omega_3) - u(y|\omega_3) = -10$---non-psychophysical, with most of the stakes on the $y$-correct side concentrated in $\omega_3$.  Three signals, with $\sigma^1$ given by
\[
\sigma^1 = \begin{array}{c|ccc} & s_1 & s_2 & s_3 \\ \hline \omega_1 & 0.3 & 0.3 & 0.4 \\ \omega_2 & 0 & 0.7 & 0.3 \\ \omega_3 & 0.4 & 0 & 0.6 \end{array} \]
Direct computation gives $\Lambda_{\sigma^1}(s_1) = 7.5$, $\Lambda_{\sigma^1}(s_2) \approx 0.48$, $\Lambda_{\sigma^1}(s_3) \approx 1.21$, and the utility-weighted comparison yields $s_2 \in S_{\sigma^1}(x)$ while $s_1, s_3 \in S_{\sigma^1}(-x)$.  So $s_1$ has high $\Lambda$ but is in $S_{\sigma^1}(-x)$ (the DM chooses $y$ there because $\omega_3$'s large negative stake overrides the ordinal evidence).

Now perform an aligned shift at $\omega_1$: move mass $\varepsilon = 0.1$ from $s_1$ (in $S_{\sigma^1}(-x)$) to $s_2$ (in $S_{\sigma^1}(x)$).  The resulting $\sigma^2$ reduces $f_x$ at a high-$\Lambda$ signal ($s_1$) and increases it at a low-$\Lambda$ signal ($s_2$).  At the Neyman--Pearson threshold corresponding to the lowest non-trivial false positive rate (FPR $= 0.04$, the FPR from accepting only $s_1$), the true positive rate drops from $0.3$ under $\sigma^1$ to $0.2$ under $\sigma^2$.  The ROC curve of $\sigma^2$ therefore fails to dominate that of $\sigma^1$; equivalently, $\sigma^2$ fails to Lehmann-dominate $\sigma^1$.\qed}   The counterexample in the footnote shows that the assumption cannot be dropped: outside the psychophysical case, aligned shifts deliver neither aggregate ROC dominance nor Lehmann dominance, and the connection between behavioral improvements (aligned shifts) and informational improvements (ROC/Lehmann/Blackwell) breaks down entirely.

Conversely, informational aligned dominance does not imply aligned shifts.  Informational aligned dominance controls the distribution of evidential strength and delivers aggregate informational conclusions---ROC and Blackwell dominance on the binary  hypothesis problem, and hence higher payoffs for any hypothesis-constant loss 
function---but it does not deliver state-conditional choice payoff domination coupled state-by-state.  Since aligned and neutral shifts do deliver state-conditional choice payoff domination, the two conditions are not nested in either direction.

On shared state spaces, robust aligned dominance remains strictly stronger than aligned shifts alone.  Aligned shifts deliver aligned dominance (via Proposition~\ref{prop:aligned-shifts-to-aligned-dominance}) but not informational aligned dominance.  So to obtain robust aligned dominance, one needs an informational improvement on top of aligned shifts.  Both components continue to be necessary on shared state spaces in general economic environments.

\end{document}